\documentclass[submission,copyright,creativecommons]{eptcs}
\providecommand{\event}{QPL 2017} % Name of the event you are submitting to

\usepackage{etex}

 \usepackage{url} %% for url to arXiv references

\usepackage{amssymb}
\usepackage{amsmath}
\usepackage{amsthm}

\usepackage{mathtools}

\usepackage[all]{xy}
\usepackage{enumitem}
\usepackage{color}

\theoremstyle{plain}
\newtheorem{observation}{Remark}[section]
\newtheorem{lemma}[observation]{Lemma}  %%share counter with remark
\newtheorem{theorem}[observation]{Theorem}
\newtheorem{definition}[observation]{Definition}

\newtheorem{remark}[observation]{Remark}

\newtheorem{proposition}[observation]{Proposition} 
\newtheorem{corollary}[observation]{Corollary} 

\newcommand {\X} {\mathbb{X}}

\newcommand {\Y} {\mathbb{Y}}

\newcommand {\Z} {\mathbb{Z}}

\newcommand {\N} {\mathbb{N}}

\usepackage[svgnames]{xcolor}
\usepackage{tikz}

\usetikzlibrary{decorations.markings}
\usetikzlibrary{shapes}

\tikzset{every picture/.append style={scale=.6}, every node/.style={scale=0.6}}

\pgfdeclarelayer{edgelayer}
\pgfdeclarelayer{nodelayer}
\pgfsetlayers{edgelayer,nodelayer,main}
\tikzstyle{none}=[inner sep=-1pt]
\tikzstyle{dotted}=[inner sep=-1pt]
\tikzstyle{rectangle}=[shape=rectangle,draw,scale=3]
\tikzstyle{circle}=[shape=circle,draw]
\tikzstyle{circleSmall}=[shape=circle,scale=0.2,draw]
\tikzstyle{nonescaled}=[inner sep=-1pt,scale=1.5]
\tikzstyle{rectangleempty}=[shape=rectangle,minimum height=1cm,draw]
\tikzstyle{rectSmall}=[shape=rectangle,minimum height=0.7cm, minimum width=0.7cm, draw]
\tikzstyle{rectSmallDual}=[ shape=rectangle, minimum height=0.7cm, minimum width=0.7cm, draw]
\tikzstyle{rectMedium}=[shape=rectangle,draw,minimum height = 1.3 cm, minimum width=1 cm]
\tikzset{->-/.style={decoration={
  markings,
  mark=at position .5 with {\arrow{>}}},postaction={decorate}}}
  \tikzset{-<-/.style={decoration={
  markings,
  mark=at position .5 with {\arrow{<}}},postaction={decorate}}}
\tikzstyle{rtriangleb}=[regular polygon, fill=black, regular polygon sides=3, shape border rotate=-90, scale=0.5, draw]
\tikzstyle{ltriangleb}=[regular polygon, fill=black, regular polygon sides=3, shape border rotate=90, scale=0.5, draw]
\tikzstyle{rtrianglew}=[regular polygon,  regular polygon sides=3, shape border rotate=-90, scale=0.5, draw]
\tikzstyle{ltrianglew}=[regular polygon,  regular polygon sides=3, shape border rotate=90, scale=0.5, draw]

\newcommand {\swap} {\mathsf{swap}}

\renewcommand{\bar}[1]{\overline{#1}\hspace*{.05cm}}
\renewcommand{\hat}[1]{|#1\ra}

\tikzset{%
symbol/.style={%
draw=none,
every to/.append style={%
edge node={node [sloped, allow upside down, auto=false]{$#1$}}}
}
}

\usetikzlibrary{shapes.geometric}
\usetikzlibrary{patterns}
\usetikzlibrary{fit}
\usetikzlibrary{positioning}
\usetikzlibrary{calc}
\usetikzlibrary{arrows}
\usetikzlibrary{decorations.markings}

\pgfdeclarelayer{nodelayer}
\pgfdeclarelayer{edgelayer}
\pgfsetlayers{edgelayer,nodelayer,main}

\tikzset{simple/.style={}}
\tikzset{rn/.style={outer sep=-3.4pt}}
\tikzset{square/.style={draw,fill=white, regular polygon,regular polygon sides=4}}

\tikzset{dot/.style={thick, fill=black, circle, scale=0.4, inner sep = .05cm}}

\tikzset{oplus/.style={draw, scale=0.9,minimum height=.1cm,circle,append after command={
[shorten >=\pgflinewidth, shorten <=\pgflinewidth,]
(\tikzlastnode.north) edge (\tikzlastnode.south)
(\tikzlastnode.east) edge (\tikzlastnode.west)
}
}
}

\tikzset{fanin/.style={
draw,
shape border rotate=30,
regular polygon,
regular polygon sides=3,
fill=white,
inner sep = .1cm
}
}

\tikzset{fanout/.style={
draw,
shape border rotate=-30,
regular polygon,
regular polygon sides=3,
fill=white,
inner sep = .1cm
}
}

\tikzset{onein/.style={
draw,
shape border rotate=30,
regular polygon,
regular polygon sides=3,
fill=black,
inner sep = .04cm,
scale=1.2
}
}

\tikzset{oneout/.style={
draw,
shape border rotate=-30,
regular polygon,
regular polygon sides=3,
fill=black,
inner sep = .04cm,
scale=1.2
}
}

\tikzset{zeroin/.style={
draw,
shape border rotate=30,
regular polygon,
regular polygon sides=3,
fill=white,
inner sep = .04cm,
scale=1.2
}
}

\tikzset{zeroout/.style={
draw,
shape border rotate=-30,
regular polygon,
regular polygon sides=3,
fill=white,
inner sep = .04cm,
scale=1.2
}
}

\tikzstyle{uptriangle}=[regular polygon, regular polygon sides=3, scale=1, draw]
\tikzstyle{downtriangle}=[regular polygon, regular polygon sides=3, shape border rotate=180, scale=1, draw]

\tikzset{wires/.style={}}

\tikzset{box/.style={inner sep=0pt, thick, draw=black, text height=1.5ex, text depth=.25ex, text centered, minimum height=3em, anchor=center}}

\newcommand{\dom}{\mathsf{dom}}
\newcommand{\cod}{\mathsf{cod }}
\newcommand{\Par}{\mathsf{Par}}

\newcommand{\cnot}{\mathsf{cnot}}
\newcommand{\CNOT}{\mathsf{CNOT}}

\newcommand{\cnv}{^{\circ}}
\newcommand{\op}{^{\text{op}}}

\newcommand{\la}{\langle}
\newcommand{\ra}{\rangle}
\newcommand{\onein}{| 1 \ra}
\newcommand{\oneout}{\la 1 |}
\newcommand{\zeroin}{| 0 \ra}
\newcommand{\zeroout}{\la 0|}
\newcommand{\Tor}{\mathsf{Tor}}
\newcommand{\CTor}{\mathsf{CTor}}
\newcommand{\ParIso}{\mathsf{ParIso}}
\newcommand{\Sets}{\mathsf{Set}} % allows tables to share a page
\newcommand{\Total}{\mathsf{Total}}
\renewcommand{\bf}[1]{\textbf{#1}}
\newcommand{\pullbackcorner}[1][u]{\save*!/#1+1.2pc/#1:(1,-1)@^{|-}\restore}

\allowdisplaybreaks

\makeatletter
\newcommand*{\doublerightarrow}[2]{\mathrel{
  \settowidth{\@tempdima}{$\scriptstyle#1$}
  \settowidth{\@tempdimb}{$\scriptstyle#2$}
  \ifdim\@tempdimb>\@tempdima \@tempdima=\@tempdimb\fi
  \mathop{\vcenter{
    \offinterlineskip\ialign{\hbox to\dimexpr\@tempdima+1em{##}\cr
    \rightarrowfill\cr\noalign{\kern.5ex}
    \rightarrowfill\cr}}}\limits^{\!#1}_{\!#2}}}
\newcommand*{\triplerightarrow}[1]{\mathrel{
  \settowidth{\@tempdima}{$\scriptstyle#1$}
  \mathop{\vcenter{
    \offinterlineskip\ialign{\hbox to\dimexpr\@tempdima+1em{##}\cr
    \rightarrowfill\cr\noalign{\kern.5ex}
    \rightarrowfill\cr\noalign{\kern.5ex}
    \rightarrowfill\cr}}}\limits^{\!#1}}}
\makeatother
\usepackage{array}
\newcolumntype{P}[1]{>{\centering\arraybackslash}p{#1}}
\newcolumntype{M}[1]{>{\centering\arraybackslash}m{#1}}
\title{The Category \texorpdfstring{$\CNOT$}{CNOT}}
\author{Robin Cockett  \thanks{partially supported by NSERC} \quad \quad Cole Comfort \quad \quad Priyaa Srinivasan \thanks{partially supported by the NSERC grant of Dr. Barry C. Sanders}
\institute{Department of Computer Science \\ and Institute for Quantum Science and Technology \\
University of Calgary\\
Alberta, Canada}
\email{robin@ucalgary.ca}
}

\def\titlerunning{The Category $\CNOT$}
\def\authorrunning{Robin Cockett, Cole Comfort \& Priyaa Srinivasan}

\begin{document}
\maketitle

\begin{abstract}
We exhibit a complete set of identities for $\CNOT$, the symmetric monoidal category generated by the controlled-not gate, the $\swap$ gate, and the computational ancill\ae. We prove that $\CNOT$ is a discrete inverse category.  Moreover, we prove that $\CNOT$ is equivalent to the category of partial isomorphisms of finitely-generated non-empty commutative torsors of characteristic $2$. Equivalently this is the category of affine partial isomorphisms between finite-dimensional $\mathbb{Z}_2$ vector spaces.
\end{abstract}

%%%%%%%%%%%%% Section 1 %%%%%%%%%%%%%%%%%%%%%

\section{Introduction}

In this paper, we model the behaviour of circuits comprised of $\cnot$ gates and the four computational ancill\ae ---which restrict certain inputs and outputs to be either 0 or 1.  We model these circuits as maps in the symmetric monoidal category $\CNOT$ given by finite generators and relations.  Although the $\cnot$ gate is unitary, the ancill\ae\ are not.  This is because ancill\ae\ model state preparation and measurement which are irreversible operations.  Ancilli\ae\ are commonly used in quantum error correction codes \cite{chiaverini2004realization, gottesman1997stabilizer}; moreover, the proof that the Toffolli gate is universal uses ancill\ae\ \cite{1506.03777}.  Although unitary transformations are an active area of research \cite{1408.5728}, and there is a finite, faithful set of identities for circuits composed of $\cnot$ gates \cite{1701.00140}, the structure of circuits composed of $\cnot$ gates \emph{and} ancill\ae\ is not yet studied. 

This research extends the work of Lafont, who classified several similar categories \cite{Lafont}.  We prove that $\CNOT$ is equivalent to a concrete category of torsors and partial maps -- in other words the category of affine partial isomorphisms between finite dimensional $\mathbb{Z}_2$ vector spaces.  We have yet to prove that there is a {\em faithful\/} embedding of $\CNOT$ into the category of dagger Frobenius Algebras in finite-dimensional Hilbert spaces and completely positive maps, $\mathsf{CPM(FHilb)}$ (see \cite{selinger2007dagger}).

%%%%%%%%%%%%%%%%%%%%%%%%%%%%
\section{Defining the Category \texorpdfstring{$\CNOT$}{CNOT}}
%%%%%%%%%%%%%%%%%%%%%%%%%%%%

$\CNOT$ is a symmetric monoidal category presented by generating maps and identities.  We use string diagrams to express the maps in $\CNOT$.  For the $\cnot$-gate and the 
upside-down $\cnot$-gate, we use the following notation respectively:
\begin{flalign*}
\begin{tikzpicture}[baseline={([yshift=-.5ex]current bounding box.center)}]
\begin{pgfonlayer}{nodelayer}
\node [style=rn] (0) at (0, .5) {};
\node [style=dot] (1) at (.5, .5) {};
\node [style=rn] (2) at (1, .5) {};
\node [style=rn] (10) at (0, .) {};
\node [style=oplus] (11) at (.5, 0) {};
\node [style=rn] (12) at (1, .0) {};
\end{pgfonlayer}
\begin{pgfonlayer}{edgelayer}
\draw [style=simple] (0) to (2);
\draw [style=simple] (10) to (12);
\draw [style=simple] (1) to (11);
\end{pgfonlayer}
\end{tikzpicture}
\hspace*{.5cm}
\hspace*{.5cm}
\begin{tikzpicture}[baseline={([yshift=-.5ex]current bounding box.center)}]
\begin{pgfonlayer}{nodelayer}
\node [style=rn] (0) at (0, .5) {};
\node [style=oplus] (1) at (.5, .5) {};
\node [style=rn] (2) at (1, .5) {};
\node [style=rn] (10) at (0, .) {};
\node [style=dot] (11) at (.5, 0) {};
\node [style=rn] (12) at (1, .0) {};
\end{pgfonlayer}
\begin{pgfonlayer}{edgelayer}
\draw [style=simple] (0) to (2);
\draw [style=simple] (10) to (12);
\draw [style=simple] (1) to (11);
\end{pgfonlayer}
\end{tikzpicture}
:=
\begin{tikzpicture}[baseline={([yshift=-.5ex]current bounding box.center)}]
\begin{pgfonlayer}{nodelayer}
\node [style=rn] (0) at (0, 0) {};
\node [style=rn] (1) at (.25, 0) {};
\node [style=oplus] (2) at (1, 0) {};
\node [style=rn] (3) at (1.75, 0) {};
\node [style=rn] (4) at (2, 0) {};
\node [style=rn] (10) at (0, .5) {};
\node [style=rn] (11) at (.25, .5) {};
\node [style=dot] (12) at (1, .5) {};
\node [style=rn] (13) at (1.75, .5) {};
\node [style=rn] (14) at (2, .5) {};
\end{pgfonlayer}
\begin{pgfonlayer}{edgelayer}
\draw plot [smooth, tension=1] coordinates { (0) (1) (12) (3) (4)};
\draw plot [smooth, tension=1] coordinates { (10) (11) (2) (13) (14)};
\draw plot [smooth, tension=0.3] coordinates { (2) (12)};
\end{pgfonlayer}
\end{tikzpicture}\\
\end{flalign*}

We call the wire of the $\cnot$ gate with the dot the \textbf{control bit} and the other wire with the $\oplus$ the \textbf{operating bit}.
We graphically denote the input and output ancillary bits for 1 by 
\tikz{
\node [style=onein] (0) at (0, 0) {};
\node [style=rn] (1) at (1, 0) {};
\draw [style=simple] (0) to (1);
}
 and 
\tikz{
\node [style=rn] (0) at (0, 0) {};
\node [style=oneout] (1) at (1, 0) {};
\draw [style=simple] (0) to (1);
}
.  Algebraically we denote the input ancillary bit for 1 by $\onein$ and the output ancillary bit for 1 by $\oneout$.  

$\CNOT$ is the symmetric monoidal category generated by $\cnot$, and the 1 ancill\ae\ satisfying the identities \ref{CNTi}-\ref{CNTix}.  Two maps are the same in $\CNOT$ if and only if they can be transformed to another using the following identities:
\begin{enumerate}[label=\textit{(CNT.\arabic*)},leftmargin=*]
\item \leavevmode\vadjust{\vspace{-\baselineskip}}\newline
\vspace*{-.7cm}
\label{CNTi}
%\centering
\setlength{\jot}{10pt}
\begin{flalign*}
\begin{tikzpicture}[baseline={([yshift=-.5ex]current bounding box.center)}]
\begin{pgfonlayer}{nodelayer}
\node [style=rn] (0) at (0, 0) {};
\node [style=rn] (1) at (0, .5) {};
\node [style=oplus] (2) at (.5, 0) {};
\node [style=dot] (3) at (.5, .5) {};
\node [style=dot] (4) at (1, 0) {};
\node [style=oplus] (5) at (1, .5) {};
\node [style=oplus] (6) at (1.5, 0) {};
\node [style=dot] (7) at (1.5, .5) {};
\node [style=rn] (8) at (2, 0) {};
\node [style=rn] (9) at (2, .5) {};
\end{pgfonlayer}
\begin{pgfonlayer}{edgelayer}
\draw [style=simple] (0) to (8);
\draw [style=simple] (1) to (9);
\draw [style=simple] (2) to (3);
\draw [style=simple] (4) to (5);
\draw [style=simple] (6) to (7);
\end{pgfonlayer}
\end{tikzpicture}
=
\begin{tikzpicture}[baseline={([yshift=-.5ex]current bounding box.center)}]
\begin{pgfonlayer}{nodelayer}
\node [style=rn] (0) at (0, 0) {};
\node [style=rn] (1) at (0, 1) {};
\node [style=rn] (2) at (1, 0.5) {};
\node [style=rn] (3) at (2, 0) {};
\node [style=rn] (4) at (2, 1) {};
\end{pgfonlayer}
\begin{pgfonlayer}{edgelayer}
\draw [style=simple, bend right] (0) to (2);
\draw [style=simple, bend left] (1) to (2);
\draw [style=simple, bend right] (2) to (3);
\draw [style=simple, bend left] (2) to (4);
\end{pgfonlayer}
\end{tikzpicture}
\end{flalign*}

\item \leavevmode\vadjust{\vspace{-\baselineskip}}\newline
\vspace*{-.9cm}
\label{CNTii}
%\centering
\setlength{\jot}{10pt}

\begin{flalign*}
\begin{tikzpicture}[baseline={([yshift=-.5ex]current bounding box.center)}]
\begin{pgfonlayer}{nodelayer}
\node [style=rn] (0) at (0, 0) {};
\node [style=rn] (1) at (0, .5) {};
\node [style=oplus] (2) at (.5, 0) {};
\node [style=dot] (3) at (.5, .5) {};
\node [style=oplus] (6) at (1, 0) {};
\node [style=dot] (7) at (1, .5) {};
\node [style=rn] (8) at (1.5, 0) {};
\node [style=rn] (9) at (1.5, .5) {};
\end{pgfonlayer}
\begin{pgfonlayer}{edgelayer}
\draw [style=simple] (0) to (8);
\draw [style=simple] (1) to (9);
\draw [style=simple] (2) to (3);
\draw [style=simple] (6) to (7);
\end{pgfonlayer}
\end{tikzpicture}
=
\begin{tikzpicture}[baseline={([yshift=-.5ex]current bounding box.center)}]
\begin{pgfonlayer}{nodelayer}
\node [style=rn] (0) at (0, 0) {};
\node [style=rn] (1) at (0, .5) {};
\node [style=rn] (3) at (1.5, 0) {};
\node [style=rn] (4) at (1.5, .5) {};
\end{pgfonlayer}
\begin{pgfonlayer}{edgelayer}
\draw [style=simple] (0) to (3);
\draw [style=simple] (1) to (4);
\end{pgfonlayer}
\end{tikzpicture}
\end{flalign*}

\item \leavevmode\vadjust{\vspace{-\baselineskip}}\newline
\vspace*{-.9cm}
\label{CNTiii}
%\centering
\setlength{\jot}{10pt}

\begin{flalign*}
\begin{tikzpicture}[baseline={([yshift=-.5ex]current bounding box.center)}]
\begin{pgfonlayer}{nodelayer}
\node [style=rn] (0) at (0, 1) {};
\node [style=rn] (1) at (0, .5) {};
\node [style=rn] (2) at (0, 0) {};
\node [style=oplus] (3) at (.75, 1) {};
\node [style=dot] (4) at (.75, .5) {};
\node [style=dot] (5) at (1.25, .5) {};
\node [style=oplus] (6) at (1.25, 0) {};
\node [style=rn] (7) at (2, 1) {};
\node [style=rn] (8) at (2, .5) {};
\node [style=rn] (9) at (2, 0) {};
\end{pgfonlayer}
\begin{pgfonlayer}{edgelayer}
\draw [style=simple] (0) to (7);
\draw [style=simple] (1) to (8);
\draw [style=simple] (2) to (9);
\draw [style=simple] (3) to (4);
\draw [style=simple] (5) to (6);
\end{pgfonlayer}
\end{tikzpicture}
=
\begin{tikzpicture}[baseline={([yshift=-.5ex]current bounding box.center)}]
\begin{pgfonlayer}{nodelayer}
\node [style=rn] (0) at (0, 1) {};
\node [style=rn] (1) at (0, .5) {};
\node [style=rn] (2) at (0, 0) {};
\node [style=oplus] (3) at (1.25, 1) {};
\node [style=dot] (4) at (1.25, .5) {};
\node [style=dot] (5) at (.75, .5) {};
\node [style=oplus] (6) at (.75, 0) {};
\node [style=rn] (7) at (2, 1) {};
\node [style=rn] (8) at (2, .5) {};
\node [style=rn] (9) at (2, 0) {};
\end{pgfonlayer}
\begin{pgfonlayer}{edgelayer}
\draw [style=simple] (0) to (7);
\draw [style=simple] (1) to (8);
\draw [style=simple] (2) to (9);
\draw [style=simple] (3) to (4);
\draw [style=simple] (5) to (6);
\end{pgfonlayer}
\end{tikzpicture}
\end{flalign*}

\item \leavevmode\vadjust{\vspace{-\baselineskip}}\newline
\vspace*{-.9cm}
\label{CNTiv}
%\centering
\setlength{\jot}{10pt}

\begin{flalign*}
\begin{tikzpicture}[baseline={([yshift=-.5ex]current bounding box.center)}]
\begin{pgfonlayer}{nodelayer}
\node [style=onein] (0) at (0, .5) {};
\node [style=rn] (1) at (0, 0) {};
\node [style=dot] (2) at (.5, .5) {};
\node [style=oplus] (3) at (.5, 0) {};
\node [style=rn] (4) at (1, .5) {};
\node [style=rn] (5) at (1, 0) {};
\end{pgfonlayer}
\begin{pgfonlayer}{edgelayer}
\draw [style=simple] (0) to (4);
\draw [style=simple] (1) to (5);
\draw [style=simple] (2) to (3);
\end{pgfonlayer}
\end{tikzpicture}
=
\begin{tikzpicture}[baseline={([yshift=-.5ex]current bounding box.center)}]
\begin{pgfonlayer}{nodelayer}
\node [style=onein] (0) at (0, .5) {};
\node [style=rn] (1) at (0, 0) {};
\node [style=dot] (2) at (.5, .5) {};
\node [style=oplus] (3) at (.5, 0) {};
\node [style=oneout] (4) at (1, .5) {};
\node [style=rn] (5) at (2, 0) {};
\node [style=onein] (6) at (1.5, .5) {};
\node [style=rn] (7) at (2, 0.5) {};
\end{pgfonlayer}
\begin{pgfonlayer}{edgelayer}
\draw [style=simple] (0) to (4);
\draw [style=simple] (1) to (5);
\draw [style=simple] (2) to (3);
\draw [style=simple] (6) to (7);
\end{pgfonlayer}
\end{tikzpicture}
\hspace*{.5cm}
\text{and}
\hspace*{.5cm}
\begin{tikzpicture}[baseline={([yshift=-.5ex]current bounding box.center)}]
\begin{pgfonlayer}{nodelayer}
\node [style=rn] (0) at (0, .5) {};
\node [style=rn] (1) at (0, 0) {};
\node [style=dot] (2) at (.5, .5) {};
\node [style=oplus] (3) at (.5, 0) {};
\node [style=oneout] (4) at (1, .5) {};
\node [style=rn] (5) at (1, 0) {};
\end{pgfonlayer}
\begin{pgfonlayer}{edgelayer}
\draw [style=simple] (0) to (4);
\draw [style=simple] (1) to (5);
\draw [style=simple] (2) to (3);
\end{pgfonlayer}
\end{tikzpicture}
=
\begin{tikzpicture}[baseline={([yshift=-.5ex]current bounding box.center)}]
\begin{pgfonlayer}{nodelayer}
\node [style=oneout] (0) at (2, .5) {};
\node [style=rn] (1) at (2, 0) {};
\node [style=dot] (2) at (1.5, .5) {};
\node [style=oplus] (3) at (1.5, 0) {};
\node [style=onein] (4) at (1, .5) {};
\node [style=rn] (5) at (0, 0) {};
\node [style=oneout] (6) at (.5, .5) {};
\node [style=rn] (7) at (0, 0.5) {};
\end{pgfonlayer}
\begin{pgfonlayer}{edgelayer}
\draw [style=simple] (0) to (4);
\draw [style=simple] (1) to (5);
\draw [style=simple] (2) to (3);
\draw [style=simple] (6) to (7);
\end{pgfonlayer}
\end{tikzpicture}
\end{flalign*}

\item \leavevmode\vadjust{\vspace{-\baselineskip}}\newline
\vspace*{-.9cm}
\label{CNTv}
%\centering
\setlength{\jot}{10pt}

\begin{flalign*}
\begin{tikzpicture}[baseline={([yshift=-.5ex]current bounding box.center)}]
\begin{pgfonlayer}{nodelayer}
\node [style=rn] (0) at (0, 1) {};
\node [style=rn] (1) at (0, .5) {};
\node [style=rn] (2) at (0, 0) {};
\node [style=dot] (3) at (.75, 1) {};
\node [style=oplus] (4) at (.75, .5) {};
\node [style=oplus] (5) at (1.25, .5) {};
\node [style=dot] (6) at (1.25, 0) {};
\node [style=rn] (7) at (2, 1) {};
\node [style=rn] (8) at (2, .5) {};
\node [style=rn] (9) at (2, 0) {};
\end{pgfonlayer}
\begin{pgfonlayer}{edgelayer}
\draw [style=simple] (0) to (7);
\draw [style=simple] (1) to (8);
\draw [style=simple] (2) to (9);
\draw [style=simple] (3) to (4);
\draw [style=simple] (5) to (6);
\end{pgfonlayer}
\end{tikzpicture}
=
\begin{tikzpicture}[baseline={([yshift=-.5ex]current bounding box.center)}]
\begin{pgfonlayer}{nodelayer}
\node [style=rn] (0) at (0, 1) {};
\node [style=rn] (1) at (0, .5) {};
\node [style=rn] (2) at (0, 0) {};
\node [style=dot] (3) at (1.25, 1) {};
\node [style=oplus] (4) at (1.25, .5) {};
\node [style=oplus] (5) at (.75, .5) {};
\node [style=dot] (6) at (.75, 0) {};
\node [style=rn] (7) at (2, 1) {};
\node [style=rn] (8) at (2, .5) {};
\node [style=rn] (9) at (2, 0) {};
\end{pgfonlayer}
\begin{pgfonlayer}{edgelayer}
\draw [style=simple] (0) to (7);
\draw [style=simple] (1) to (8);
\draw [style=simple] (2) to (9);
\draw [style=simple] (3) to (4);
\draw [style=simple] (5) to (6);
\end{pgfonlayer}
\end{tikzpicture}
\end{flalign*}

\item \leavevmode\vadjust{\vspace{-\baselineskip}}\newline
\vspace*{-.9cm}
\label{CNTvi}
%\centering
\setlength{\jot}{10pt}

\begin{flalign*}
\begin{tikzpicture}[baseline={([yshift=-.5ex]current bounding box.center)}]
\begin{pgfonlayer}{nodelayer}
\node [style=onein] (0) at (0, 0) {};
\node [style=oneout] (1) at (1, 0) {};
\end{pgfonlayer}
\begin{pgfonlayer}{edgelayer}
\draw [style=simple] (0) to (1);
\end{pgfonlayer}
\end{tikzpicture}
=
\end{flalign*}

\item \leavevmode\vadjust{\vspace{-\baselineskip}}\newline
\vspace*{-.9cm}
\label{CNTvii}
%\centering
\setlength{\jot}{10pt}

\begin{flalign*}
\begin{tikzpicture}[baseline={([yshift=-.5ex]current bounding box.center)}]
\begin{pgfonlayer}{nodelayer}
\node [style=onein] (0) at (0, 1) {};
\node [style=onein] (1) at (0, .5) {};
\node [style=rn] (2) at (0, 0) {};
\node [style=dot] (3) at (.5, 1) {};
\node [style=oplus] (4) at (.5, .5) {};
\node [style=dot] (5) at (1, .5) {};
\node [style=oplus] (6) at (1, 0) {};
\node [style=oneout] (7) at (1, 1) {};
\node [style=rn] (8) at (1.5, .5) {};
\node [style=rn] (9) at (1.5, 0) {};
\end{pgfonlayer}
\begin{pgfonlayer}{edgelayer}
\draw [style=simple] (0) to (7);
\draw [style=simple] (1) to (8);
\draw [style=simple] (2) to (9);
\draw [style=simple] (3) to (4);
\draw [style=simple] (5) to (6);
\end{pgfonlayer}
\end{tikzpicture}
=
\begin{tikzpicture}[baseline={([yshift=-.5ex]current bounding box.center)}]
\begin{pgfonlayer}{nodelayer}
\node [style=onein] (0) at (0, 1) {};
\node [style=onein] (1) at (0, .5) {};
\node [style=rn] (2) at (0, 0) {};
\node [style=dot] (3) at (.5, 1) {};
\node [style=oplus] (4) at (.5, .5) {};
\node [style=oneout] (7) at (1, 1) {};
\node [style=rn] (8) at (1.5, .5) {};
\node [style=rn] (9) at (1.5, 0) {};
\end{pgfonlayer}
\begin{pgfonlayer}{edgelayer}
\draw [style=simple] (0) to (7);
\draw [style=simple] (1) to (8);
\draw [style=simple] (2) to (9);
\draw [style=simple] (3) to (4);
\end{pgfonlayer}
\end{tikzpicture}
\hspace*{.5cm}
\text{and}
\hspace*{.5cm}
\begin{tikzpicture}[baseline={([yshift=-.5ex]current bounding box.center)}]
\begin{pgfonlayer}{nodelayer}
\node [style=oneout] (0) at (1.5, 1) {};
\node [style=oneout] (1) at (1.5, .5) {};
\node [style=rn] (2) at (1.5, 0) {};
\node [style=dot] (3) at (1, 1) {};
\node [style=oplus] (4) at (1, .5) {};
\node [style=dot] (5) at (.5, .5) {};
\node [style=oplus] (6) at (.5, 0) {};
\node [style=onein] (7) at (.5, 1) {};
\node [style=rn] (8) at (0, .5) {};
\node [style=rn] (9) at (0, 0) {};
\end{pgfonlayer}
\begin{pgfonlayer}{edgelayer}
\draw [style=simple] (0) to (7);
\draw [style=simple] (1) to (8);
\draw [style=simple] (2) to (9);
\draw [style=simple] (3) to (4);
\draw [style=simple] (5) to (6);
\end{pgfonlayer}
\end{tikzpicture}
=
\begin{tikzpicture}[baseline={([yshift=-.5ex]current bounding box.center)}]
\begin{pgfonlayer}{nodelayer}
\node [style=oneout] (0) at (1.5, 1) {};
\node [style=oneout] (1) at (1.5, .5) {};
\node [style=rn] (2) at (1.5, 0) {};
\node [style=dot] (3) at (1, 1) {};
\node [style=oplus] (4) at (1, .5) {};
\node [style=onein] (7) at (.5, 1) {};
\node [style=rn] (8) at (0, .5) {};
\node [style=rn] (9) at (0, 0) {};
\end{pgfonlayer}
\begin{pgfonlayer}{edgelayer}
\draw [style=simple] (0) to (7);
\draw [style=simple] (1) to (8);
\draw [style=simple] (2) to (9);
\draw [style=simple] (3) to (4);
\end{pgfonlayer}
\end{tikzpicture}
\end{flalign*}

\item \leavevmode\vadjust{\vspace{-\baselineskip}}\newline
\vspace*{-.9cm}
\label{CNTviii}
%\centering
\setlength{\jot}{10pt}

\begin{flalign*}
\begin{tikzpicture}[baseline={([yshift=-.5ex]current bounding box.center)}]
\begin{pgfonlayer}{nodelayer}
\node [style=rn] (0) at (0, 1) {};
\node [style=rn] (1) at (0, .5) {};
\node [style=rn] (2) at (0, 0) {};
\node [style=dot] (3) at (.5, 1) {};
\node [style=oplus] (4) at (.5, .5) {};
\node [style=dot] (5) at (1, .5) {};
\node [style=oplus] (6) at (1, 0) {};
\node [style=dot] (7) at (1.5, 1) {};
\node [style=oplus] (8) at (1.5, .5) {};
\node [style=rn] (9) at (2, 1) {};
\node [style=rn] (10) at (2, .5) {};
\node [style=rn] (11) at (2, 0) {};
\end{pgfonlayer}
\begin{pgfonlayer}{edgelayer}
\draw [style=simple] (0) to (9);
\draw [style=simple] (1) to (10);
\draw [style=simple] (2) to (11);
\draw [style=simple] (3) to (4);
\draw [style=simple] (5) to (6);
\draw [style=simple] (7) to (8);
\end{pgfonlayer}
\end{tikzpicture}
=
\begin{tikzpicture}[baseline={([yshift=-.5ex]current bounding box.center)}]
\begin{pgfonlayer}{nodelayer}
\node [style=rn] (0) at (0, 1) {};
\node [style=rn] (1) at (0, .5) {};
\node [style=rn] (2) at (0, 0) {};
\node [style=dot] (5) at (.5, .5) {};
\node [style=oplus] (6) at (.5, 0) {};
\node [style=dot] (7) at (1, 1) {};
\node [style=oplus] (8) at (1, 0) {};
\node [style=rn] (9) at (1.5, 1) {};
\node [style=rn] (10) at (1.5, .5) {};
\node [style=rn] (11) at (1.5, 0) {};
\end{pgfonlayer}
\begin{pgfonlayer}{edgelayer}
\draw [style=simple] (0) to (9);
\draw [style=simple] (1) to (10);
\draw [style=simple] (2) to (11);
\draw [style=simple] (5) to (6);
\draw [style=simple] (7) to (8);
\end{pgfonlayer}
\end{tikzpicture}
\end{flalign*}

\item \leavevmode\vadjust{\vspace{-\baselineskip}}\newline
\vspace*{-.9cm}
\label{CNTix}
%\centering
\setlength{\jot}{10pt}

\begin{flalign*}
\begin{tikzpicture}[baseline={([yshift=-.5ex]current bounding box.center)}]
\begin{pgfonlayer}{nodelayer}
\node [style=onein] (0) at (0, 1) {};
\node [style=onein] (1) at (0, .5) {};
\node [style=rn] (2) at (0, 0) {};
\node [style=dot] (3) at (.5, 1) {};
\node [style=oplus] (4) at (.5, .5) {};
\node [style=oneout] (7) at (1, 1) {};
\node [style=oneout] (8) at (1, .5) {};
\node [style=rn] (9) at (1, 0) {};
\end{pgfonlayer}
\begin{pgfonlayer}{edgelayer}
\draw [style=simple] (0) to (7);
\draw [style=simple] (1) to (8);
\draw [style=simple] (2) to (9);
\draw [style=simple] (3) to (4);
\end{pgfonlayer}
\end{tikzpicture}
=
\begin{tikzpicture}[baseline={([yshift=-.5ex]current bounding box.center)}]
\begin{pgfonlayer}{nodelayer}
\node [style=onein] (0) at (0, 1) {};
\node [style=onein] (1) at (0, .5) {};
\node [style=rn] (2) at (0, 0) {};
\node [style=dot] (3) at (1, 1) {};
\node [style=oplus] (4) at (1, .5) {};
\node [style=oneout] (7) at (2, 1) {};
\node [style=oneout] (8) at (2, .5) {};
\node [style=rn] (9) at (2, 0) {};
\node [style=oneout] (10) at (0.75, 0) {};
\node [style=onein] (11) at (1.25, 0) {};
\end{pgfonlayer}
\begin{pgfonlayer}{edgelayer}
\draw [style=simple] (0) to (7);
\draw [style=simple] (1) to (8);
\draw [style=simple] (2) to (10);
\draw [style=simple] (11) to (9);
\draw [style=simple] (3) to (4);
\end{pgfonlayer}
\end{tikzpicture}
\end{flalign*}
\end{enumerate}

Note that all of the identities \ref{CNTi}-\ref{CNTix}\ are horizontally symmetric.  This symmetry expresses a functorial involution which will be useful later.

While the first eight identities are quite familiar, \ref{CNTix}\ may be unexpected: it is reminiscent of the absorbing scalar Axiom {\em (ZO)} in the ZX calculus \cite{ZX}.

Note that $\CNOT$ only has 3 generating gates as we can construct the 0-ancillary bits from $\cnot$ and the 1 ancillary bits.
The following gates are respectively the input and output 0 ancillary bits:
%\begin{figure}[H]
\begin{flalign*}
\begin{tikzpicture}[baseline={([yshift=-.5ex]current bounding box.center)}]
\begin{pgfonlayer}{nodelayer}
\node [style=zeroin] (0) at (0, 0) {};
\node [style=rn] (1) at (.5, 0) {};
\end{pgfonlayer}
\begin{pgfonlayer}{edgelayer}
\draw [style=simple] (0) to (1);
\end{pgfonlayer}
\end{tikzpicture}
:=
\begin{tikzpicture}[baseline={([yshift=-.5ex]current bounding box.center)}]
\begin{pgfonlayer}{nodelayer}
\node [style=onein] (0) at (0, .5) {};
\node [style=dot] (1) at (.5, .5) {};
\node [style=oneout] (2) at (1, .5) {};
\node [style=onein] (10) at (0, 0) {};
\node [style=oplus] (11) at (.5, 0) {};
\node [style=rn] (12) at (1, 0) {};
\end{pgfonlayer}
\begin{pgfonlayer}{edgelayer}
\draw [style=simple] (0) to (2);
\draw [style=simple] (10) to (12);
\draw [style=simple] (1) to (11);
\end{pgfonlayer}
\end{tikzpicture}
\hspace*{.5cm}
\text{ and dually }
\hspace*{.5cm}
\begin{tikzpicture}[baseline={([yshift=-.5ex]current bounding box.center)}]
\begin{pgfonlayer}{nodelayer}
\node [style=rn] (0) at (0, 0) {};
\node [style=zeroout] (1) at (.5, 0) {};
\end{pgfonlayer}
\begin{pgfonlayer}{edgelayer}
\draw [style=simple] (0) to (1);
\end{pgfonlayer}
\end{tikzpicture}
:=
\begin{tikzpicture}[baseline={([yshift=-.5ex]current bounding box.center)}]
\begin{pgfonlayer}{nodelayer}
\node [style=onein] (0) at (0, .5) {};
\node [style=dot] (1) at (.5, .5) {};
\node [style=oneout] (2) at (1, .5) {};
\node [style=rn] (10) at (0, 0) {};
\node [style=oplus] (11) at (.5, 0) {};
\node [style=oneout] (12) at (1, 0) {};
\end{pgfonlayer}
\begin{pgfonlayer}{edgelayer}
\draw [style=simple] (0) to (2);
\draw [style=simple] (10) to (12);
\draw [style=simple] (1) to (11);
\end{pgfonlayer}
\end{tikzpicture}
\end{flalign*}

We denote these ancillary bits algebraically by $\zeroin$ and $\zeroout$ respectively.

%%%%%%%%%%%%%%%%%%%%%%%%%%%%%%%%%%%%%%%%%%%%%%%%%%%%%%%%%%%%%%

\subsection{Restriction and inverse categories}

%%%%%%%%%%%%%%%%%%%%%%%%%%%%%%%%%%%%%%%%%%%%%%%%%%%%%%%%%%%%%%

In this section, we introduce the basic theory and terminology of restriction categories which we use later.

\begin{definition}\cite[Def. 2.1.1]{Cockett}

A \textbf{restriction structure} on a category $\mathbb{X}$ is an assignment $\bar{f}: A\to A$ for each map $f:A\to B$ in $\mathbb{X}$ satisfying the following four axioms:
\begin{enumerate}[label=\textit{(R.\arabic*)},leftmargin=*]
\item \label{R.1} $\bar{f}f =f$ for every map $f$ in $\mathbb{X}$
\item \label{R.2} $\bar{f}\bar{g}=\bar{g}\bar{f}$ whenever $\dom f = \dom g$ for maps $f,g$ in $\mathbb{X}$.
\item \label{R.3} $\bar{\bar{g}f}=\bar{f}\bar{g}$ whenever $\dom f = \dom g$ for maps $f,g$ in $\mathbb{X}$.
\item \label{R.4} $f \bar{g}=\bar{fg}f$ whenever $\cod f = \dom g$ for maps $f,g$ in $\mathbb{X}$.
\end{enumerate}
\end{definition}

A \textbf{restriction category} is a category equipped with a restriction structure. A \textbf{restriction functor} is a functor which preserves the restriction structure.  An endomorphism $e: A \rightarrow A$ is called a \textbf{restriction idempotent} if $e=\bar{e}$.  In particular, each $\bar{f}$  is an idempotent ($\bar{f}\bar{f} = \bar{\bar{f}f} = \bar{f}$) and $\bar{\bar{f}} = \bar{f}$.

In a restriction category, a \textbf{total map} is a map $f$ such that $\bar{f}=1$.  The total maps of a restriction category $\mathbb{X}$ form a subcategory $\Total(\mathbb{X})$ of $\mathbb{X}$. 

A restriction category is a $2$-category with 2-cells given by the partial order $f\leq g\iff f = \bar f g$.  In particular, if $f$ and $g$ are restriction idempotents, then $f \leq g\iff f=f g$ \cite[Sec. 2.1.4]{Cockett}.

A basic example of a restriction categories is a partial map category.  One can form a partial map category from any category which has pullbacks:

\begin{definition}\cite[Sec. 3]{Cockett}  \label{defn:partialMaps}
Given a category $\mathbb{X}$ with pullbacks, \textbf{the category of partial maps} of $\mathbb{X}$, $\text{Par}(\mathbb{X})$ is defined as follows:
\begin{description}
\item[Objects: ] Objects in $\mathbb{X}$.
\item[Maps: ] Every map from $A$ to $B$ is a pair $(m,f)$ such that $m:A'\to A$ and $f:A'\to B$ such that $m$ is monic---up to an equivalence relation $(m,f)\sim (m',f')$ if and only if there exists an isomorphism $\alpha$ such that $\alpha m'=m$ and $\alpha f' = f$.
\item[Identities: ]  The identity on $A$ is the pair $(1_A,1_A)$.
\item[Composition: ] For maps $(m,f):A\to B$ and $(m',g):B\to C$, $(m,f)(m',g):=(m''m,f'g)$ where $m''$ and $f'$ are determined by the following pullback:
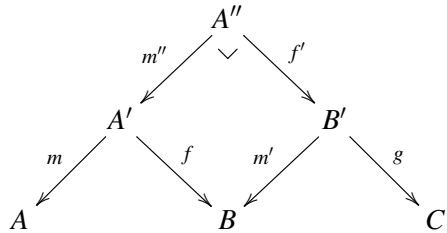
\[
\xymatrix{
&&A'' \pullbackcorner \ar[dl]_{m''} \ar[dr]^{f'}\\
&A' \ar[dl]_m \ar[dr]^f &&B' \ar[dl]_{m'} \ar[dr]^g\\
A&&B&&C
}
\]
Composition is well-defined even though pullbacks are determined only up to isomorphism as the maps are taken modulo the equivalence relation.
\end{description}
$\Par(\mathbb{X})$ is endowed with a restriction structure by $\bar{(m,f)}:=(m,m)$.
\end{definition}

The notion of a partial map category can be generalized by restricting the monics to any class of monics closed to composition, isomorphisms and pullbacks \cite[Sec. 3]{Cockett}. 
However, here we consider only the class of all monics.

\begin{definition}\cite[Sec. 2.3]{Cockett}
A map $f$ is a \textbf{partial isomorphism} when there exists another map $g$, called the \textbf{partial inverse} of $f$, such that $\bar f = fg$ and $\bar g = gf$. 
\end{definition}

Partial isomorphisms generalize the notion of an isomorphisms to restriction categories; thus, the composition of partial isomorphisms is a partial isomorphism and partial inverses are unique.
A restriction category $\mathbb{X}$ is an \textbf{inverse category} when all its maps are partial isomorphisms. A one object inverse category is an inverse monoid. 

Given any restriction category $\mathbb{X}$, there is a subcategory of partial isomorphisms of $\mathbb{X}$, denoted by $\ParIso(\mathbb{X})$ which is an inverse category.

There is an important alternate way to view an inverse category:

\begin{theorem}\cite[Thm. 2.20]{Cockett}  \label{defn:inverseCategory}
A category $\mathbb{X}$ is an inverse category if and only if there exists an involution ${(\_)\cnv:\mathbb{X}\op \to\mathbb{X}}$ which is the identity on objects, satisfying the following three axioms:
\begin{enumerate}[label=\textit{(Inv.\arabic*)},leftmargin=*]
\item \label{INV:1} $(c\cnv)\cnv = c$
\item \label{INV:2} ${cc\cnv c = c}$
\item \label{INV:3} ${cc\cnv dd\cnv = dd\cnv cc\cnv}$
\end{enumerate} 
\end{theorem}

Inverse categories have restriction structure given by $\bar c := cc\cnv$.
It is not hard to show that every idempotent in an inverse category is necessarily a restriction idempotent.

%%%%%%%%%%%%%%%%%%%%%%%%%%
\subsection{Discrete restriction categories and inverse products}
%%%%%%%%%%%%%%%%%%%%%%%%%%

If a category $\X$ has products then $\Par(\X)$ has restriction products: these are ``lax'' products for which the pairing operation satisfies  
$\la f,g\ra \pi_0 = \bar{g} f$ (and $\la f,g\ra \pi_1 = \bar{f} g$).  If the category $\X$ has a final object then $\Par(\X)$ has a restriction final object, that is 
an object $!$ for which, for each object $A$, there is a unique total map $!: A \to 1$ so that for any map $k: A \to 1$, $k= \bar{k} !$.  A restriction category with 
restriction products and a restriction terminal object is called a \textbf{Cartesian  restriction category}.  

A Cartesian restriction category in which the diagonal map, $\Delta_A: A \to A \times A$, is a partial isomorphism is called 
a \textbf{discrete Cartesian restriction category}.   The partial map category, $\Par(\X)$, of a category with products and pullbacks is always a 
discrete Cartesian restriction category (as the diagonal map is a partial isomorphism).  Discrete Cartesian restriction categories are
equivalently characterized as those Cartesian restriction categories which  have meets: 

\iffalse
\begin{definition}\cite[Lemma 2.11]{range}
A restriction category $\mathbb{X}$ has \textbf{meets} when it has a combinator $\_\cap\_:\mathbb{X}(A,B)^2\to \mathbb{X}(A,B)$ for all objects $A$ and $B$ in $\mathbb{X}$ such that for all $f,g:A\to B$ and $h:B\to C$ in $\mathbb{X}$ the three following axioms hold:
\begin{description}
\item[(M.1)] $f\cap f =f$
\item[(M.2)] $f \cap g \leq f$ and $f \cap g \leq g$ with respect to the restriction ordering.
\item[(M.3)] $(f\cap g)h = (fh) \cap (gh)$
\end{description}
\end{definition}

In a discrete Cartesian restriction category one defines the meet as $f \cap g := \Delta (f \times g) \Delta\cnv$  conversely if one has meets in a Cartesian restriction category one can define 
the partial inverse of the diagonal map as $\Delta\cnv := \pi_0 \cap \pi_1$.  When one considers the subcategory of partial isomorphism of a discrete Cartesian category one obtains an inverse 
category which has a residual product structure. As the projections are not (partial) isomorphisms this structure centres round the behaviour of $\Delta$ and $\Delta\cnv$:

\fi

\begin{definition}\cite[Def. 4.3.1]{Giles} Given an inverse category $\mathbb{X}$ equipped with a tensor product ${\_\otimes\_:\mathbb{X}\times\mathbb{X}\to \mathbb{X}}$ which preserves $(\_)\cnv$, we say $\mathbb{X}$ has \bf{inverse products} if there exists a total natural diagonal transformation $\Delta$ which satisfies the properties:

\begin{enumerate}[label=\textit{(DInv.\arabic*)},leftmargin=*]
\item
\label{DNV:1}
$\Delta$ is cocommutative for each $A \in \mathbb{X}$: % and commutative,
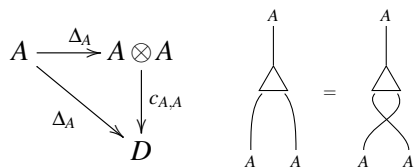
\[\begin{array}[c]{c}
\xymatrix{
A \ar[dr]_{\Delta_A} \ar[r]^{\Delta_A} &A\otimes A\ar[d]^{c_{A,A}}\\
&D
} \end{array}
 ~~~~
 \begin{array}[c]{c}
\begin{tikzpicture}
	\begin{pgfonlayer}{nodelayer}
		\node [style=uptriangle] (0) at (-2.25, 2) {};
		\node [style=uptriangle] (1) at (0.25, 2) {};
		\node [style=none] (2) at (-2.25, 3.25) {};
		\node [style=none] (3) at (-2.75, 0.5) {};
		\node [style=none] (4) at (-1.75, 0.5) {};
		\node [style=none] (5) at (0.75, 0.5) {};
		\node [style=none] (6) at (-0.25, 0.5) {};
		\node [style=none] (7) at (0.25, 3.25) {};
		\node [style=none] (8) at (-2.5, 1.75) {};
		\node [style=none] (9) at (-2, 1.75) {};
		\node [style=none] (10) at (0, 1.75) {};
		\node [style=none] (11) at (0.5, 1.75) {};
		\node [style=none] (12) at (-1, 1.75) {$=$};
		\node [style=none] (13) at (-2.25, 3.5) {$A$};
		\node [style=none] (14) at (0.25, 3.5) {$A$};
		\node [style=none] (15) at (-2.75, 0.25) {$A$};
		\node [style=none] (16) at (-1.75, 0.25) {$A$};
		\node [style=none] (17) at (-0.25, 0.25) {$A$};
		\node [style=none] (18) at (0.75, 0.25) {$A$};
	\end{pgfonlayer}
	\begin{pgfonlayer}{edgelayer}
		\draw [style=none, in=90, out=-135, looseness=0.75] (8.center) to (3.center);
		\draw [style=none, in=90, out=-49, looseness=0.75] (9.center) to (4.center);
		\draw [style=none] (0) to (2.center);
		\draw [style=none] (1) to (7.center);
		\draw [style=none, in=90, out=-135, looseness=0.75] (10.center) to (5.center);
		\draw [style=none, in=90, out=-60, looseness=1.00] (11.center) to (6.center);
	\end{pgfonlayer}
\end{tikzpicture}
\end{array}
\]
\item
\label{DNV:2}
$\Delta$ is coassociative for each $A \in \mathbb{X}$:% and associative,
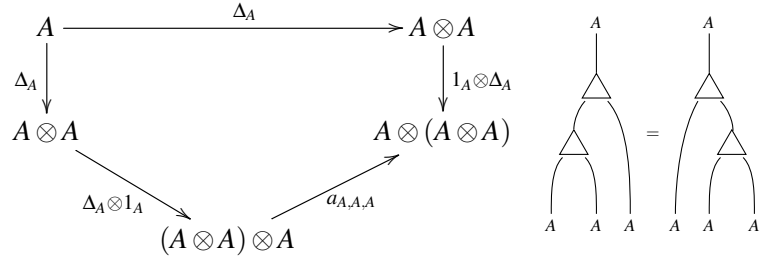
\[ \begin{array}[c]{c}
\xymatrix{
A \ar[rr]^{\Delta_A} \ar[d]_{\Delta_A} && A\otimes A \ar[d]^{1_A\otimes \Delta_A}\\
A\otimes A \ar[dr]_{\Delta_A\otimes 1_A} & & A\otimes(A\otimes A)\\
& (A\otimes A)\otimes A \ar[ur]_{a_{A,A,A}} &
}
\end{array}
\begin{array}[c]{c}
\begin{tikzpicture}
	\begin{pgfonlayer}{nodelayer}
		\node [style=uptriangle] (0) at (1, 2) {};
		\node [style=none] (1) at (3.5, 3.25) {};
		\node [style=none] (2) at (3.75, 1.75) {};
		\node [style=none] (3) at (1.75, -0.75) {};
		\node [style=none] (4) at (1.25, 1.75) {};
		\node [style=none] (5) at (3.25, 1.75) {};
		\node [style=uptriangle] (6) at (3.5, 2) {};
		\node [style=none] (7) at (1, 3.25) {};
		\node [style=none] (8) at (0.75, 1.75) {};
		\node [style=none] (9) at (2.75, -0.75) {};
		\node [style=uptriangle] (10) at (4, 0.75) {};
		\node [style=none] (11) at (3.75, 0.5) {};
		\node [style=none] (12) at (4.25, 0.5) {};
		\node [style=none] (13) at (3.5, -0.75) {};
		\node [style=none] (14) at (4.5, -0.75) {};
		\node [style=none] (15) at (0.75, 0.5) {};
		\node [style=none] (16) at (1, -0.75) {};
		\node [style=none] (17) at (0, -0.75) {};
		\node [style=uptriangle] (18) at (0.5, 0.75) {};
		\node [style=none] (19) at (0.25, 0.5) {};
		\node [style=none] (20) at (2.25, 1) {$=$};
		\node [style=none] (21) at (1, 3.5) {$A$};
		\node [style=none] (22) at (3.5, 3.5) {$A$};
		\node [style=none] (23) at (0, -1) {$A$};
		\node [style=none] (24) at (1, -1) {$A$};
		\node [style=none] (25) at (1.75, -1) {$A$};
		\node [style=none] (26) at (2.75, -1) {$A$};
		\node [style=none] (27) at (3.5, -1) {$A$};
		\node [style=none] (28) at (4.5, -1) {$A$};
	\end{pgfonlayer}
	\begin{pgfonlayer}{edgelayer}
		\draw [style=none, in=90, out=-49, looseness=0.75] (4.center) to (3.center);
		\draw [style=none] (0) to (7.center);
		\draw [style=none] (6) to (1.center);
		\draw [style=none, in=90, out=-120, looseness=0.75] (5.center) to (9.center);
		\draw [style=none, in=86, out=-45, looseness=0.75] (2.center) to (10);
		\draw [style=none, in=90, out=-135, looseness=0.75] (11.center) to (13.center);
		\draw [style=none, in=90, out=-45, looseness=0.75] (12.center) to (14.center);
		\draw [style=none, in=90, out=-135, looseness=0.75] (19.center) to (17.center);
		\draw [style=none, in=90, out=-45, looseness=0.75] (15.center) to (16.center);
		\draw [style=none, in=90, out=-135, looseness=0.75] (8.center) to (18);
	\end{pgfonlayer}
\end{tikzpicture}
\end{array}
\]

\item
\label{DNV:3}
$(\Delta,\Delta\cnv)$ is a semi-Frobenius object for each $A \in \mathbb{X}$:
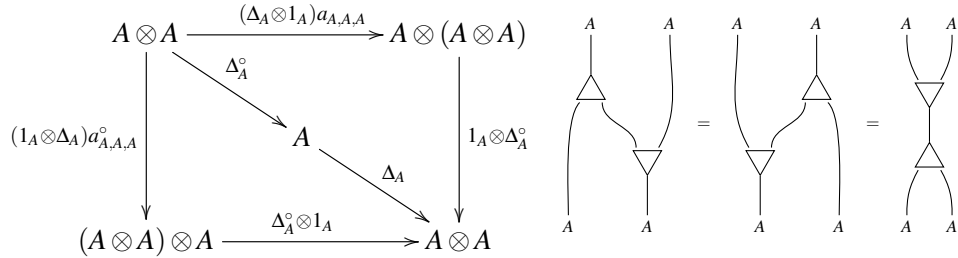
\[\begin{array}[c]{c}
\xymatrix{
A\otimes A \ar[rr]^{(\Delta_A\otimes 1_A)a_{A,A,A}} \ar[dr]^{\Delta_A\cnv} \ar[dd]_{(1_A\otimes \Delta_A)a_{A,A,A}\cnv} && A\otimes(A\otimes A) \ar[dd]^{1_A\otimes \Delta_A\cnv}\\
& A  \ar[dr]^{\Delta_A} &\\
(A\otimes A)\otimes A \ar[rr]^{\Delta_A\cnv\otimes1_A} && A\otimes A
}
\end{array}
\begin{array}[c]{c}
\begin{tikzpicture}
	\begin{pgfonlayer}{nodelayer}
		\node [style=uptriangle] (0) at (-4, 1.25) {};
		\node [style=none] (1) at (-4, 2.5) {};
		\node [style=none] (2) at (-4.5, -1.5) {};
		\node [style=downtriangle] (3) at (-2.75, -0.25) {};
		\node [style=none] (4) at (-2.75, -1.5) {};
		\node [style=none] (5) at (-2.25, 2.5) {};
		\node [style=none] (6) at (-4.25, 1) {};
		\node [style=none] (7) at (-3.75, 1) {};
		\node [style=none] (8) at (-3, 0) {};
		\node [style=none] (9) at (-2.5, 0) {};
		\node [style=none] (10) at (0, 0) {};
		\node [style=uptriangle] (11) at (1, 1.25) {};
		\node [style=downtriangle] (12) at (-0.25, -0.25) {};
		\node [style=none] (13) at (1.5, -1.5) {};
		\node [style=none] (14) at (-0.25, -1.5) {};
		\node [style=none] (15) at (1.25, 1) {};
		\node [style=none] (16) at (-0.5, 0) {};
		\node [style=none] (17) at (1, 2.5) {};
		\node [style=none] (18) at (0.75, 1) {};
		\node [style=none] (19) at (-0.75, 2.5) {};
		\node [style=downtriangle] (20) at (3.5, 1.25) {};
		\node [style=uptriangle] (21) at (3.5, -0.25) {};
		\node [style=none] (22) at (3.25, 1.5) {};
		\node [style=none] (23) at (3, 2.5) {};
		\node [style=none] (24) at (4, 2.5) {};
		\node [style=none] (25) at (3.75, 1.5) {};
		\node [style=none] (26) at (3.25, -0.5) {};
		\node [style=none] (27) at (3.75, -0.5) {};
		\node [style=none] (28) at (3, -1.5) {};
		\node [style=none] (29) at (4, -1.5) {};
		\node [style=none] (30) at (-1.5, 0.5) {$=$};
		\node [style=none] (31) at (2.25, 0.5) {$=$};
		\node [style=none] (32) at (-4, 2.75) {$A$};
		\node [style=none] (33) at (-2.25, 2.75) {$A$};
		\node [style=none] (34) at (-0.75, 2.75) {$A$};
		\node [style=none] (35) at (1, 2.75) {$A$};
		\node [style=none] (36) at (3, 2.75) {$A$};
		\node [style=none] (37) at (4, 2.75) {$A$};
		\node [style=none] (38) at (-4.5, -1.75) {$A$};
		\node [style=none] (39) at (-2.75, -1.75) {$A$};
		\node [style=none] (40) at (-0.25, -1.75) {$A$};
		\node [style=none] (41) at (1.5, -1.75) {$A$};
		\node [style=none] (42) at (3, -1.75) {$A$};
		\node [style=none] (43) at (4, -1.75) {$A$};
	\end{pgfonlayer}
	\begin{pgfonlayer}{edgelayer}
		\draw [style=none] (1.center) to (0);
		\draw [style=none] (3) to (4.center);
		\draw [style=none, in=90, out=-135, looseness=0.50] (6.center) to (2.center);
		\draw [style=none, in=90, out=-90, looseness=1.00] (7.center) to (8.center);
		\draw [style=none, in=-90, out=60, looseness=1.00] (9.center) to (5.center);
		\draw [style=none] (17.center) to (11);
		\draw [style=none] (12) to (14.center);
		\draw [style=none, in=90, out=-45, looseness=0.50] (15.center) to (13.center);
		\draw [style=none, in=90, out=-90, looseness=1.00] (18.center) to (10.center);
		\draw [style=none, in=-90, out=120, looseness=1.00] (16.center) to (19.center);
		\draw [style=none, in=119, out=-91, looseness=1.00] (23.center) to (22.center);
		\draw [style=none, bend right=15, looseness=1.00] (25.center) to (24.center);
		\draw [style=none, bend right=15, looseness=1.00] (26.center) to (28.center);
		\draw [style=none, bend left=15, looseness=1.00] (27.center) to (29.center);
		\draw [style=none] (20) to (21);
	\end{pgfonlayer}
\end{tikzpicture}
\end{array}
\]
\item
\label{DNV:4}
$\Delta$ satisfies the uniform copying identity for each $A,B \in \mathbb{X}$:
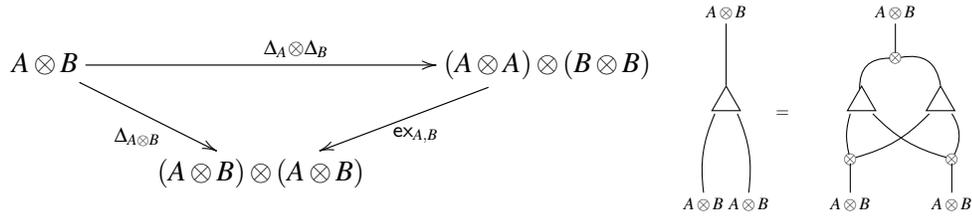
\[
\begin{array}[c]{c}
\xymatrix{
A\otimes B \ar[rr]^{\Delta_A\otimes \Delta_B} \ar[dr]_{\Delta_{A\otimes B}}&& (A\otimes A)\otimes (B\otimes B) \ar[dl]^{\mathsf{ex}_{A,B}}\\
&(A\otimes B)\otimes (A\otimes B)
}
\end{array}
\begin{array}[c]{c}
\begin{tikzpicture}
	\begin{pgfonlayer}{nodelayer}
		\node [style=none] (0) at (-3.75, 3.25) {};
		\node [style=uptriangle] (1) at (-3.75, 1.5) {};
		\node [style=none] (2) at (-4.25, -0.5) {};
		\node [style=none] (3) at (-3.25, -0.5) {};
		\node [style=none] (4) at (-4, 1.25) {};
		\node [style=none] (5) at (-3.5, 1.25) {};
		\node [style=none] (6) at (0, 3.25) {};
		\node [style=none] (7) at (0, 2.5) {$\otimes$};
		\node [style=uptriangle] (8) at (-0.75, 1.5) {};
		\node [style=uptriangle] (9) at (1, 1.5) {};
		\node [style=none] (10) at (-1, 1.25) {};
		\node [style=none] (11) at (-0.5, 1.25) {};
		\node [style=none] (12) at (0.75, 1.25) {};
		\node [style=none] (13) at (1.25, 1.25) {};
		\node [style=none] (14) at (-1, 0.25) {$\otimes$};
		\node [style=none] (15) at (1.25, 0.25) {$\otimes$};
		\node [style=none] (16) at (-1, -0.5) {};
		\node [style=none] (17) at (1.25, -0.5) {};

		\node [style=none] (18) at (-2.5, 1.25) {$=$};
		\node [style=none] (19) at (-3.75, 3.5) {$A\otimes B$};
		\node [style=none] (20) at (0, 3.5) {$A\otimes B$};
		\node [style=none] (21) at (-4.25, -0.75) {$A\otimes B$};
		\node [style=none] (22) at (-3.25, -0.75) {$A\otimes B$};
		\node [style=none] (23) at (-1, -0.75) {$A\otimes B$};
		\node [style=none] (24) at (1.25, -0.75) {$A\otimes B$};
	\end{pgfonlayer}
	\begin{pgfonlayer}{edgelayer}
		\draw [style=none] (0.center) to (1);
		\draw [style=none, bend right=15, looseness=1.00] (4.center) to (2.center);
		\draw [style=none, bend left=15, looseness=1.00] (5.center) to (3.center);
		\draw [style=none, bend right=15, looseness=1.00] (10.center) to (14);
		\draw [style=none, bend left=15, looseness=1.00] (12.center) to (14);
		\draw [style=none] (14) to (16.center);
		\draw [style=none, bend right=15, looseness=1.00] (11.center) to (15);
		\draw [style=none, bend left, looseness=1.00] (13.center) to (15);
		\draw [style=none] (15) to (17.center);
		\draw [style=none, bend right=45, looseness=1.25] (7) to (8);
		\draw [style=none, bend left=45, looseness=1.25] (7) to (9);
		\draw [style=none] (7) to (6.center);
	\end{pgfonlayer}
\end{tikzpicture}
\end{array}
\]

\end{enumerate}
\end{definition}

Where the natural isomorphism,
\begin{align*}
{\sf ex} :&= a (1\otimes a\cnv)(1\otimes (c \otimes 1))((1\otimes a)a\cnv)
:(A\otimes B)\otimes(C\otimes D)\to(A\otimes C)\otimes(B \otimes D)
\end{align*}
is called the exchange map.

A \bf{discrete inverse category} is an inverse category with inverse products.  Note that $\Delta$ is total if and only if $\Delta$ is separable (special), that is, $\Delta \Delta\cnv = 1$.  A discrete inverse category  
always has meets \cite{Giles}: $f \cap g := \Delta (f \otimes g) \Delta\cnv$.  Furthermore, the partial isomorphisms of a discrete Cartesian restriction category (such as $\Par(\X)$ for a $\X$ with finite limits) is always a 
discrete inverse category.   Conversely -- and more surprisingly -- every discrete inverse category has a ``completion" to a discrete Cartesian restriction category (see \cite{Giles} for more details).

%%%%%%%%%%%%%%%%%%%%%%%%%%%%%
\section{Torsors}
%%%%%%%%%%%%%%%%%%%%%%%%%%%%%

We will prove that $\CNOT$ is equivalent to the category of partial isomorphism between finitely generated non-empty commutative torsors of characteristic 2, $\ParIso(\CTor_2)^{*}$.   Torsors are essentially groups without a fixed multiplicative identity: the category $\ParIso(\CTor_2)^{*}$ may, thus, also be viewed as the partial isomorphism category of finite-dimensional $\mathbb{Z}_2$ vector spaces with affine maps.

\begin{definition}%\cite[Def. 18]{torsororig}, \cite[Sec. 0.2]{torsor}, 
\label{defn:Torsor}
A \textbf{torsor} is a set $X$ along with a ternary operator $(\_)\times_{(\_)}(\_): X \times X \times X  \to X$ called para-multiplication, such that for any $a,b,c,d,e \in X$, the following laws hold \cite{torsororig}:
\begin{description}
\item[Para-associative law: ]
$$(a\times_b c)\times_d e = a\times_{d\times_c b} e = a\times_b (c\times_d e)$$
\item[Para-identity law: ]
$$a\times_b b = b\times_b a = a$$
\end{description}
\end{definition}
A torsor is said to be \textbf{commutative}, when $a\times_b c = c \times_b a$.
A torsor is said to have \textbf{characteristic 2}, when $a\times_b a = b$.

The category of torsors $\Tor$ has objects torsors and maps homomorphisms of torsors.  A homomorphisms of torsors, $f:(X,\times)\to (Y,\times)$, is a functions $X\to Y$ which preserve para-multiplication.  As this is a category 
of algebras we know that it is a finitely complete category.  This allows us to form $\Par(\Tor)$ and $\ParIso(\Tor)$ immediately.

Note that the empty set is a torsor, however, if $(X,\times)$ is a non-empty torsor, then $X$ has, for each element of $X$, a group structure. Thus, given any $z \in X$, $X$ is a group under the multiplication 
$$\_\cdot\_:X^2\to X; (x,y)\mapsto x\times_z y.$$ 
Conversely, if $(X,\cdot)$ is a group, then $X$ has a non-empty torsor structure $\_\times_{\_}\_:X^3\to X$ such that $(x,z,y) \mapsto x\cdot z^{-1}\cdot y$ \cite[Sec. 0.2]{torsor}.  Note that this correspondence does not imply that the category of torsors and groups are equivalent since their homomorphisms are different.

Some authors, including their originator \cite[Def. 18]{torsororig}, require the underlying set of a torsor to be nonempty so that torsors always arise as groups.  However, following \cite[Sec. 0.2]{torsor}, we will not impose this condition as we need a category closed to pullbacks, and the empty torsor arises as a pullback.  A torsors is also also known as a ``... heap, groud, flock, herd, principal homogeneous space, abstract coset, pregroup ...'' \cite[Sec. 0.2]{torsor} with the non-emptiness condition appearing in some cases.
%https://arxiv.org/pdf/0903.5441.pdf

\begin{definition}
\label{defn:tor2}
Define $\CTor_2$ to be the full subcategory of torsors whose objects are finitely generated commutative torsors of characteristic 2 (including the empty torsor).
\end{definition}

There is an equivalent characterization of  the objects of $\CTor_2$.

\begin{proposition}
\label{rem:tor2}
Every object in $\CTor_2$ is either empty or isomorphic to a finite dimensional $\mathbb{Z}_2$ vector spaces; furthermore, torsor homomorphisms are precisely the affine maps.
\end{proposition}

\begin{proof}
Suppose that $X$ is a finitely generated commutative torsor under the para-multiplication $\_\times_{\_}\_:X^3\to X$. If $X$ is nonempty, fix some element $z \in X$.  As $X$ has characteristic $2$ as a commutative torsor, it has characteristic $2$ as a Abelian group under the addition $\_+\_:=\_\times_z\_:X^2\to X$.  Furthermore, the dimension of such a torsor is one more than the dimension of this corresponding group (as the base point must be added).  Thus, finitely generated commutative torsors of characteristic 2 are finite.  Therefore, by the fundamental theorem of finitely generated Abelian groups:
$$X\cong \mathbb{Z}^0 \oplus \left(\bigoplus_{i=1}^n \mathbb{Z}_2\right)\cong \mathbb{Z}_2^n$$
with para-multiplication given by:
$$(a,b,c) \mapsto a \oplus (-b) \oplus c = a\oplus b\oplus c$$

Given a morphism of non-empty torsors, $f:(\mathbb{Z}_2^n, \_\oplus\_\oplus\_)\to (\mathbb{Z}_2^m, \_\oplus\_\oplus\_)$, for any $x$ and $y$ in $\mathbb{Z}_2^n$,
$$f(x\oplus y) = f(x\oplus 0 \oplus y)= f(x)\oplus f(y) \oplus f(0)$$
Therefore, $f$ is an affine transformation when $\mathbb{Z}_2^n$ and $\mathbb{Z}_2^m$ are seen as vector spaces over $\mathbb{Z}_2$.

Conversely, consider an affine transformation of vector spaces $f:\mathbb{Z}_2^n\to \mathbb{Z}_2^m$.  Then $f$ can be regarded as morphism of torsors as for any $x,y,z \in \mathbb{Z}_2^n$,
\begin{align*}
f(x\oplus y\oplus z) &= f(x\oplus (y\oplus z)) = f(x) \oplus f(y\oplus z)\oplus f(0)\\
&= f(x) \oplus f(y) \oplus f(z) \oplus f(0) \oplus f(0)\\
&= f(x) \oplus f(y) \oplus f(z)
\end{align*}

The empty torsor and the empty affine space are strict initial objects; therefore, they have the same maps in torsors and Abelian groups viewed as affine spaces.
\end{proof}

As  $\CTor_2$ is a category of algebras and so is finitely complete, we may construct $\Par(\CTor_2)$ and $\ParIso(\CTor_2)$. Let $\ParIso(\CTor_2)^*$ denote $\ParIso(\CTor_2)$ without the empty torsor. 

\begin{proposition}
\label{prop:torDiscreteInverse}
$\ParIso(\CTor_2)^*$ is a discrete inverse category.
\end{proposition}

\begin{proof}
$\ParIso(\CTor_2)^*$ is an inverse category by construction.  Because $\Par(\CTor_2)^*$ is a discrete Cartesian restriction category and $\ParIso(\CTor_2)^*$ is the category of partial isomorphisms, it therefore has inverse products (see \cite[Theorem 4.3.7]{Giles}).  Thus, it is a discrete inverse category. 
\end{proof}

One further property of $\CTor_2$ is worth mentioning: all its monic maps are {\em regular\/} monics.  This means, more concretely, every subobject of a torsor is determined by some set of equations.  This then means that in $\Par(\CTor_2)^{*}$ the restriction idempotents correspond to equations.

\section{Overview of Proof}

The main theorem of this paper is:

\vspace{0.5cm}
\noindent
\textbf{Theorem}  \ref{thm:CNOTEquiv} \\
{\em There is an equivalence of categories between $\CNOT$ and $\ParIso(\CTor)^{*}$.}

\vspace{0.5cm}
%To prove this, it is necessary to understand some of the properties of the category $\CNOT$. 

The equivalence is shown in the following steps:

\begin{enumerate}

\item Proof that $\CNOT$ is a discrete inverse category.

The first major challenge is to prove that $\CNOT$ is a discrete inverse category.   We approach this by setting up the ``discrete'' part of the structure first. For this, we construct the ``copy'' natural transformation $\Delta$ which is defined inductively. The base case of $1$ wire is defined by applying a $\cnot$ gate to a 0 ancillary bit.  Then we prove that $\Delta$ has the properties required by an inverse product.  This involves showing that the family of maps $\{\Delta_n:n\to 2n\}_{n \in \mathbb{N}}$ is a natural transformation i.e., for any circuit $f:n\to m$ in $\CNOT$, $f\Delta_m = \Delta_{2n} (f\otimes f)$. Naturality of $\Delta$ is proven by a structural induction on $f$.  Next, we prove that $\Delta$ forms a total semi-Frobenius algebra.  

$\CNOT$ has an important symmetry expressed by a functor $(\_)^\circ: \CNOT^{\rm op} \to \CNOT$ which ``horizontally flips'' maps (circuits).   We use this functor to prove that $\CNOT$ is an inverse category.  %This shows that the involution applied to a map produces its partial inverse.

\item Construction of a functor $\tilde H_0: \CNOT \rightarrow \ParIso(\CTor_2)^*$.

Our final objective is to prove that the category of partial isomorphisms between non-empty, finitely generated commutative torsors of characteristic $2$, $\ParIso(\CTor_2)^*$ is equivalent to $\CNOT$. In order to establish this, we construct a functor $\tilde H_0: \CNOT \rightarrow \ParIso(\CTor_2)^*$ indirectly by constructing a functor $h_0: \CNOT \rightarrow \Par(\Sets)$. We show that $h_0$ can be factored in two ways:
\begin{itemize}
\item Through the inclusion $\iota: \ParIso(\Sets) \to \Par(\Sets)$ 
\item Through the underlying functor $\Par(U): \Par(\CTor_2)^* \rightarrow \Par(\Sets)$
\end{itemize}
These factorizations imply $H_0$ factors through the pullback of $\iota$ and $\Par(U)$ which is $\ParIso(\CTor_2)^*$. Thus, we obtain a map $\tilde H_0: \CNOT \to \ParIso(\CTor_2)^*$.

There is another important way to describe this functor: as $\CNOT$ is the freely generated symmetric monoidal category on the gates $\cnot$, $\onein$, and $\oneout$, it suffices to interpret these into the category $\ParIso(\CTor_2)^*$ and check that all of the identities hold.  The result is equivalent to the functor we have produced and has the virtue of showing fairly immediately that $\tilde H_0$ preserves discrete inverse structure.

Once we have proven that the functor $\tilde H_0:\CNOT\to\ParIso(\CTor_2)^*$ is well-defined, we prove that it an equivalence of categories.  That is, that it is full, faithful, and essentially surjective.  The proofs are in Appendix \ref{Appendix C}.  The essential surjectivity is straightforward; however, the other two are not.

\item Proof that $\tilde H_0: \CNOT \rightarrow \ParIso(\CTor_2)^*$ is a full functor.

To obtain the fullness we need to show that any partial isomorphism between torsors can be simulated using circuits in $\CNOT$.  However, we first show that we can simulate the graph,  $\la 1,f \ra$,  of a total map $f$ between torsors (see Lemma \ref{simulating_total_maps}).  Next we observe that any partial map between torsors is dominated by a total map: so we can simulate the graph of this total map.   The difficulty is now to introduce the partiality.  To achieve this we use the fact that the functor is full on restriction idempotents (see Theorem \ref{Theorem: restriction faithful} - more on this soon): this allows us to simulate $\la \bar{f},f \ra$ for any partial map $f$.  However, by a general result (see Lemma \ref{lemma:fullCopy}) we observe that simulating graphs of partial isomorphisms is sufficient to ensure we can simulate any partial isomorphism.

\item Proof that $\tilde H_0: \CNOT \rightarrow \ParIso(\CTor_2)^*$ is a faithful functor.

To secure faithfulness we start by reducing the problem to showing faithfulness on restriction idempotents.  This involves two observations.  First we prove that a functor from an inverse category is faithful if and only if it is faithful on restriction idempotents and reflects them (see Lemma \ref{lemma:faithfulidempotents}).  Next, we prove that a functor between discrete inverse categories which preserves the inverse product structure and is faithful on restriction idempotents always reflects restriction idempotents and, thus, is faithful (see Lemma \ref{lemma:faithfullemma}).  Thus to establish faithfulness in our case it suffices to prove that $\tilde H_0$ is faithful on restriction idempotents.

Toward this end we introduce a normal form, called ``clausal form'', for restriction idempotents in $\CNOT$  (Definition \ref{clausalform}).  A circuit is in clausal form if and only if it is the composite of a finite number of clauses.  A clause is a wire starting and ending with ancillary bits which is controlled from the identity (on $n$ wires).  It has the effect of restricting the ``legal'' throughputs on the identity.  Clauses correspond to torsor equations and can be manipulated by Gaussian elimination: and this means they are in bijective correspondence to the restriction idempotents in $\ParIso(\CTor_2)^{*}$.  This is the content of Theorem \ref{Theorem: restriction faithful} which is the crux of the proof of Theorem  \ref{thm:CNOTEquiv}.

\end{enumerate}

\section{Concluding Remarks}

In this work, we provided a complete set of identities for $\CNOT$, the symmetric monoidal category generated by the $\cnot$ gate and the computational ancill\ae. We proved that $\CNOT$ is equivalent to the category of partial isomorphisms of non-empty finitely generated commutative torsors of characteristic $2$, $\ParIso(\CTor_2)^*$.

To the best of our knowledge, our work is the first to provide a complete set of identities for the $\cnot$ gate and computational ancill\ae. The proof we present shows that $\CNOT$ is equivalent to a certain category of torsors. This we admit leaves a gap as one might expect a faithful functor from $\CNOT$ to $\mathsf{CPM(FHilb)}$, the category of finite dimensional Hilbert spaces  with completely positive maps.  Furthermore, as this embedding factors through the 
subcategory of stabilizer circuits which is equivalent to the ZX-calculus \cite{ZXcomplete}, there should also be a {\em faithful\/} embedding of $\CNOT$ into the ZX-calculus.

\medskip
\noindent
\textbf{Acknowledgements:}

We are very grateful to the referees who not only provided many useful comments but also were able to see in our first -- and very rough -- version of this paper that a complete proof might be hiding within!

%\nocite{*}
\bibliographystyle{eptcs}
\bibliography{CNOT}

\newpage
\appendix

\section{Preliminary results for \texorpdfstring{$\CNOT$}{CNOT}}
\label{Appendix A}
Because of Axiom \ref{CNTix}, certain circuits can interfere with the circuits with which they are tensored.  This peculiar behaviour warrants discussion as these circuits are important later.

\begin{definition}
\label{defn:Omega}

The \textbf{degenerate circuit} in $\CNOT$ is:
\begin{align*}
\Omega:=
\begin{tikzpicture}[baseline={([yshift=-.5ex]current bounding box.center)}]
\begin{pgfonlayer}{nodelayer}
\node [style=onein] (0) at (0, .5) {};
\node [style=dot] (1) at (.5, .5) {};
\node [style=oneout] (2) at (1, .5) {};
\node [style=onein] (10) at (0, 0) {};
\node [style=oplus] (11) at (.5, 0) {};
\node [style=oneout] (12) at (1, 0) {};
\end{pgfonlayer}
\begin{pgfonlayer}{edgelayer}
\draw [style=rn]  (0) to (2) {};
\draw [style=rn]  (10) to (12) {};
\draw [style=rn]  (1) to (11) {};
\end{pgfonlayer}
\end{tikzpicture}
:0\to 0
\end{align*}

\end{definition}

This circuit consumes itself in the following sense:

\begin{lemma}
\label{lem:omegaotimesomega}

$\Omega\otimes\Omega = \Omega$

\end{lemma}

\begin{proof}
We observe:
\begin{flalign*}
\begin{tikzpicture}[baseline={([yshift=-.5ex]current bounding box.center)}]
\begin{pgfonlayer}{nodelayer}
\node [style=onein] (0) at (0, 1.5) {};
\node [style=dot] (1) at (.5, 1.5) {};
\node [style=oneout] (2) at (1, 1.5) {};
\node [style=onein] (10) at (0, 1) {};
\node [style=oplus] (11) at (.5, 1) {};
\node [style=oneout] (12) at (1, 1) {};
\node [style=onein] (20) at (0, .5) {};
\node [style=dot] (21) at (.5, .5) {};
\node [style=oneout] (22) at (1, .5) {};
\node [style=onein] (30) at (0, 0) {};
\node [style=oplus] (31) at (.5, 0) {};
\node [style=oneout] (32) at (1, 0) {};
\end{pgfonlayer}
\begin{pgfonlayer}{edgelayer}
\draw [style=rn]  (0) to (2) {};
\draw [style=rn]  (10) to (12) {};
\draw [style=rn]  (1) to (11) {};
\draw [style=rn]  (20) to (22) {};
\draw [style=rn]  (30) to (32) {};
\draw [style=rn]  (21) to (31) {};
\end{pgfonlayer}
\end{tikzpicture}
&=
\begin{tikzpicture}[baseline={([yshift=-.5ex]current bounding box.center)}]
\begin{pgfonlayer}{nodelayer}
\node [style=onein] (0) at (1.5, .5) {};
\node [style=oplus] (1) at (2, .5) {};
\node [style=oneout] (2) at (2.5, .5) {};
\node [style=onein] (10) at (1.5, 0) {};
\node [style=dot] (11) at (2, 0) {};
\node [style=oneout] (12) at (2.5, 0) {};
\node [style=onein] (20) at (0, .5) {};
\node [style=dot] (21) at (.5, .5) {};
\node [style=oneout] (22) at (1, .5) {};
\node [style=onein] (30) at (0, 0) {};
\node [style=oplus] (31) at (.5, 0) {};
\node [style=oneout] (32) at (1, 0) {};
\end{pgfonlayer}
\begin{pgfonlayer}{edgelayer}
\draw [style=rn]  (0) to (2) {};
\draw [style=rn]  (10) to (12) {};
\draw [style=rn]  (1) to (11) {};
\draw [style=rn]  (20) to (22) {};
\draw [style=rn]  (30) to (32) {};
\draw [style=rn]  (21) to (31) {};
\end{pgfonlayer}
\end{tikzpicture}\\
&=
\begin{tikzpicture}[baseline={([yshift=-.5ex]current bounding box.center)}]
\begin{pgfonlayer}{nodelayer}
\node [style=oplus] (1) at (1, .5) {};
\node [style=oneout] (2) at (1.5, .5) {};
\node [style=dot] (11) at (1, 0) {};
\node [style=oneout] (12) at (1.5, 0) {};
\node [style=onein] (20) at (0, .5) {};
\node [style=dot] (21) at (.5, .5) {};
\node [style=onein] (30) at (0, 0) {};
\node [style=oplus] (31) at (.5, 0) {};
\end{pgfonlayer}
\begin{pgfonlayer}{edgelayer}
\draw [style=rn]  (20) to (2) {};
\draw [style=rn]  (30) to (12) {};
\draw [style=rn]  (1) to (11) {};
\draw [style=rn]  (21) to (31) {};
\end{pgfonlayer}
\end{tikzpicture} & \ref{CNTiv}\times 2\\
&=
\begin{tikzpicture}[baseline={([yshift=-.5ex]current bounding box.center)}]
\begin{pgfonlayer}{nodelayer}
\node [style=onein] (0) at (0, .5) {};
\node [style=oplus] (1) at (.5, .5) {};
\node [style=oplus] (2) at (1, .5) {};
\node [style=dot] (3) at (1.5, .5) {};
\node [style=oplus] (4) at (2, .5) {};
\node [style=oneout] (5) at (2.5, .5) {};
\node [style=onein] (10) at (0, 0) {};
\node [style=dot] (11) at (.5, 0) {};
\node [style=dot] (12) at (1, 0) {};
\node [style=oplus] (13) at (1.5, 0) {};
\node [style=dot] (14) at (2, 0) {};
\node [style=oneout] (15) at (2.5, 0) {};
\end{pgfonlayer}
\begin{pgfonlayer}{edgelayer}
\draw [style=rn]  (0) to (5) {};
\draw [style=rn]  (10) to (15) {};
\draw [style=rn]  (1) to (11) {};
\draw [style=rn]  (2) to (12) {};
\draw [style=rn]  (3) to (13) {};
\draw [style=rn]  (4) to (14) {};
\end{pgfonlayer}
\end{tikzpicture} & \ref{CNTii}\\
&=
\begin{tikzpicture}[baseline={([yshift=-.5ex]current bounding box.center)}]
\begin{pgfonlayer}{nodelayer}
\node [style=onein] (0) at (0, .5) {};
\node [style=oplus] (1) at (.5, .5) {};
\node [style=oneout] (5) at (1.5, .5) {};
\node [style=onein] (10) at (0, 0) {};
\node [style=dot] (11) at (.5, 0) {};
\node [style=oneout] (15) at (1.5, 0) {};
\end{pgfonlayer}
\begin{pgfonlayer}{edgelayer}
\draw [style=rn]  (0) to (1) to (11) {};
\draw [style=rn]  (1) to (15) {};
\draw [style=rn]  (10) to (11) to (5) {};
\end{pgfonlayer}
\end{tikzpicture} & \ref{CNTi}\\
&=
\begin{tikzpicture}[baseline={([yshift=-.5ex]current bounding box.center)}]
\begin{pgfonlayer}{nodelayer}
\node [style=onein] (0) at (0, .5) {};
\node [style=dot] (1) at (.5, .5) {};
\node [style=oneout] (2) at (1, .5) {};
\node [style=onein] (10) at (0, 0) {};
\node [style=oplus] (11) at (.5, 0) {};
\node [style=oneout] (12) at (1, 0) {};
\end{pgfonlayer}
\begin{pgfonlayer}{edgelayer}
\draw [style=rn]  (0) to (2) {};
\draw [style=rn]  (10) to (12) {};
\draw [style=rn]  (1) to (11) {};
\end{pgfonlayer}
\end{tikzpicture}
\end{flalign*}
\end{proof}

We may generalize $\Omega$ to an arbitrary domain and codomain:

\begin{definition}
\label{defn:Omeganm}

For any $n$ and $m$ in $\mathbb{N}$, define the degenerate circuit from $n$ to $m$, $\Omega_{n,m}:n\to m$ by the circuit $(\otimes^n\oneout)\Omega(\otimes^m\onein)$.

Note that $\otimes^n\_$ denotes the $n$-fold iterated tensor product.
\end{definition}

These circuits exhibit the following ``absorbing'' property similar to Axiom {\em (ZO)} in the ZX calculus \cite{ZX, coecke2011interacting}:

\begin{lemma}~
\label{lemma:tensorOmegaDegen}
\begin{enumerate}[label=(\roman*)]
\item If $f=f\otimes\Omega$, for some $f:n\to m$ then $f=\Omega_{n,m}$
\item If $g: m \rightarrow p$, then $\Omega_{n,m}g = \Omega_{n,p}$
\item If $h: p \rightarrow n$, then $h \Omega_{n,m} = \Omega_{p,n}$
\end{enumerate}
\end{lemma}
\begin{proof}\
\begin{enumerate}[label=(\roman*)]
\item Consider an arbitrary circuit $f:n\to m$ such that $f=f\otimes \Omega$. Use $\ref{CNTix}$ to cut the wires around every gate in $f$.  Then every cut gate must either be $\Omega$ or $\onein\oneout$.  In the first case, use  Lemma \ref{lem:omegaotimesomega}  to consume the $\Omega$ obtained by cutting.  In the second case, by \ref{CNTvi}, allow one to remove the circuit. Thus $f = f\otimes\Omega = \Omega_{n,m}$.
\item Clearly $\Omega_{n,m} \Omega = \Omega_{n,m} \otimes \Omega = \Omega_{n,m}$ but then $h \Omega_{n,m} = (h \Omega_{n,m}) \Omega$ so $h \Omega_{n,m} = \Omega_{p,m}$.
\item Dual to (ii).
\end{enumerate}
\end{proof}

Next, we prove some identities that will be used later to simplify proofs:

\begin{lemma} ~
\label{lemmas}
\begin{enumerate}[label=(\roman*)]

\item
\label{remark:CNT:10}

\begin{flalign*}
 & \ref{CNTiv}
\end{flalign*}

\end{enumerate}
\end{lemma}

\section{\texorpdfstring{$\CNOT$}{CNOT} is a Discrete Inverse Category}
\label{Appendix B}
In this section, we prove that $\CNOT$ is a discrete inverse category.  We show discreteness before establishing the inverse category properties.  

\subsection{Inverse products in \texorpdfstring{$\CNOT$}{CNOT}}

We begin by defining two families of maps $\Delta_n$ and $\nabla_n$ for all $n \in \N$. Then we show that $\Delta$ is coassociative and satisfies uniform copying law. $\Delta$ and show that $\Delta$ is a copy natural transformation which forms a cosemigroup such that the uniform copying law holds.
 
%The 0 ancilla is used to define $\Delta$:

\begin{definition}
\label{defn:fanoutFanin}
Define two families of maps $\{\Delta_n:n\to 2n\}_{n\in \mathbb{N}}$ and $\{\nabla_n:2n\to n\}_{n \in \mathbb{N}}$ as follows.

On zero wires, define $\Delta_0:=1_0$.

On one wire, define $\Delta_1$, $\nabla_1$, respectively by:
%\begin{figure}[H]
\begin{flalign*}
\begin{tikzpicture}[baseline={([yshift=-.5ex]current bounding box.center)}]
\begin{pgfonlayer}{nodelayer}
\node [style=rn] (0) at (0, 0.5) {};
\node [style=fanout] (1) at (1, 0.5) {};
\node [style=rn] (2) at (2, 1) {};
\node [style=rn] (3) at (2, 0) {};
\end{pgfonlayer}
\begin{pgfonlayer}{edgelayer}
\draw [style=simple] (0) to (1);
\draw [style=simple, bend left] (1) to (2);
\draw [style=simple, bend right] (1) to (3);
\end{pgfonlayer}
\end{tikzpicture}
:=
\begin{tikzpicture}[baseline={([yshift=-.5ex]current bounding box.center)}]
\begin{pgfonlayer}{nodelayer}
\node [style=zeroin] (0) at (0, .5) {};
\node [style=oplus] (1) at (.5, .5) {};
\node [style=rn] (2) at (1, .5) {};
\node [style=rn] (10) at (0, 0) {};
\node [style=dot] (11) at (.5, 0) {};
\node [style=rn] (12) at (1, 0) {};
\end{pgfonlayer}
\begin{pgfonlayer}{edgelayer}
\draw [style=simple] (0) to (2);
\draw [style=simple] (10) to (12);
\draw [style=simple] (1) to (11);
\end{pgfonlayer}
\end{tikzpicture}
\hspace*{.5cm}
\text{ and dually }
\hspace*{.5cm}
\begin{tikzpicture}[baseline={([yshift=-.5ex]current bounding box.center)}]
\begin{pgfonlayer}{nodelayer}
\node [style=rn] (0) at (2, 0.5) {};
\node [style=fanin] (1) at (1, 0.5) {};
\node [style=rn] (2) at (0, 1) {};
\node [style=rn] (3) at (0, 0) {};
\end{pgfonlayer}
\begin{pgfonlayer}{edgelayer}
\draw [style=simple] (0) to (1);
\draw [style=simple, bend right] (1) to (2);
\draw [style=simple, bend left] (1) to (3);
\end{pgfonlayer}
\end{tikzpicture}
:=
\begin{tikzpicture}[baseline={([yshift=-.5ex]current bounding box.center)}]
\begin{pgfonlayer}{nodelayer}
\node [style=rn] (0) at (0, .5) {};
\node [style=oplus] (1) at (.5, .5) {};
\node [style=zeroout] (2) at (1, .5) {};
\node [style=rn] (10) at (0, 0) {};
\node [style=dot] (11) at (.5, 0) {};
\node [style=rn] (12) at (1, 0) {};
\end{pgfonlayer}
\begin{pgfonlayer}{edgelayer}
\draw [style=simple] (0) to (2);
\draw [style=simple] (10) to (12);
\draw [style=simple] (1) to (11);
\end{pgfonlayer}
\end{tikzpicture}
\end{flalign*}
%\label{fig:fanoutBaseCase}
%\caption{Definition \ref{defn:fanoutFanin}: $\Delta$ and $\nabla$ on one wire, respectively}
%\end{figure}
On $n$ wires define $\Delta$, $\nabla$ inductively as follows:
%\begin{figure}[H]
\begin{flalign*}
\begin{tikzpicture}[baseline={([yshift=-.5ex]current bounding box.center)}]
\begin{pgfonlayer}{nodelayer}
\node [style=rn] (0) at (0, 0.5) {};
\node [style=fanout] (1) at (1, 0.5) {};
\node [style=rn] (2) at (2, 1) {};
\node [style=rn] (3) at (2, 0) {};
\end{pgfonlayer}
\begin{pgfonlayer}{edgelayer}
\draw [style=wires]  (0) -- (1) node[pos=.4, above] {$n$};
\draw [style=simple, bend left] (1) to (2);
\draw [style=simple, bend right] (1) to (3);
\end{pgfonlayer}
\end{tikzpicture}
&:=
\begin{tikzpicture}[baseline={([yshift=-.5ex]current bounding box.center)}]
\begin{pgfonlayer}{nodelayer}
\node [style=rn] (0) at (0, 0.5) {};
\node [style=fanout] (1) at (1, 0.5) {};
\node [style=rn] (2) at (2, 1) {};
\node [style=rn] (3) at (2, 0) {};
\node [style=rn] (10) at (0, 1.25) {};
\node [style=fanout] (11) at (1, 1.25) {};
\node [style=rn] (12) at (2, 1.75) {};
\node [style=rn] (13) at (2, 0.75) {};
\end{pgfonlayer}
\begin{pgfonlayer}{edgelayer}
\draw [style=wires]  (10) -- (11) node[pos=.4, above] {$n-1$};
\draw [style=simple, bend left] (1) to (2);
\draw [style=simple, bend right] (1) to (3);
\draw [style=simple] (0) to (1);
\draw [style=simple, bend left] (11) to (12);
\draw [style=simple, bend right] (11) to (13);
\end{pgfonlayer}
\end{tikzpicture}
\hspace*{.5cm}
\text{ and dually }
\hspace*{.5cm}
\begin{tikzpicture}[baseline={([yshift=-.5ex]current bounding box.center)}]
\begin{pgfonlayer}{nodelayer}
\node [style=rn] (0) at (2, 0.5) {};
\node [style=fanin] (1) at (1, 0.5) {};
\node [style=rn] (2) at (0, 1) {};
\node [style=rn] (3) at (0, 0) {};
\end{pgfonlayer}
\begin{pgfonlayer}{edgelayer}
\draw [style=wires]  (0) -- (1) node[pos=.4, above] {$n$};
\draw [style=simple, bend right] (1) to (2);
\draw [style=simple, bend left] (1) to (3);
\end{pgfonlayer}
\end{tikzpicture}
:=
\begin{tikzpicture}[baseline={([yshift=-.5ex]current bounding box.center)}]
\begin{pgfonlayer}{nodelayer}
\node [style=rn] (0) at (2, 0.5) {};
\node [style=fanin] (1) at (1, 0.5) {};
\node [style=rn] (2) at (0, 1) {};
\node [style=rn] (3) at (0, 0) {};
\node [style=rn] (10) at (2, 1.25) {};
\node [style=fanin] (11) at (1, 1.25) {};
\node [style=rn] (12) at (0, 1.75) {};
\node [style=rn] (13) at (0, 0.75) {};
\end{pgfonlayer}
\begin{pgfonlayer}{edgelayer}
\draw [style=wires]  (10) -- (11) node[pos=.4, above] {$n-1$};
\draw [style=simple, bend right] (1) to (2);
\draw [style=simple, bend left] (1) to (3);
\draw [style=simple] (0) to (1);
\draw [style=simple, bend right] (11) to (12);
\draw [style=simple, bend left] (11) to (13);
\end{pgfonlayer}
\end{tikzpicture}
\end{flalign*}
%\caption{Definition \ref{defn:fanoutFanin}: $\Delta$ and $\nabla$ on $n$ wires}
%\label{fig:fanout}
%\end{figure}
\end{definition}

As we have defined $\CNOT$ by its generators, most of the proofs which follow involve structural induction. 

\begin{definition}
For any circuit $f$ we define the \text{size of the circuit} $|f|$ as the number of $\cnot$ gates and ancill\ae\ in $f$.
\end{definition}

\begin{lemma}  \label{lem:fanoutNatural}
$\Delta$ is a natural transformation.
\end{lemma}

\begin{proof} We prove $\Delta$ is natural by a structural induction on circuits.

\begin{itemize}
\item
For $\onein$:

\begin{flalign*}

\end{flalign*}

Therefore, $\cnot = \Delta_2 (\cnot \otimes\cnot)$.
\end{itemize}
Therefore, $\Delta$ is natural by structural induction.
\end{proof}

\begin{lemma} \label{lemma:fanoutCommutative}
$\Delta$ is cocommutative.
\end{lemma}

\begin{proof}
We prove that $\Delta$ is cocommutative by induction on the number of wires:

\begin{itemize}
\item On no wires the result is immediate.
\item On one wire:
\begin{flalign*}
%
\end{flalign*}

\item Suppose inductively that $\Delta$ is cocommutative on up to $n$ wires. On $n+1$ wires:

\begin{flalign*}
%
 & \text{By the inductive hypothesis}\\
&=:
%
\end{flalign*}

Therefore, the inductive claim holds.
\end{itemize}
\end{proof}

\begin{lemma} \label{lemma:fanoutFaninInverses}
$\Delta$ is separable (or special).
\end{lemma}
\begin{proof}
We prove that $\Delta$ is separable by induction on the number of wires. 

\begin{itemize}

\item On zero wires it is the identity.

\item On one wire:

\begin{flalign*}
%
& \text{Lemma \ref{lemmas} \ref{lem:zerovanish}}\\
\end{flalign*}

Therefore, the base case on one wire holds.

\item Suppose inductively that $\Delta_n\nabla_n=1_n$ for some $n \in \mathbb{N}$. On $n+1$ wires:

\begin{flalign*}
%
& \text{By the inductive hypothesis}\\
&=
%
\\
\end{flalign*}

\end{itemize}
\end{proof}

\begin{lemma}\label{lemma:fanoutAssociative}
 $\Delta$ is coassociative.
\end{lemma}
\begin{proof}
We prove that $\Delta$ is coassociative by induction on the number of wires.
\begin{itemize}

\item On zero wires it is immediate.

\item On one wire:

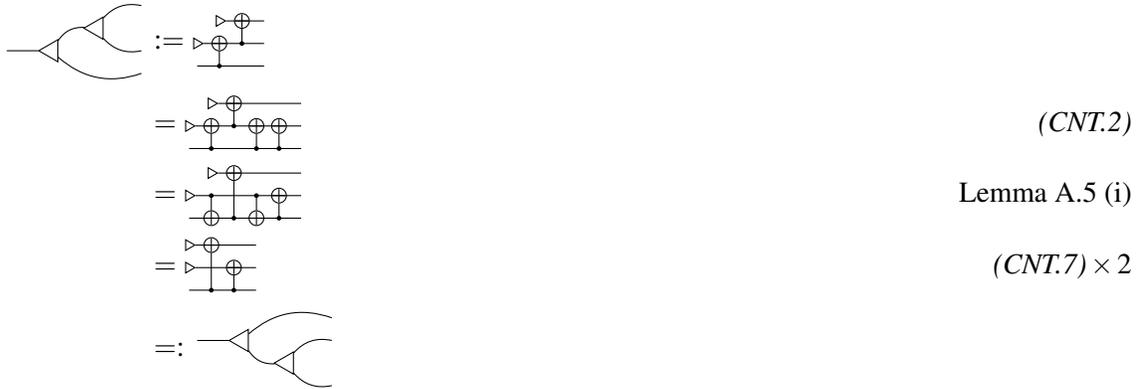
\begin{flalign*}
%
\end{flalign*}

Therefore, the base case on one wire holds.

\item Suppose that $\Delta$ is coassociative on up to $n$ wires. On $n+1$ wires: 

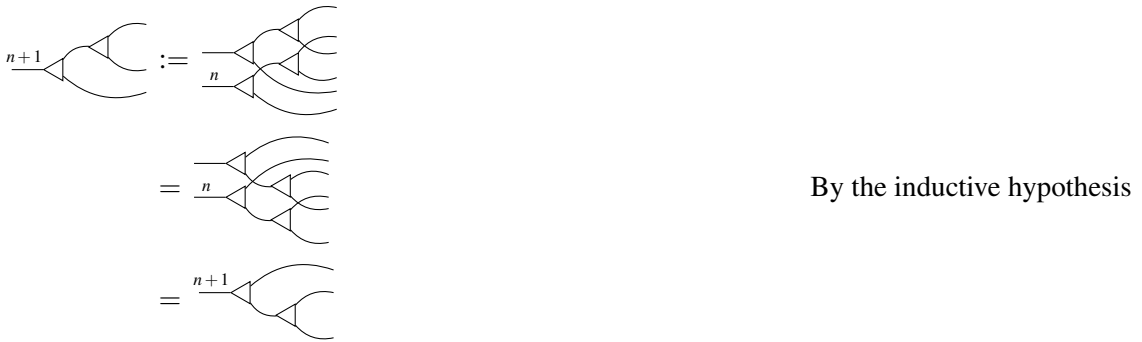
\begin{flalign*}
%
 & \text{By the inductive hypothesis}\\
&=
%
\end{flalign*}

Therefore, the inductive claim holds.
\end{itemize}
\end{proof}

The dual propositions for $\nabla$ hold since, $\Delta\cnv=\nabla$.

\begin{lemma} \label{lemma:fanoutSemiFrobenius}
$(n, \Delta_n, \nabla_n)$ is a semi-Frobenius algebra for all $n \in \N$.
\end{lemma}
\begin{proof}
We prove that $(\Delta, \nabla)$ is a semi-Frobenius algebra by induction on the number of wires.

\begin{itemize}
\item On one wire it is immediate.
\item On one wire:
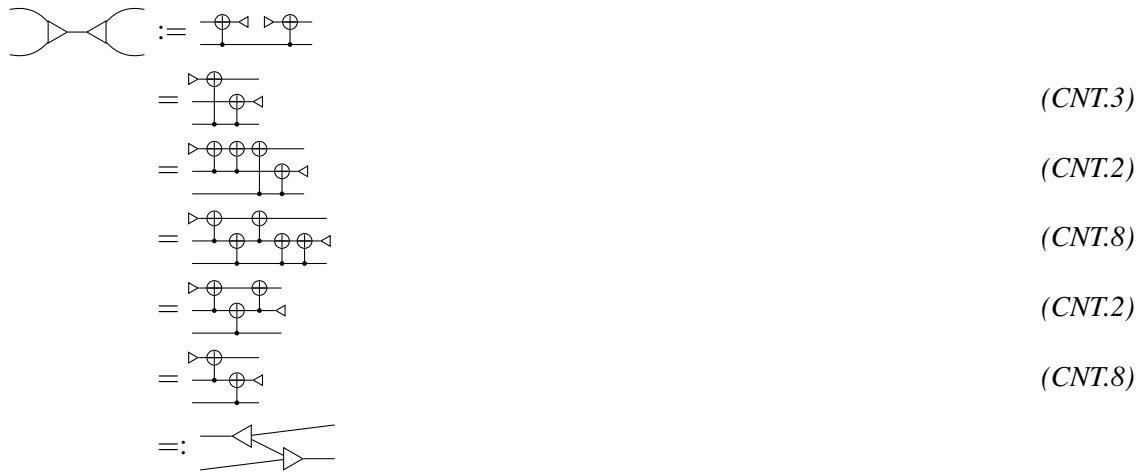
\begin{flalign*}
%
\end{flalign*}

Therefore, the base case on one wire holds.

\item Suppose inductively that $\Delta$ is semi-Frobenius algebra on up to $n$ wires. On $n+1$ wires:

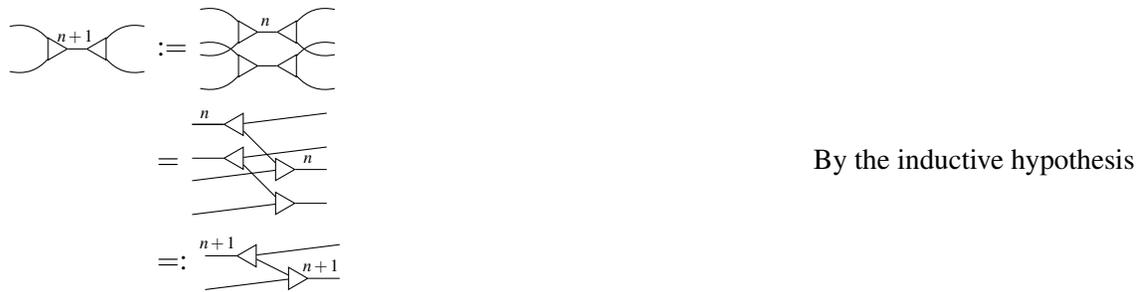
\begin{flalign*}
%
 & \text{By the inductive hypothesis}\\
&=:
%
\end{flalign*}

Therefore, the inductive claim holds.

\end{itemize}
\end{proof}

\begin{lemma} \label{lemma:fanoutCanonical} 
      \ref{DNV:4} holds for $\Delta$.
\end{lemma}
\begin{proof}
The uniform copying law holds for $\Delta$ by construction.
\end{proof}

%%%%%%%%%%%%%%%%%%%%%%%%%%%
\subsection{\texorpdfstring{$\CNOT$}{CNOT} is an inverse category}
%%%%%%%%%%%%%%%%%%%%%%%%%%

To prove that $\CNOT$ is an inverse category, we need to prove that the functor $(\_)\cnv: \CNOT\op \rightarrow \CNOT$ which horizontally flips circuits satisfies \ref{INV:1}, \ref{INV:2} and \ref{INV:3}. It is immediate that \ref{INV:1} holds. It remains to show that \ref{INV:2} and \ref{INV:3} hold.

%%%%%%%%%%%%%%%%%%%%%%%
\subsubsection{Latchable Circuits}
%%%%%%%%%%%%%%%%%%%%%%%

In order to prove that  \ref{INV:2} holds, we identify the restriction idempotents $ff^\circ$ in $\CNOT$ with what we call ``latchable circuits".  We then show that for every  $f$ that map $ff\cnv$ is latchable.

\begin{definition}%{Latchable circuit}
\label{defn:Latchable}

A circuit $f:n\to n$ is \bf{latchable} when
$$
f=\Delta_n (f \otimes 1_n) \nabla_n
$$
That is, when the following holds:
%\begin{figure}[H]
\begin{flalign*}

\end{flalign*}
%\caption{Definition \ref{defn:Latchable}:  Latchable circuit definition}
%\label{fig:latchable}
%\end{figure}

\end{definition}

In order to use Theorem \ref{defn:inverseCategory} to prove that $\CNOT$ is a discrete inverse category, we require the following:

\begin{lemma}{Latchable circuits commute and are idempotent.}
\label{prop:latchableCircuitsCommuteIdempotent}
\end{lemma}

\begin{proof}%[Proof of Lemma \ref{prop:latchableCircuitsCommuteIdempotent}]
Suppose two circuits $f,g: n \to n$ are latchable, then:
\begin{flalign*}
&%
 & \text{As $\Delta$ is commutative and cocommutative}\\
&=
%
 & \text{As $\Delta$ is associative and coassociative}\\
&=
%
 & \text{By symmetry}
\end{flalign*}\\

Therefore, latchable circuits commute.  Likewise, by a similar argument:

\begin{flalign*}
&%
 & \text{As $\Delta$ is commutative and cocommutative}\\
&=
%
 & \text{As $\Delta$ is associative and coassociative}\\
&=
%
 & \text{As $\Delta$ is separable}\\
\end{flalign*}

Therefore, latchable circuits are idempotent.
\end{proof}

Before we prove all circuits of the form $ff\cnv$ are latchable, we establish the following identity:

\begin{lemma}
\label{lem:symmetricCircuitsLatchable}
$\Delta_1((\oneout\onein)\otimes 1_1) \nabla_1$
\end{lemma}

\begin{proof}
We observe:

%\begin{flalign*}
%%
& \ref{CNTvi}\times 2
\end{flalign*}

\end{proof}

\begin{proposition}{All circuits of the form $ff\cnv$ are latchable.}
\label{prop:symmetricLatchable}
\end{proposition}
\begin{proof}
We shall prove that all circuits of the form $ff\cnv$ are latchable by induction on the size of  $f$.

Any circuit $p$ with  $|p|=0$ is a permutation and, thus, $pp\cnv$ is the identity. So being latchable in this case amounts to separability.
Furthermore, adding a permutation, $p$, in front of any latchable circuit, $h$, gives a latchable circuit as:
         $$(ph)(ph)\cnv = ph h\cnv p\cnv = p \Delta(hh\cnv\otimes 1)\nabla p\cnv = \Delta (phh\cnv p\cnv \otimes pp\cnv)\nabla = \Delta (((ph)(ph)\cnv)\otimes 1) \nabla$$
Thus we need only consider adding gates to the top left of circuits: adding a gate anywhere else can be simulated by precomposing with a permutation to move the 
gates wires to the top, then adding the gate at the top left, and then precomposing with the inverse of the permutation.

Thus, it suffices to show  inductively that when a circuit of the form $ff\cnv$ is latchable, for $|f|<k$, then adding a gate to the top left of $f$ results in a circuit which is still latchable.
\begin{description}
\item[Adding $|1 \rangle$:]
Given that $n\geq 1$, the symmetric circuit
$$
{h=((\onein\otimes 1_{n-1}) f)((\onein\otimes 1_{n-1}) f)\cnv:n-1\to n-1}
$$
is latchable, as:
\begin{flalign*}
&\begin{tikzpicture}[baseline={([yshift=-.5ex]current bounding box.center)}]
\begin{pgfonlayer}{nodelayer}
\node [style=rn] (0) at (0, 0) {};
\draw[fill=white] (2,0) circle [radius=.25] node (1) {$h$};
\node [style=rn] (2) at (4, 0) {};
\end{pgfonlayer}
\begin{pgfonlayer}{edgelayer}
\draw [style=wires]  (0) -- (1) node[pos=.4, above] {$n-1$};
\draw [style=wires]  (1) -- (2) node[pos=.4, above] {$n-1$};
\end{pgfonlayer}
\end{tikzpicture}\\
&:=
\begin{tikzpicture}[baseline={([yshift=-.5ex]current bounding box.center)}]
\begin{pgfonlayer}{nodelayer}
\node [style=rn] (0) at (0, 0) {};
\node [style=onein] (10) at (0, 1) {};
%\node [style=fanout] (11) at (1.5, 1.75) {};
%
\draw[fill=white] (2,0) circle [radius=.4] node (20) {$ff\cnv$};
\node [style=rn] (3) at (4, 0) {};
%
%\node [style=fanin] (12) at (5.5, 1.75) {};
\node [style=oneout] (13) at (4, 1) {};
\end{pgfonlayer}
\begin{pgfonlayer}{edgelayer}
\draw [style=wires]  (0) -- (20) node[pos=.4, above] {$n-1$};
\draw [style=wires]  (20) -- (3) node[pos=.4, above] {$n-1$};
\draw [style=simple, bend left] (10) to (20);
\draw [style=simple, bend left] (20) to (13);
\end{pgfonlayer}
\end{tikzpicture}& \text{By supposition}\\
&=
\begin{tikzpicture}[baseline={([yshift=-.5ex]current bounding box.center)}]
\begin{pgfonlayer}{nodelayer}
\node [style=rn] (0) at (0, 1) {};
\node [style=fanout] (1) at (1.5, 1) {};
\node [style=onein] (10) at (0, 1.75) {};
\node [style=fanout] (11) at (1.5, 1.75) {};
\draw[fill=white] (2.5,2) circle [radius=.4] node (20) {$ff\cnv$};
\node [style=fanin] (2) at (3.5, 1) {};
\node [style=rn] (3) at (5, 1) {};
\node [style=fanin] (12) at (3.5, 1.75) {};
\node [style=oneout] (13) at (5, 1.75) {};
\end{pgfonlayer}
\begin{pgfonlayer}{edgelayer}
\draw [style=wires]  (0) -- (1) node[pos=.4, above] {$n-1$};
\draw [style=wires]  (2) -- (3) node[pos=.4, above] {$n-1$};
\draw [style=simple] (1) to (20);
\draw [style=simple] (20) to (2);
\draw [style=simple,bend right] (1) to (2);
\draw [style=wires]  (10) -- (11) node[pos=.4, above] {};
\draw [style=wires]  (12) -- (13) node[pos=.4, above] {};
\draw [style=simple,bend left] (11) to (20);
\draw [style=simple,bend left] (20) to (12);
\draw [style=simple,bend right] (11) to (12);
\end{pgfonlayer}
\end{tikzpicture}& \text{By supposition}\\
&=
\begin{tikzpicture}[baseline={([yshift=-.5ex]current bounding box.center)}]
\begin{pgfonlayer}{nodelayer}
\node [style=rn] (0) at (0, 1) {};
\node [style=fanout] (1) at (1.5, 1) {};
\draw[fill=white] (2.5,2) circle [radius=.4] node (20) {$ff\cnv$};
\node [style=fanin] (2) at (3.5, 1) {};
\node [style=rn] (3) at (5, 1) {};
\node [style=onein] (11) at (1, 2) {};
\node [style=oneout] (12) at (4, 2) {};
\node [style=onein] (30) at (1, 1.5) {};
\node [style=oneout] (31) at (4, 1.5) {};
\end{pgfonlayer}
\begin{pgfonlayer}{edgelayer}
\draw [style=wires]  (0) -- (1) node[pos=.4, above] {$n-1$};
\draw [style=wires]  (2) -- (3) node[pos=.4, above] {$n-1$};
\draw [style=simple] (1) to (20);
\draw [style=simple] (20) to (2);
\draw [style=simple,bend right] (1) to (2);
\draw [style=simple] (11) to (20);
\draw [style=simple] (20) to (12);
\draw [style=simple] (30) to (31);
\end{pgfonlayer}
\end{tikzpicture}& \text{ Lemma } \ref{lem:fanoutNatural}\\
&=
\begin{tikzpicture}[baseline={([yshift=-.5ex]current bounding box.center)}]
\begin{pgfonlayer}{nodelayer}
\node [style=rn] (0) at (0, 1) {};
\node [style=fanout] (1) at (1.5, 1) {};
\draw[fill=white] (2.5,1.5) circle [radius=.45] node (20) {$ff\cnv$};
\node [style=fanin] (2) at (3.5, 1) {};
\node [style=rn] (3) at (5, 1) {};
\node [style=onein] (11) at (1.5, 2) {};
\node [style=oneout] (12) at (3.5, 2) {};
\end{pgfonlayer}
\begin{pgfonlayer}{edgelayer}
\draw [style=wires]  (0) -- (1) node[pos=.4, above] {$n-1$};
\draw [style=wires]  (2) -- (3) node[pos=.4, above] {$n-1$};
\draw [style=simple] (1) to (20);
\draw [style=simple] (20) to (2);
\draw [style=simple,bend right] (1) to (2);
\draw [style=simple] (11) to (20);
\draw [style=simple] (20) to (12);
\end{pgfonlayer}
\end{tikzpicture}& \ref{CNTvi}\\
&=:
\begin{tikzpicture}[baseline={([yshift=-.5ex]current bounding box.center)}]
\begin{pgfonlayer}{nodelayer}
\node [style=rn] (0) at (0, 1) {};
\node [style=fanout] (1) at (1.5, 1) {};
\draw[fill=white] (2.5,1.5) circle [radius=.25] node (20) {$h$};
\node [style=fanin] (2) at (3.5, 1) {};
\node [style=rn] (3) at (5, 1) {};
\end{pgfonlayer}
\begin{pgfonlayer}{edgelayer}
\draw [style=wires]  (0) -- (1) node[pos=.4, above] {$n-1$};
\draw [style=wires]  (2) -- (3) node[pos=.4, above] {$n-1$};
\draw [style=simple] (1) to (20);
\draw [style=simple] (20) to (2);
\draw [style=simple,bend right] (1) to (2);
\end{pgfonlayer}
\end{tikzpicture}\\
\end{flalign*}

\item[Adding $\langle 1 |$:]
Given that $n\geq 1$, the symmetric circuit
$$
h=((\oneout\otimes 1_{n}) f)((\oneout\otimes 1_{n}) f)\cnv:n+1\to n+1
$$
is latchable, as:
\begin{flalign*}
&\begin{tikzpicture}[baseline={([yshift=-.5ex]current bounding box.center)}]
\begin{pgfonlayer}{nodelayer}
\node [style=rn] (0) at (0, 0) {};
\draw[fill=white] (2,0) circle [radius=.25] node (1) {$h$};
\node [style=rn] (2) at (4, 0) {};
\end{pgfonlayer}
\begin{pgfonlayer}{edgelayer}
\draw [style=wires]  (0) -- (1) node[pos=.4, above] {$n+1$};
\draw [style=wires]  (1) -- (2) node[pos=.4, above] {$n+1$};
\end{pgfonlayer}
\end{tikzpicture}\\
&:=
\begin{tikzpicture}[baseline={([yshift=-.5ex]current bounding box.center)}]
\begin{pgfonlayer}{nodelayer}
\node [style=rn] (0) at (0, 0) {};
\node [style=rn] (10) at (0, .5) {};
\node [style=oneout] (11) at (1.25, .5) {};
\draw[fill=white] (2,0) circle [radius=.4] node (20) {$ff\cnv$};
\node [style=rn] (3) at (4, 0) {};
\node [style=onein] (12) at (2.75, .5) {};
\node [style=rn] (13) at (4, .5) {};
\end{pgfonlayer}
\begin{pgfonlayer}{edgelayer}
\draw [style=wires]  (0) -- (20) node[pos=.4, above] {$n$};
\draw [style=wires]  (20) -- (3) node[pos=.4, above] {$n$};
\draw [style=simple] (10) to (11);
\draw [style=simple] (12) to (13);
\end{pgfonlayer}
\end{tikzpicture}& \text{By supposition}\\
&=
\begin{tikzpicture}[baseline={([yshift=-.5ex]current bounding box.center)}]
\begin{pgfonlayer}{nodelayer}
\node [style=rn] (0) at (0, 1) {};
\node [style=fanout] (1) at (1.5, 1) {};
\node [style=rn] (10) at (0, 1.75) {};
\node [style=oneout] (11) at (1.5, 1.75) {};
\draw[fill=white] (2.5,1.5) circle [radius=.4] node (30) {$ff\cnv$};
\node [style=fanin] (2) at (3.5, 1) {};
\node [style=rn] (3) at (5, 1) {};
\node [style=onein] (12) at (3.5, 1.75) {};
\node [style=rn] (13) at (5, 1.75) {};
\end{pgfonlayer}
\begin{pgfonlayer}{edgelayer}
\draw [style=wires]  (0) -- (1) node[pos=.4, above] {$n$};
\draw [style=wires]  (2) -- (3) node[pos=.4, above] {$n$};
\draw [style=simple,bend left] (1) to (30);
\draw [style=simple,bend left] (30) to (2);
\draw [style=simple,bend right] (1) to (2);
\draw [style=simple] (10) to (11);
\draw [style=simple] (12) to (13);
\end{pgfonlayer}
\end{tikzpicture}& \text{By supposition}\\
&=
\begin{tikzpicture}[baseline={([yshift=-.5ex]current bounding box.center)}]
\begin{pgfonlayer}{nodelayer}
\node [style=rn] (0) at (0, 1) {};
\node [style=fanout] (1) at (1.5, 1) {};
\node [style=rn] (10) at (0, 1.75) {};
\node [style=fanout] (11) at (1.5, 1.75) {};
\node [style=oneout] (14) at (2.25, 2) {};
\node [style=oneout] (15) at (2.25, 1) {};
\draw[fill=white] (2.5,1.5) circle [radius=.4] node (30) {$ff\cnv$};
\node [style=fanin] (2) at (3.5, 1) {};
\node [style=rn] (3) at (5, 1) {};
\node [style=fanin] (12) at (3.5, 1.75) {};
\node [style=rn] (13) at (5, 1.75) {};
\node [style=onein] (16) at (2.75, 2) {};
\node [style=onein] (17) at (2.75, 1) {};
\end{pgfonlayer}
\begin{pgfonlayer}{edgelayer}
\draw [style=wires]  (0) -- (1) node[pos=.4, above] {$n$};
\draw [style=wires]  (2) -- (3) node[pos=.4, above] {$n$};
\draw [style=simple,bend left] (1) to (30);
\draw [style=simple,bend left] (30) to (2);
\draw [style=simple,bend right] (1) to (2);
\draw [style=simple] (10) to (11);
\draw [style=simple] (12) to (13);
\draw [style=simple] (11) to (14);
\draw [style=simple,bend right] (11) to (15);
\draw [style=simple] (12) to (16);
\draw [style=simple,bend left] (12) to (17);
\end{pgfonlayer}
\end{tikzpicture}&\text{ As $\Delta$ is natural}\\
&=
\begin{tikzpicture}[baseline={([yshift=-.5ex]current bounding box.center)}]
\begin{pgfonlayer}{nodelayer}
\node [style=rn] (0) at (0, 1) {};
\node [style=fanout] (1) at (1.5, 1) {};
\node [style=rn] (10) at (0, 1.75) {};
\node [style=fanout] (11) at (1.5, 1.75) {};
\node [style=oneout] (14) at (2.25, 2) {};
\draw[fill=white] (2.5,1.5) circle [radius=.4] node (30) {$ff\cnv$};
\node [style=fanin] (2) at (3.5, 1) {};
\node [style=rn] (3) at (5, 1) {};
\node [style=fanin] (12) at (3.5, 1.75) {};
\node [style=rn] (13) at (5, 1.75) {};
\node [style=onein] (16) at (2.75, 2) {};
\node [style=rn] (17) at (2.5, 1) {};
\end{pgfonlayer}
\begin{pgfonlayer}{edgelayer}
\draw [style=wires]  (0) -- (1) node[pos=.4, above] {$n$};
\draw [style=wires]  (2) -- (3) node[pos=.4, above] {$n$};
\draw [style=simple,bend left] (1) to (30);
\draw [style=simple,bend left] (30) to (2);
\draw [style=simple,bend right] (1) to (2);
\draw [style=simple] (10) to (11);
\draw [style=simple] (12) to (13);
\draw [style=simple] (11) to (14);
\draw [style=simple,bend right] (11) to (17) to (12);
\draw [style=simple] (12) to (16);
\end{pgfonlayer}
\end{tikzpicture} & \text{Lemma } \ref{lem:symmetricCircuitsLatchable}\\
&=:
\begin{tikzpicture}[baseline={([yshift=-.5ex]current bounding box.center)}]
\begin{pgfonlayer}{nodelayer}
\node [style=rn] (0) at (0, 1) {};
\node [style=fanout] (1) at (1.5, 1) {};
\draw[fill=white] (2.5,1.5) circle [radius=.25] node (20) {$h$};
\node [style=fanin] (2) at (3.5, 1) {};
\node [style=rn] (3) at (5, 1) {};
\end{pgfonlayer}
\begin{pgfonlayer}{edgelayer}
\draw [style=wires]  (0) -- (1) node[pos=.4, above] {$n+1$};
\draw [style=wires]  (2) -- (3) node[pos=.4, above] {$n+1$};
\draw [style=simple] (1) to (20);
\draw [style=simple] (20) to (2);
\draw [style=simple,bend right] (1) to (2);
\end{pgfonlayer}
\end{tikzpicture}
\end{flalign*}

\item[Adding $\mathsf{cnot}$:]
Given that $n\geq 2$, the symmetric circuit 
$$
h=((\cnot\otimes 1_{n-2}) f)((\cnot\otimes 1_{n-2}) f)\cnv:n\to n
$$
is latchable, as:
\begin{flalign*}
&\begin{tikzpicture}[baseline={([yshift=-.5ex]current bounding box.center)}]
\begin{pgfonlayer}{nodelayer}
\node [style=rn] (0) at (0, 0) {};
\draw[fill=white] (2,0) circle [radius=.25] node (1) {$h\cnv$};
\node [style=rn] (2) at (4, 0) {};
\end{pgfonlayer}
\begin{pgfonlayer}{edgelayer}
\draw [style=wires]  (0) -- (1) node[pos=.4, above] {$n$};
\draw [style=wires]  (1) -- (2) node[pos=.4, above] {$n$};
\end{pgfonlayer}
\end{tikzpicture}\\
&:=
\begin{tikzpicture}[baseline={([yshift=-.5ex]current bounding box.center)}]
\begin{pgfonlayer}{nodelayer}
\node [style=rn] (0) at (0, 1.4) {};
\node [style=oplus] (1) at (1, 1.4) {};
\node [style=oplus] (2) at (3.5, 1.4) {};
\node [style=rn] (3) at (4.5, 1.4) {};
\node [style=rn] (10) at (0, 1) {};
\node [style=dot] (11) at (1, 1) {};
\node [style=dot] (12) at (3.5, 1) {};
\node [style=rn] (13) at (4.5, 1) {};
\node [style=rn] (20) at (0, 0) {};
\node [style=rn] (23) at (4.5, 0) {};
\draw[fill=white] (2.25,0) circle [radius=.45] node (100) {$ff\cnv$};
\end{pgfonlayer}
\begin{pgfonlayer}{edgelayer}
\draw [style=simple] (0) to (1);
\draw [style=simple] (2) to (3);
\draw [style=simple, bend left] (1) to (100);
\draw [style=simple, bend left] (100) to (2);
\draw [style=simple] (10) to (11);
\draw [style=simple] (12) to (13);
\draw [style=simple] (11) to (100);
\draw [style=simple] (100) to (12);
\draw [style=wires]  (20) -- (100) node[pos=.4, above] {$n-2$};
\draw [style=wires]  (100) -- (23) node[pos=.4, above] {$n-2$};
\draw [style=simple] (20) to (100);
\draw [style=simple] (100) to (23);
\draw [style=simple] (1) to (11);
\draw [style=simple] (2) to (12);
\end{pgfonlayer}
\end{tikzpicture}& \text{By supposition}\\
&=
\begin{tikzpicture}[baseline={([yshift=-.5ex]current bounding box.center)}]
\begin{pgfonlayer}{nodelayer}
\node [style=rn] (-1) at (-.5, 1) {};
\node [style=dot] (0) at (0, 1) {};
\node [style=fanout] (1) at (1.5, 1) {};
\node [style=rn] (-10) at (-.5, 1.75) {};
\node [style=oplus] (10) at (0, 1.75) {};
\node [style=fanout] (11) at (1.5, 1.75) {};
\draw[fill=white] (3,2) circle [radius=.45] node (20) {$ff\cnv$};
\node [style=fanin] (2) at (4.5, 1) {};
\node [style=dot] (3) at (6, 1) {};
\node [style=rn] (4) at (6.5, 1) {};
\node [style=fanin] (12) at (4.5, 1.75) {};
\node [style=oplus] (13) at (6, 1.75) {};
\node [style=rn] (14) at (6.5, 1.75) {};
\node [style=rn] (31) at (-.5, 0) {};
\node [style=fanout] (32) at (1.5, 0) {};
\node [style=fanin] (33) at (4.5, 0) {};
\node [style=rn] (34) at (6.5, 0) {};
\end{pgfonlayer}
\begin{pgfonlayer}{edgelayer}
\draw [style=wires]  (0) -- (1) node[pos=.4, above] {};
\draw [style=wires]  (2) -- (3) node[pos=.4, above] {};
\draw [style=simple] (1) to (20);
\draw [style=simple] (20) to (2);
\draw [style=simple,bend right] (1) to (2);
\draw [style=wires]  (10) -- (11) node[pos=.4, above] {};
\draw [style=wires]  (12) -- (13) node[pos=.4, above] {};
\draw [style=simple,bend left] (11) to (20);
\draw [style=simple,bend left] (20) to (12);
\draw [style=simple,bend right] (11) to (12);
\draw [style=simple] (0) to (10);
\draw [style=simple] (3) to (13);
\draw [style=simple] (-1) to (0);
\draw [style=simple] (-10) to (10);
\draw [style=simple] (3) to (4);
\draw [style=simple] (13) to (14);
\draw [style=wires]  (31) -- (32) node[pos=.4, above] {$n-2$};
\draw [style=wires]  (33) -- (34) node[pos=.4, above] {$n-2$};
\draw [style=simple,bend right] (32) to (33);
\draw [style=simple] (32) to (20);
\draw [style=simple] (20) to (33);
\end{pgfonlayer}
\end{tikzpicture}& \text{By supposition}\\
&=
\begin{tikzpicture}[baseline={([yshift=-.5ex]current bounding box.center)}]
\begin{pgfonlayer}{nodelayer}
\node [style=rn] (0) at (0, 1) {};
\node [style=fanout] (1) at (1.5, 1) {};
\node [style=rn] (10) at (0, 1.75) {};
\node [style=fanout] (11) at (1.5, 1.75) {};
\draw[fill=white] (4,2.25) circle [radius=.45] node (20) {$ff\cnv$};
\node [style=fanin] (2) at (6.5, 1) {};
\node [style=rn] (3) at (8, 1) {};
\node [style=fanin] (12) at (6.5, 1.75) {};
\node [style=rn] (13) at (8, 1.75) {};
\node [style=dot] (30) at (3.5, .5) {};
\node [style=dot] (31) at (4.5, .5) {};
\node [style=oplus] (32) at (3.5, 1.1) {};
\node [style=oplus] (33) at (4.5, 1.1) {};
\node [style=dot] (34) at (2.5, 1.75) {};
\node [style=dot] (35) at (5.5, 1.75) {};
\node [style=oplus] (36) at (2.5, 2.25) {};
\node [style=oplus] (37) at (5.5, 2.25) {};
\node [style=rn] (41) at (0, 0) {};
\node [style=fanout] (42) at (1.5, 0) {};
\node [style=fanin] (43) at (6.5, 0) {};
\node [style=rn] (44) at (8, 0) {};
\end{pgfonlayer}
\begin{pgfonlayer}{edgelayer}
\draw [style=wires]  (0) -- (1) node[pos=.4, above] {};
\draw [style=wires]  (2) -- (3) node[pos=.4, above] {};
\draw [style=simple] (1) to (34) to (20);
\draw [style=simple] (20) to (35) to (2);
\draw [style=simple] (1) to (30) to (31) to (2);
\draw [style=wires]  (10) -- (11) node[pos=.4, above] {};
\draw [style=wires]  (12) -- (13) node[pos=.4, above] {};
\draw [style=simple] (11) to (36) to (20);
\draw [style=simple] (20) to (37) to (12);
\draw [style=simple] (11) to (32) to (33) to (12);
\draw [style=simple] (34) to (36);
\draw [style=simple] (35) to (37);
\draw [style=simple] (30) to (32);
\draw [style=simple] (31) to (33);
\draw [style=wires]  (41) -- (42) node[pos=.4, above] {$n-2$};
\draw [style=wires]  (43) -- (44) node[pos=.4, above] {$n-2$};
\draw [style=simple,bend right] (42) to (43);
\draw [style=simple] (42) to (20);
\draw [style=simple] (20) to (43);
\end{pgfonlayer}
\end{tikzpicture} & \text{As $\Delta$ is natural}\\
&=
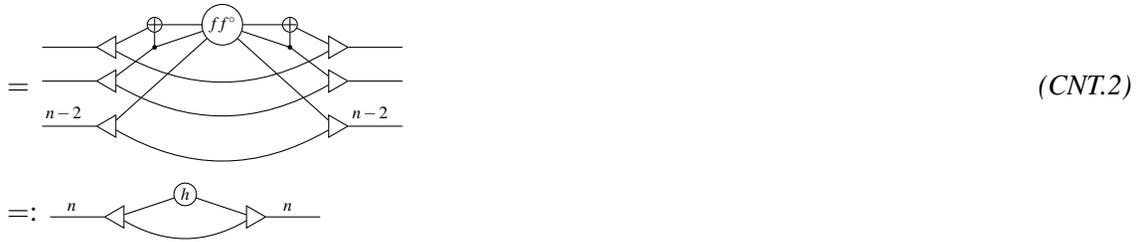
\begin{tikzpicture}[baseline={([yshift=-.5ex]current bounding box.center)}]
\begin{pgfonlayer}{nodelayer}
\node [style=rn] (0) at (0, 1) {};
\node [style=fanout] (1) at (1.5, 1) {};
\node [style=rn] (10) at (0, 1.75) {};
\node [style=fanout] (11) at (1.5, 1.75) {};
\draw[fill=white] (4,2.25) circle [radius=.45] node (20) {$ff\cnv$};
\node [style=fanin] (2) at (6.5, 1) {};
\node [style=rn] (3) at (8, 1) {};
\node [style=fanin] (12) at (6.5, 1.75) {};
\node [style=rn] (13) at (8, 1.75) {};
%
%\node [style=dot] (30) at (3, .5) {};
%\node [style=dot] (31) at (6, .5) {};
%
%\node [style=oplus] (32) at (3, 1.1) {};
%\node [style=oplus] (33) at (6, 1.1) {};
%
\node [style=dot] (34) at (2.5, 1.75) {};
\node [style=dot] (35) at (5.5, 1.75) {};
\node [style=oplus] (36) at (2.5, 2.25) {};
\node [style=oplus] (37) at (5.5, 2.25) {};
\node [style=rn] (41) at (0, 0) {};
\node [style=fanout] (42) at (1.5, 0) {};
\node [style=fanin] (43) at (6.5, 0) {};
\node [style=rn] (44) at (8, 0) {};
\end{pgfonlayer}
\begin{pgfonlayer}{edgelayer}
\draw [style=wires]  (0) -- (1) node[pos=.4, above] {};
\draw [style=wires]  (2) -- (3) node[pos=.4, above] {};
\draw [style=simple] (1) to (34) to (20);
\draw [style=simple] (20) to (35) to (2);
%\draw [style=simple] (1) to (30) to (31) to (2);
\draw [style=simple,bend right] (1) to (2);
\draw [style=wires]  (10) -- (11) node[pos=.4, above] {};
\draw [style=wires]  (12) -- (13) node[pos=.4, above] {};
\draw [style=simple] (11) to (36) to (20);
\draw [style=simple] (20) to (37) to (12);
%\draw [style=simple] (11) to (32) to (33) to (12);
\draw [style=simple,bend right] (11) to (12);
\draw [style=simple] (34) to (36);
\draw [style=simple] (35) to (37);
\draw [style=wires]  (41) -- (42) node[pos=.4, above] {$n-2$};
\draw [style=wires]  (43) -- (44) node[pos=.4, above] {$n-2$};
\draw [style=simple,bend right] (42) to (43);
\draw [style=simple] (42) to (20);
\draw [style=simple] (20) to (43);
\end{pgfonlayer}
\end{tikzpicture} & \ref{CNTii}\\
&=:
\begin{tikzpicture}[baseline={([yshift=-.5ex]current bounding box.center)}]
\begin{pgfonlayer}{nodelayer}
\node [style=rn] (0) at (0, 1) {};
\node [style=fanout] (1) at (1.5, 1) {};
\draw[fill=white] (3,1.5) circle [radius=.25] node (20) {$h$};
\node [style=fanin] (2) at (4.5, 1) {};
\node [style=rn] (3) at (6, 1) {};
\end{pgfonlayer}
\begin{pgfonlayer}{edgelayer}
\draw [style=wires]  (0) -- (1) node[pos=.4, above] {$n$};
\draw [style=wires]  (2) -- (3) node[pos=.4, above] {$n$};
\draw [style=simple] (1) to (20);
\draw [style=simple] (20) to (2);
\draw [style=simple,bend right] (1) to (2);
\end{pgfonlayer}
\end{tikzpicture}
\end{flalign*}
\end{description}

Therefore, every circuit of the form $hh\cnv$ is latchable when $|h|=k+1$ completeing the inductive step.
\end{proof}

This allows us to prove that \ref{INV:3} holds:

\begin{proposition} \label{cor:symmetricCommute}
  Circuits of the form $ff\cnv$ commute.
\end{proposition}
\begin{proof}  This is immediate from Proposition \ref{prop:symmetricLatchable} and Lemma \ref{prop:latchableCircuitsCommuteIdempotent}.
\end{proof}

It remains to prove that \ref{INV:2} holds which we prove by induction. 

\begin{lemma} 
\label{prop:inv3impinv2} \ref{INV:2} holds in $\CNOT$ under the functor $(\_)\cnv: \CNOT\op \rightarrow \CNOT$.
\end{lemma}
\begin{proof}
 We will prove \ref{INV:2} holds under the functor $(\_)\cnv$ by induction on the size of circuits.
\begin{itemize}
\item Take $p:n\to n$ to be a circuit with $|p|=0$.  Then ${pp\cnv p = p}$, since $p$ is a permutation.

\item Suppose inductively that \ref{INV:2} holds for all circuits with up to size strictly lesser than $k$.  Consider an arbitrary circuit $f:n\to m$ such that $|f|=k$.  We proceed by cases:

\begin{itemize}
\item Suppose that $f=(\oneout\otimes 1_{n-1}) g$.  Then:

\begin{flalign*}
&\begin{tikzpicture}[baseline={([yshift=-.5ex]current bounding box.center)}]
\begin{pgfonlayer}{nodelayer}
\node [style=rn] (0) at (0, 0) {};
\draw[fill=white] (1,0) circle [radius=.25] node (1) {$f$};
\draw[fill=white] (2,0) circle [radius=.25] node (2) {$f\cnv$};
\draw[fill=white] (3,0) circle [radius=.25] node (3) {$f$};
\node [style=rn] (4) at (4, 0) {};
\end{pgfonlayer}
\begin{pgfonlayer}{edgelayer}
\draw [style=wires]  (0) -- (1) node[pos=.4, above] {$n$};
\draw [style=wires]  (1) -- (2) node[pos=.4, above] {$m$};
\draw [style=wires]  (2) -- (3) node[pos=.4, above] {$n$};
\draw [style=wires]  (3) -- (4) node[pos=.4, above] {$m$};
\end{pgfonlayer}
\end{tikzpicture}\\
&=
\begin{tikzpicture}[baseline={([yshift=-.5ex]current bounding box.center)}]
\begin{pgfonlayer}{nodelayer}
\node [style=rn] (0) at (0, 0) {};
\node [style=oneout] (1) at (1, 0) {};
%\draw (1.5,.5) circle [radius=.4] node (2) {$gg\cnv$};
\draw[fill=white] (1.5,.5) circle [radius=.4] node (2) {$gg\cnv$};
\node [style=onein] (3) at (2, 0) {};
\node [style=oneout] (4) at (2.5, 0) {};
\draw[fill=white] (3,.5) circle [radius=.25] node (5) {$g$};
\node [style=rn] (6) at (4, 0.5) {};
\node [style=rn] (10) at (0, .5) {};
\end{pgfonlayer}
\begin{pgfonlayer}{edgelayer}
\draw [style=simple] (0) to (1);
\draw [style=simple] (3) to (4);
\draw [style=wires]  (10) -- (2) node[pos=.4, above] {$n-1$};
\draw [style=wires]  (2) -- (5) node[pos=.4, above] {$n-1$};
\draw [style=wires]  (5) -- (6) node[pos=.4, above] {$m$};
\end{pgfonlayer}
\end{tikzpicture} & \text{By supposition}\\
&=
\begin{tikzpicture}[baseline={([yshift=-.5ex]current bounding box.center)}]
\begin{pgfonlayer}{nodelayer}
\node [style=rn] (0) at (0, 0) {};
\node [style=oneout] (1) at (1, 0) {};
\draw[fill=white] (1.5,.5) circle [radius=.4] node (2) {$gg\cnv$};
\draw[fill=white] (3,.5) circle [radius=.25] node (5) {$g$};
\node [style=rn] (6) at (4, 0.5) {};
\node [style=rn] (10) at (0, .5) {};
\end{pgfonlayer}
\begin{pgfonlayer}{edgelayer}
\draw [style=simple] (0) to (1);
\draw [style=wires]  (10) -- (2) node[pos=.4, above] {$n-1$};
\draw [style=wires]  (2) -- (5) node[pos=.4, above] {$n-1$};
\draw [style=wires]  (5) -- (6) node[pos=.4, above] {$m$};
\end{pgfonlayer}
\end{tikzpicture} & \ref{CNTvi}\\
&=
\begin{tikzpicture}[baseline={([yshift=-.5ex]current bounding box.center)}]
\begin{pgfonlayer}{nodelayer}
\node [style=rn] (0) at (0, 0) {};
\node [style=oneout] (1) at (1, 0) {};
\draw[fill=white] (1.5,.5) circle [radius=.25] node (2) {$g$};
\node [style=rn] (6) at (3, 0.5) {};
\node [style=rn] (10) at (0, .5) {};
\end{pgfonlayer}
\begin{pgfonlayer}{edgelayer}
\draw [style=simple] (0) to (1);
\draw [style=wires]  (10) -- (2) node[pos=.4, above] {$n-1$};
\draw [style=wires]  (2) -- (6) node[pos=.4, above] {$m$};
\end{pgfonlayer}
\end{tikzpicture} & \text{By the inductive hypothesis}\\
&=
\begin{tikzpicture}[baseline={([yshift=-.5ex]current bounding box.center)}]
\begin{pgfonlayer}{nodelayer}
\draw[fill=white] (1.5,.5) circle [radius=.25] node (2) {$f$};
\node [style=rn] (6) at (3, 0.5) {};
\node [style=rn] (10) at (0, .5) {};
\end{pgfonlayer}
\begin{pgfonlayer}{edgelayer}
\draw [style=wires]  (10) -- (2) node[pos=.4, above] {$n$};
\draw [style=wires]  (2) -- (6) node[pos=.4, above] {$m$};
\end{pgfonlayer}
\end{tikzpicture} & \text{By supposition}\\
\end{flalign*}

\item
Suppose that $f=(\onein \otimes 1_{n}) g$.  Then:

\begin{flalign*}
&\begin{tikzpicture}[baseline={([yshift=-.5ex]current bounding box.center)}]
\begin{pgfonlayer}{nodelayer}
\node [style=rn] (0) at (0, 0) {};
\draw[fill=white] (1,0) circle [radius=.25] node (1) {$f$};
\draw[fill=white] (2,0) circle [radius=.25] node (2) {$f\cnv$};
\draw[fill=white] (3,0) circle [radius=.25] node (3) {$f$};
\node [style=rn] (4) at (4, 0) {};
\end{pgfonlayer}
\begin{pgfonlayer}{edgelayer}
\draw [style=wires]  (0) -- (1) node[pos=.4, above] {$n$};
\draw [style=wires]  (1) -- (2) node[pos=.4, above] {$m$};
\draw [style=wires]  (2) -- (3) node[pos=.4, above] {$n$};
\draw [style=wires]  (3) -- (4) node[pos=.4, above] {$m$};
\end{pgfonlayer}
\end{tikzpicture}\\
&=
\begin{tikzpicture}[baseline={([yshift=-.5ex]current bounding box.center)}]
\begin{pgfonlayer}{nodelayer}
\node [style=onein] (1) at (0.25, 0) {};
\draw[fill=white] (1.5,0.5) circle [radius=.4] node (2) {$gg\cnv$};
\node [style=oneout] (3) at (2.5, 0) {};
\node [style=onein] (4) at (3, 0) {};
\draw[fill=white] (4,0.5) circle [radius=.25] node (5) {$g$};
\node [style=rn] (6) at (5.5, 0.5) {};
\node [style=rn] (10) at (0, .5) {};
\end{pgfonlayer}
\begin{pgfonlayer}{edgelayer}
\draw [style=simple, bend right] (1) to (2);
\draw [style=simple, bend right] (2) to (3);
\draw [style=simple, bend right] (4) to (5);
\draw [style=wires]  (10) -- (2) node[pos=.4, above] {$n-1$};
\draw [style=wires]  (2) -- (5) node[pos=.4, above] {$n-1$};
\draw [style=wires]  (5) -- (6) node[pos=.4, above] {$m$};
\end{pgfonlayer}
\end{tikzpicture} & \text{By supposition}\\
&=
\begin{tikzpicture}[baseline={([yshift=-.5ex]current bounding box.center)}]
\begin{pgfonlayer}{nodelayer}
\draw[fill=white] (2,0.5) circle [radius=.4] node (2) {$gg\cnv g$};
\node [style=rn] (6) at (4, 0.5) {};
\node [style=rn] (10) at (0, .5) {};
\node [style=onein] (20) at (0, 0) {};
\node [style=oneout] (21) at (.5, 0) {};
\node [style=onein] (22) at (1, 0) {};
\end{pgfonlayer}
\begin{pgfonlayer}{edgelayer}
\draw [style=wires]  (10) -- (2) node[pos=.4, above] {$n-1$};
\draw [style=wires]  (2) -- (6) node[pos=.4, above] {$m$};
\draw [style=simple] (20) to (21);
\draw [style=simple, bend right] (22) to (2);
\end{pgfonlayer}
\end{tikzpicture} & \ref{INV:3}\\
&=
\begin{tikzpicture}[baseline={([yshift=-.5ex]current bounding box.center)}]
\begin{pgfonlayer}{nodelayer}
\draw[fill=white] (2,0.5) circle [radius=.4] node (2) {$gg\cnv g$};
\node [style=rn] (6) at (4, 0.5) {};
\node [style=rn] (10) at (0, .5) {};
\node [style=onein] (22) at (1, 0) {};
\end{pgfonlayer}
\begin{pgfonlayer}{edgelayer}
\draw [style=wires]  (10) -- (2) node[pos=.4, above] {$n-1$};
\draw [style=wires]  (2) -- (6) node[pos=.4, above] {$m$};
\draw [style=simple, bend right] (22) to (2);
\end{pgfonlayer}
\end{tikzpicture} & \ref{CNTvi}\\
&=
\begin{tikzpicture}[baseline={([yshift=-.5ex]current bounding box.center)}]
\begin{pgfonlayer}{nodelayer}
\draw[fill=white] (2,0.5) circle [radius=.25] node (2) {$g$};
\node [style=rn] (6) at (4, 0.5) {};
\node [style=rn] (10) at (0, .5) {};
\node [style=onein] (22) at (1, 0) {};
\end{pgfonlayer}
\begin{pgfonlayer}{edgelayer}
\draw [style=wires]  (10) -- (2) node[pos=.4, above] {$n-1$};
\draw [style=wires]  (2) -- (6) node[pos=.4, above] {$m$};
\draw [style=simple, bend right] (22) to (2);
\end{pgfonlayer}
\end{tikzpicture}& \text{By the inductive hypothesis}\\
&=
\begin{tikzpicture}[baseline={([yshift=-.5ex]current bounding box.center)}]
\begin{pgfonlayer}{nodelayer}
\draw[fill=white] (1,0) circle [radius=.25] node (2) {$f$};
\node [style=rn] (6) at (0, 0) {};
\node [style=rn] (10) at (2, 0) {};
\end{pgfonlayer}
\begin{pgfonlayer}{edgelayer}
\draw [style=wires]  (10) -- (2) node[pos=.4, above] {$m$};
\draw [style=wires]  (2) -- (6) node[pos=.4, above] {$n$};
\end{pgfonlayer}
\end{tikzpicture}& \text{By supposition}\\
\end{flalign*}

\item
Suppose that $f=(\cnot \otimes 1_{n-2} )g$.    Then:
\begin{flalign*}
&\begin{tikzpicture}[baseline={([yshift=-.5ex]current bounding box.center)}]
\begin{pgfonlayer}{nodelayer}
\node [style=rn] (0) at (0, 0) {};
\draw[fill=white] (1,0) circle [radius=.25] node (1) {$f$};
\draw[fill=white] (2,0) circle [radius=.25] node (2) {$f\cnv$};
\draw[fill=white] (3,0) circle [radius=.25] node (3) {$f$};
\node [style=rn] (4) at (4, 0) {};
\end{pgfonlayer}
\begin{pgfonlayer}{edgelayer}
\draw [style=wires]  (0) -- (1) node[pos=.4, above] {$n$};
\draw [style=wires]  (1) -- (2) node[pos=.4, above] {$m$};
\draw [style=wires]  (2) -- (3) node[pos=.4, above] {$n$};
\draw [style=wires]  (3) -- (4) node[pos=.4, above] {$m$};
\end{pgfonlayer}
\end{tikzpicture}\\
&=
\begin{tikzpicture}[baseline={([yshift=-.5ex]current bounding box.center)}]
\begin{pgfonlayer}{nodelayer}
\node [style=rn] (0) at (0, 0) {};
\node [style=rn] (20) at (0, -.5) {};
\node [style=oplus] (1) at (0.25, 0) {};
\node [style=dot] (21) at (0.25, -.5) {};
\draw[fill=white] (1.5,0.5) circle [radius=.4] node (2) {$gg\cnv$};
\node [style=oplus] (3) at (2.5, 0) {};
\node [style=dot] (23) at (2.5, -.5) {};
\node [style=oplus] (4) at (3, 0) {};
\node [style=dot] (24) at (3, -.5) {};
\draw[fill=white] (4,0.5) circle [radius=.25] node (5) {$g$};
\node [style=rn] (6) at (5.5, 0.5) {};
\node [style=rn] (10) at (0, .5) {};
\end{pgfonlayer}
\begin{pgfonlayer}{edgelayer}
\draw [style=simple, bend right] (1) to (2);
\draw [style=simple, bend right] (2) to (3);
\draw [style=simple, bend right] (4) to (5);
\draw [style=wires]  (10) -- (2) node[pos=.4, above] {$n-2$};
\draw [style=wires]  (2) -- (5) node[pos=.4, above] {$n-2$};
\draw [style=wires]  (5) -- (6) node[pos=.4, above] {$m$};
\draw [style=simple, bend right] (21) to (2);
\draw [style=simple, bend right] (2) to (23);
\draw [style=simple, bend right] (24) to (5);
\draw [style=simple] (0) to (1);
\draw [style=simple] (20) to (21);
\draw [style=simple] (1) to (21);
\draw [style=simple] (3) to (23);
\draw [style=simple] (4) to (24);
\draw [style=simple] (3) to (4);
\draw [style=simple] (23) to (24);
\end{pgfonlayer}
\end{tikzpicture} & \text{By supposition}\\
&=
\begin{tikzpicture}[baseline={([yshift=-.5ex]current bounding box.center)}]
\begin{pgfonlayer}{nodelayer}
\node [style=rn] (0) at (0, 0) {};
\node [style=rn] (20) at (0, -.5) {};
\node [style=oplus] (1) at (0.25, 0) {};
\node [style=dot] (21) at (0.25, -.5) {};
\draw[fill=white] (1.5,0.5) circle [radius=.45] node (2) {$gg\cnv g$};
\node [style=rn] (6) at (3, 0.5) {};
\node [style=rn] (10) at (0, .5) {};
\end{pgfonlayer}
\begin{pgfonlayer}{edgelayer}
\draw [style=simple, bend right] (1) to (2);
\draw [style=wires]  (10) -- (2) node[pos=.4, above] {$n-2$};
\draw [style=simple, bend right] (21) to (2);
\draw [style=simple] (0) to (1);
\draw [style=simple] (20) to (21);
\draw [style=simple] (1) to (21);
\draw [style=wires]  (2) -- (6) node[pos=.4, above] {$m$};
\end{pgfonlayer}
\end{tikzpicture} & \ref{CNTii}\\
&=
\begin{tikzpicture}[baseline={([yshift=-.5ex]current bounding box.center)}]
\begin{pgfonlayer}{nodelayer}
\node [style=rn] (0) at (0, 0) {};
\node [style=rn] (20) at (0, -.5) {};
\node [style=oplus] (1) at (0.25, 0) {};
\node [style=dot] (21) at (0.25, -.5) {};
\draw[fill=white] (1.5,0.5) circle [radius=.25] node (2) {$g$};
\node [style=rn] (6) at (3, 0.5) {};
\node [style=rn] (10) at (0, .5) {};
\end{pgfonlayer}
\begin{pgfonlayer}{edgelayer}
\draw [style=simple, bend right] (1) to (2);
\draw [style=wires]  (10) -- (2) node[pos=.4, above] {$n-2$};
\draw [style=simple, bend right] (21) to (2);
\draw [style=simple] (0) to (1);
\draw [style=simple] (20) to (21);
\draw [style=simple] (1) to (21);
\draw [style=wires]  (2) -- (6) node[pos=.4, above] {$m$};
\end{pgfonlayer}
\end{tikzpicture}& \text{By the inductive hypothesis}\\
&=
\begin{tikzpicture}[baseline={([yshift=-.5ex]current bounding box.center)}]
\begin{pgfonlayer}{nodelayer}
\draw[fill=white] (1,0) circle [radius=.25] node (2) {$f$};
\node [style=rn] (6) at (0, 0) {};
\node [style=rn] (10) at (2, 0) {};
\end{pgfonlayer}
\begin{pgfonlayer}{edgelayer}
\draw [style=wires]  (10) -- (2) node[pos=.4, above] {$m$};
\draw [style=wires]  (2) -- (6) node[pos=.4, above] {$n$};
\end{pgfonlayer}
\end{tikzpicture} & \text{By supposition}\\
\end{flalign*}

\end{itemize}

Therefore, the inductive claim holds.
\end{itemize}
\end{proof}

\begin{theorem}{$\CNOT$ is a discrete inverse category.}
\label{thm:CNOTINV}
\end{theorem}
\begin{proof}
The functor $(\_)\cnv:\CNOT\op \rightarrow \CNOT$ which flips circuits horizontally satsifies \ref{INV:1} by construction. It satisfies \ref{INV:2} by Lemma \ref{prop:inv3impinv2} and \ref{INV:3} is satisfied by Proposition \ref{cor:symmetricCommute}. Hence, $\CNOT$ is an inverse category.

$\CNOT$ is equipped with a tensor product ${\_\otimes\_:\CNOT\times\CNOT\to\CNOT}$.  $\Delta$ is a natural transformation by Lemma \ref{lem:fanoutNatural}, which is cocommutative by Lemma \ref{lemma:fanoutCommutative}, coassociative by Lemma \ref{lemma:fanoutAssociative}, a semi-Frobenius algebra by Lemma \ref{lemma:fanoutSemiFrobenius}, and it satisfies \ref{DNV:4} by Lemma \ref{lemma:fanoutCanonical}, where all the dual propositions hold by symmtery. This proves that $\CNOT$ has inverse products. Hence, $\CNOT$ is a discrete inverse category.
\end{proof}
%%%%%%%%%%%%%%%%%%%%%%%%%%%%%%%%%%%%%%%
\section{The Equivalence \texorpdfstring{$\CNOT \cong \ParIso(\CTor_2)^{*}$}{CNOT=ParIso(CTor2)}}
%%%%%%%%%%%%%%%%%%%%%%%%%%%%%%%%%%%%%%%
\label{Appendix C}
The objective of this appendix is to prove that $\CNOT$ and $\ParIso(\CTor_2)^{*}$ are equivalent. The proof involves in five steps:
\begin{description}
\item{(1)} Defining a functor $\tilde H_0: \CNOT \to \ParIso(\CTor_2)^{*}$, which preserves inverse products.
\item{(2)} Showing that $\tilde H_0$ is full and faithful on restriction idempotents.
\item{(3)} Showing that $\tilde H_0$ is essentially surjective.
\item{(4)} Showing that $\tilde H_0$ is full.
\item{(5)} Showing that $\tilde H_0$ is faithful.
\end{description}
The key technical steps ((4) and (5) above) are to reduce the full and faithfulness of $\tilde H_0$ to its full and faithfulness on restriction idempotents (step (2) above).   This latter result is based on the clausal normal form for restriction idempotents in $\CNOT$, which is developed in Section \ref{C.2}.  

%%%%%%%%%%%%%%%%%%%%%%%%%%%%%%%%%%%%%%%%%%%%%%%%%%%%%%%%%%%%%%%%%%%%%%%%%%%%%%%%%%%%

\subsection{Defining the functor \texorpdfstring{$\tilde H_0: \CNOT \to \ParIso(\CTor_2)^{*}$}{H0:CNOT -> ParIso(CTor2)*}}

%%%%%%%%%%%%%%%%%%%%%%%%%%%%%%%%%%%%%%%%%%%%%%%%%%%%%%%%%%%%%%%%%%%%%%%%%%%%%%%%%%%%

To construct a functor $\tilde H_0: \CNOT \rightarrow \ParIso(\CTor_2)^{*}$ we consider the following pullback, where $U:\CTor_2 \to \Sets$ is the underlying functor:

\[
\xymatrix{
\ParIso(\CTor_2)^* \ar[rr]^{\ParIso(U)} \ar@{^{(}->}[d] && \ParIso(\mathsf{Set}) \ar@{^{(}->}[d] \\
\Par(\CTor_2)^* \ar[rr]_{\Par(U)} && \Par(\mathsf{Set})
}
\]

To prove that $\CNOT$ is equivalent to $\ParIso(\CTor_2)^*$, the category of partial isomorphisms of finitely generated non-empty commutative torsors of characteristic $2$, we start by considering a functor $h_0: \CNOT \rightarrow \Par(\Sets)$ and on the one hand lift to a functor $H_0: \CNOT \rightarrow \Par(\CTor_2)^*$ and on the other hand lift to a functor $\tilde h_0:
\CNOT\to\ParIso(\Sets)$.  Then by the pullback of the diagram $\Par(\CTor_2)^* \xrightarrow{\Par(U)} \Par(\Sets) \hookleftarrow \ParIso(\Sets)$ we are given a unique functor $\tilde H_0: \CNOT \rightarrow \ParIso(\CTor_2)^*$:

%We prove that $\CNOT$ is equivalent to $\ParIso(\CTor_2)^*$, the category of partial isomorphisms of finitely generated commutative torsors of characteristic 2.  In order to do so, we consider a functor $h_0:\CNOT\to\Par(\Sets)$ and lift it to a functor $\tilde h_0: \CNOT\to\ParIso(\Sets)$ and another functor $ H_0:\CNOT\to \Par(\CTor_2)^*$.  By the pullback of $\Par(\CTor_2) \xrightarrow{\Par(U)} \Par(\Sets) \hookleftarrow \ParIso(\Sets)$, we lift this to a functor $\CNOT \to \ParIso(\CTor_2)$.  We prove that this induced functor is an equivalence of categories.

\[
\xymatrix{
\CNOT \ar@/^1pc/[ddrrr]|>>>>>>>>>>>>>>>\hole^<<<<<{h_0} \ar@/^3pc/[drrr]^{\tilde h_0} \ar@/^-2pc/[ddr]_{H_0}  \ar@{.>}@/^-1pc/[dr]^{\tilde H_0} & & & & \\
 & \ParIso(\CTor_2)^* \ar[rr]^{\ParIso(U)} \ar@{^{(}->}[d] && \ParIso(\mathsf{Set}) \ar@{^{(}->}[d] \\
& \Par(\CTor_2)^* \ar[rr]_{\Par(U)} && \Par(\mathsf{Set})
}
\]

The functor $h_0: \CNOT \rightarrow \Par(\Sets)$ which we shall describe is a restriction hom-functor so we first prove a general result about such functors.
Given a restriction category  $\mathbb{X}$ and any $X \in \mathbb{X}$, define the following map $h_X := \Total(\mathbb{X})(X,\_): \mathbb{X} \to \Par(\Sets)$ as follows:
\begin{description}
\item[On objects: ] For each object $Y \in \mathbb{X}$,  $h_X(Y) := \{f \in \mathbb{X}(x,y)| \bar{f}=1_x \}$;
\item[On maps: ] For each map $Y \xrightarrow{f} Z$ in $\mathbb{X}$, for all $g \in h_X(Y)$,

$$(h_X(f))(g)  :=
\begin{cases}
gf & \text{ if } \bar{gf}=1_X\\
\uparrow & \text{ otherwise }
\end{cases}$$
\end{description}

\begin{lemma} 
\label{lemma:restrictionfunctorgeneral} 
 \label{remark:restrictionfunctor}
 $h_X: \X \to \Par(\Sets)$ is a restriction functor.
\end{lemma}

\begin{proof}
       To prove $h_X$ is a restriction functor is to prove $h_X$ preserves identities, composition and restriction structure.
\begin{itemize}
\item
First, we prove that $h_X$ preserves identities.  Take any object $Y \in \mathbb{X}$ and any map $f\in h_X(Y)$.  Then, $(h_X(1_Y))(f) = f1_Y=f$ as $\bar{f1_Y}=\bar{f} = 1_X$.
\item
Next we prove that $h_X$ preserves composition.  Consider arbitrary maps $Y \xrightarrow{f} Z \xrightarrow{g} W$ and an $h \in h_X(Y)$.  

Suppose that $\bar{hfg}=1_X$, then $(h_X(fg))(h) = hfg$ and $\bar{hf}=1_X$, then 
$$h_X(g)((h_X(f))(h)) = (h_X(g))(hf)= hfg.$$

On the otherhand, suppose that $\bar{hfg}\neq 1_X$ then
$$h_X(g)((h_X(f))(h)) = h_X(g)(\uparrow) =\uparrow$$
If $\bar{hf} = 1_X$, then 
$$h_X(g)(h_X(f)(h)) = h_X(g)(hf) = \uparrow.$$
Therefore, $h_X$ preserves composition.

\item
Finally, we prove that $h_X$ preserves restriction.

For any $f:Y\to Z$ and any $g \in h_X(Y)$:
$$(h_X(\bar{f}))(g)=
\begin{cases}
g\bar{f} & \text{ if } \bar{g\bar{f}}=\bar{gf}=1_X\\
\uparrow & \text{ otherwise }
\end{cases}$$

However,

$$
\bar{h_X(f)}(g)=
\begin{cases}
g & \text{ if } (h_X(f))(g)\downarrow\\
\uparrow & \text{ otherwise }
\end{cases}
$$

But, $(h_X(f))(g)$ if and only if $\bar{gf} =1_X$ ; moreover if $\bar{gf} =1_X$, then $g\bar{f} = \bar{gf}g =1_Xg =g$. Therefore, $h_X$ preserves restriction.
\end{itemize}
\end{proof}

Fixing $\mathbb{X}=\CNOT$ and $X=0$, we obtain a functor $h_0:= \Total(\CNOT)(0,\_): \CNOT \to \Par(\Sets)$.  As $\CNOT$ is an inverse category, it follows by Lemma \ref{lemma:restrictionfunctorgeneral} that every map in $h_0(\CNOT)$ is a partial isomorphism.  Therefore $h_0:\CNOT \to \Par(\Sets)$ factors through $\ParIso(\Sets)$ as $\tilde h_0: \CNOT \to \ParIso(\Sets)$.

%In our case this yields $h_0:= \Total(\CNOT)(0,\_): \CNOT \to \Par(\Sets)$ which factors as $\tilde h_0: \CNOT\to \ParIso(\Sets)$ followed by the inclusion $\iota:\ParIso(\Sets) \to \Par(\Sets)$. As $h_0:\CNOT \to \Par(\Sets)$ is a restriction functor by Lemma \ref{lemma:restrictionfunctorgeneral}, and every map $h_0(f)$ is a partial isomorphism, $h_0$ factors through $\ParIso(\Sets)$ as the functor $\tilde h_0:\CNOT \to \ParIso(\Sets)$.

Now that we have the candidate functor $\tilde h_0:\CNOT \to \ParIso(\Sets)$ for the pullback, we must also show that we can factor $h_0$ through $\Par(U): \Par(\CTor_2)^* \to\Par(\Sets)$.  To do so, we show that the objects of $\CNOT$ have an internal torsor structure and show that $h_0$ preserves this structure:  thus showing it can be factored through $\Par(\CTor_2)^*$.  This will allow us to construct the functor $H_0: \CNOT\to \Par(\CTor_2)^*$ and whence $\tilde H_0$.

Proving that $h_0$ is a functor from $\CNOT$ to $\Par(\CTor_2)^*$ is not such a trivial task.  To this end, we make several observations regarding maps in the homset $\CNOT(0,n)$ for any $n \in \mathbb{N}$.

To succinctly express total maps $0 \xrightarrow{f} n$ for any $n \in \mathbb{N}$ in $\CNOT$, we define the following family of functions.

\begin{definition}
\label{defn:hat}
Define a family of functions $\hat{\_}_n:\mathbb{Z}_2^n \to \Total(\CNOT)(0,n)$ for all $n \in \mathbb{N}$ as follows.

Take $\hat{ } := 1_{0}$. For any $b \in \mathbb{Z}_2^1$ define:

$$
\hat{b} := 
\begin{cases}
\onein & \text{ if } n =1 \\
\zeroin & \text{ otherwise }
\end{cases}
$$

Moreover, for all $n \in \mathbb{N}$ such that $n>1$, define:

$$
\hat{b_1,\cdots,b_n} := \hat{b_1} \otimes \cdots \otimes \hat{b_n}
$$
\end{definition}

%%%%%%%%%%%%%%%%%%
\begin{lemma}
\label{lemma:behaveasexpected}
Consider a circuit $f: n \to m$ with no output ancill\ae. Then, for any $x \in \Z_2^n$, there is some $y \in \Z_2^m$ such that $\hat{x}f = \hat{y}$
\end{lemma}

\begin{proof}
It will suffice to prove our claim for all  $b\in \mathbb{Z}_2^n$ on a single controlled-not gate, then by induction, the more general claim follows immediately.

\begin{itemize}
\item
Take $b=(0,0)$, then by \ref{CNTvii}:

\begin{align*}
\begin{tikzpicture}[baseline={([yshift=-.5ex]current bounding box.center)}]
\begin{pgfonlayer}{nodelayer}
\node [style=zeroin] (0) at (0, .5) {};
\node [style=dot] (1) at (.5, .5) {};
\node [style=rn] (2) at (1, .5) {};
\node [style=zeroin] (10) at (0, 0) {};
\node [style=oplus] (11) at (.5, 0) {};
\node [style=rn] (12) at (1, 0) {};
\end{pgfonlayer}
\begin{pgfonlayer}{edgelayer}
\draw [style=simple] (0) to (2);
\draw [style=simple] (10) to (12);
\draw [style=simple] (1) to (11);
\end{pgfonlayer}
\end{tikzpicture}
&=
\begin{tikzpicture}[baseline={([yshift=-.5ex]current bounding box.center)}]
\begin{pgfonlayer}{nodelayer}
\node [style=zeroin] (0) at (0, .5) {};
\node [style=rn] (2) at (.5, .5) {};
\node [style=zeroin] (10) at (0, 0) {};
\node [style=rn] (12) at (.5, 0) {};
\end{pgfonlayer}
\begin{pgfonlayer}{edgelayer}
\draw [style=simple] (0) to (2);
\draw [style=simple] (10) to (12);
\end{pgfonlayer}
\end{tikzpicture}
\end{align*}

\item
Take $b=(0,1)$, then by \ref{CNTvii}:

\begin{align*}
\begin{tikzpicture}[baseline={([yshift=-.5ex]current bounding box.center)}]
\begin{pgfonlayer}{nodelayer}
\node [style=zeroin] (0) at (0, .5) {};
\node [style=dot] (1) at (.5, .5) {};
\node [style=rn] (2) at (1, .5) {};
\node [style=onein] (10) at (0, 0) {};
\node [style=oplus] (11) at (.5, 0) {};
\node [style=rn] (12) at (1, 0) {};
\end{pgfonlayer}
\begin{pgfonlayer}{edgelayer}
\draw [style=simple] (0) to (2);
\draw [style=simple] (10) to (12);
\draw [style=simple] (1) to (11);
\end{pgfonlayer}
\end{tikzpicture}
&=
\begin{tikzpicture}[baseline={([yshift=-.5ex]current bounding box.center)}]
\begin{pgfonlayer}{nodelayer}
\node [style=zeroin] (0) at (0, .5) {};
\node [style=rn] (2) at (.5, .5) {};
\node [style=onein] (10) at (0, 0) {};
\node [style=rn] (12) at (.5, 0) {};
\end{pgfonlayer}
\begin{pgfonlayer}{edgelayer}
\draw [style=simple] (0) to (2);
\draw [style=simple] (10) to (12);
\end{pgfonlayer}
\end{tikzpicture}
\end{align*}

\item
Take $b=(1,0)$, then:

\begin{align*}
\begin{tikzpicture}[baseline={([yshift=-.5ex]current bounding box.center)}]
\begin{pgfonlayer}{nodelayer}
\node [style=onein] (0) at (0, .5) {};
\node [style=dot] (1) at (.5, .5) {};
\node [style=rn] (2) at (1, .5) {};
\node [style=zeroin] (10) at (0, 0) {};
\node [style=oplus] (11) at (.5, 0) {};
\node [style=rn] (12) at (1, 0) {};
\end{pgfonlayer}
\begin{pgfonlayer}{edgelayer}
\draw [style=simple] (0) to (2);
\draw [style=simple] (10) to (12);
\draw [style=simple] (1) to (11);
\end{pgfonlayer}
\end{tikzpicture}
&:=
\begin{tikzpicture}[baseline={([yshift=-.5ex]current bounding box.center)}]
\begin{pgfonlayer}{nodelayer}
\node [style=onein] (0) at (0, .5) {};
\node [style=dot] (1) at (1, .5) {};
\node [style=rn] (2) at (1.5, .5) {};
\node [style=onein] (10) at (0, 0) {};
\node [style=oplus] (13) at (.5, 0) {};
\node [style=oplus] (11) at (1, 0) {};
\node [style=rn] (12) at (1.5, 0) {};
\node [style=onein] (20) at (0, -.5) {};
\node [style=dot] (21) at (.5, -.5) {};
\node [style=oneout] (22) at (1, -.5) {};
\end{pgfonlayer}
\begin{pgfonlayer}{edgelayer}
\draw [style=simple] (0) to (2);
\draw [style=simple] (10) to (12);
\draw [style=simple] (1) to (11);
\draw [style=simple] (20) to (22);
\draw [style=simple] (13) to (21);
\end{pgfonlayer}
\end{tikzpicture}\\
&=
\begin{tikzpicture}[baseline={([yshift=-.5ex]current bounding box.center)}]
\begin{pgfonlayer}{nodelayer}
\node [style=onein] (0) at (0, 1) {};
\node [style=oplus] (1) at (.5, 1) {};
\node [style=dot] (2) at (1, 1) {};
\node [style=oplus] (3) at (1.5, 1) {};
\node [style=rn] (4) at (2, 1) {};
\node [style=onein] (10) at (0, .5) {};
\node [style=rn] (11) at (.5, .5) {};
\node [style=oplus] (12) at (1, .5) {};
\node [style=rn] (13) at (1.5, .5) {};
\node [style=rn] (14) at (2, .5) {};
\node [style=onein] (20) at (0, 0) {};
\node [style=dot] (21) at (.5, 0) {};
\node [style=rn] (22) at (1, 0) {};
\node [style=dot] (23) at (1.5, 0) {};
\node [style=oneout] (24) at (2, 0) {};
\end{pgfonlayer}
\begin{pgfonlayer}{edgelayer}
\draw [style=simple] (0) to (4);
\draw [style=simple] (10) to (14);
\draw [style=simple] (20) to (24);
\draw [style=simple] (1) to (21);
\draw [style=simple] (3) to (23);
\draw [style=simple] (2) to (12);
\end{pgfonlayer}
\end{tikzpicture}  & \ref{CNTviii}\\
&=
\begin{tikzpicture}[baseline={([yshift=-.5ex]current bounding box.center)}]
\begin{pgfonlayer}{nodelayer}
\node [style=onein] (0) at (0, 1) {};
\node [style=oplus] (1) at (.5, 1) {};
\node [style=dot] (2) at (1, 1) {};
\node [style=oplus] (3) at (1.5, 1) {};
\node [style=rn] (4) at (2, 1) {};
\node [style=onein] (10) at (0, .5) {};
\node [style=rn] (11) at (.5, .5) {};
\node [style=oplus] (12) at (1, .5) {};
\node [style=rn] (13) at (1.5, .5) {};
\node [style=rn] (14) at (2, .5) {};
\node [style=onein] (20) at (0, 1.5) {};
\node [style=dot] (21) at (.5, 1.5) {};
\node [style=oneout] (22) at (.85, 1.5) {};
\node [style=onein] (25) at (1.15, 1.5) {};
\node [style=dot] (23) at (1.5, 1.5) {};
\node [style=oneout] (24) at (2, 1.5) {};
\end{pgfonlayer}
\begin{pgfonlayer}{edgelayer}
\draw [style=simple] (0) to (4);
\draw [style=simple] (10) to (14);
\draw [style=simple] (20) to (22);
\draw [style=simple] (24) to (25);
\draw [style=simple] (1) to (21);
\draw [style=simple] (3) to (23);
\draw [style=simple] (2) to (12);
\end{pgfonlayer}
\end{tikzpicture}& \ref{CNTiv}\\
&=
\begin{tikzpicture}[baseline={([yshift=-.5ex]current bounding box.center)}]
\begin{pgfonlayer}{nodelayer}
\node [style=onein] (0) at (0, 1) {};
\node [style=oplus] (1) at (.5, 1) {};
\node [style=rn] (2) at (1, 1) {};
\node [style=oplus] (3) at (1.5, 1) {};
\node [style=rn] (4) at (2, 1) {};
\node [style=onein] (10) at (0, .5) {};
\node [style=rn] (11) at (.5, .5) {};
\node [style=rn] (12) at (1, .5) {};
\node [style=rn] (13) at (1.5, .5) {};
\node [style=rn] (14) at (2, .5) {};
\node [style=onein] (20) at (0, 1.5) {};
\node [style=dot] (21) at (.5, 1.5) {};
\node [style=oneout] (22) at (.85, 1.5) {};
\node [style=onein] (25) at (1.15, 1.5) {};
\node [style=dot] (23) at (1.5, 1.5) {};
\node [style=oneout] (24) at (2, 1.5) {};
\end{pgfonlayer}
\begin{pgfonlayer}{edgelayer}
\draw [style=simple] (0) to (4);
\draw [style=simple] (10) to (14);
\draw [style=simple] (20) to (22);
\draw [style=simple] (24) to (25);
\draw [style=simple] (1) to (21);
\draw [style=simple] (3) to (23);
\end{pgfonlayer}
\end{tikzpicture} & \ref{CNTvii}\\
&=
\begin{tikzpicture}[baseline={([yshift=-.5ex]current bounding box.center)}]
\begin{pgfonlayer}{nodelayer}
\node [style=onein] (0) at (0, 1) {};
\node [style=oplus] (1) at (.5, 1) {};
\node [style=rn] (2) at (1, 1) {};
\node [style=oplus] (3) at (1.5, 1) {};
\node [style=rn] (4) at (2, 1) {};
\node [style=onein] (10) at (0, .5) {};
\node [style=rn] (11) at (.5, .5) {};
\node [style=rn] (12) at (1, .5) {};
\node [style=rn] (13) at (1.5, .5) {};
\node [style=rn] (14) at (2, .5) {};
\node [style=onein] (20) at (0, 1.5) {};
\node [style=dot] (21) at (.5, 1.5) {};
\node [style=rn] (22) at (.85, 1.5) {};
\node [style=dot] (23) at (1.5, 1.5) {};
\node [style=oneout] (24) at (2, 1.5) {};
\end{pgfonlayer}
\begin{pgfonlayer}{edgelayer}
\draw [style=simple] (0) to (4);
\draw [style=simple] (10) to (14);
\draw [style=simple] (20) to (24);
\draw [style=simple] (1) to (21);
\draw [style=simple] (3) to (23);
\end{pgfonlayer}
\end{tikzpicture} & \ref{CNTiv}\\
&=
\begin{tikzpicture}[baseline={([yshift=-.5ex]current bounding box.center)}]
\begin{pgfonlayer}{nodelayer}
\node [style=onein] (0) at (0, 1) {};
\node [style=rn] (1) at (.5, 1) {};
\node [style=rn] (2) at (1, 1) {};
\node [style=rn] (3) at (1.5, 1) {};
\node [style=rn] (4) at (2, 1) {};
\node [style=onein] (10) at (0, .5) {};
\node [style=rn] (11) at (.5, .5) {};
\node [style=rn] (12) at (1, .5) {};
\node [style=rn] (13) at (1.5, .5) {};
\node [style=rn] (14) at (2, .5) {};
\node [style=onein] (20) at (0, 1.5) {};
\node [style=rn] (21) at (.5, 1.5) {};
\node [style=rn] (22) at (.85, 1.5) {};
\node [style=rn] (23) at (1.5, 1.5) {};
\node [style=oneout] (24) at (2, 1.5) {};
\end{pgfonlayer}
\begin{pgfonlayer}{edgelayer}
\draw [style=simple] (0) to (4);
\draw [style=simple] (10) to (14);
\draw [style=simple] (20) to (24);
\end{pgfonlayer}
\end{tikzpicture} & \ref{CNTii}\\
&=
\begin{tikzpicture}[baseline={([yshift=-.5ex]current bounding box.center)}]
\begin{pgfonlayer}{nodelayer}
\node [style=onein] (0) at (0, .5) {};
\node [style=rn] (2) at (2, .5) {};
\node [style=onein] (10) at (0, 0) {};
\node [style=rn] (12) at (2, 0) {};
\end{pgfonlayer}
\begin{pgfonlayer}{edgelayer}
\draw [style=simple] (0) to (2);
\draw [style=simple] (10) to (12);
\end{pgfonlayer}
\end{tikzpicture} & \ref{CNTvi}
\end{align*}

\item
Take $b=(1,1)$, then:
\begin{align*}
\begin{tikzpicture}[baseline={([yshift=-.5ex]current bounding box.center)}]
\begin{pgfonlayer}{nodelayer}
\node [style=onein] (0) at (0, .5) {};
\node [style=dot] (1) at (.5, .5) {};
\node [style=rn] (2) at (1, .5) {};
\node [style=onein] (10) at (0, 0) {};
\node [style=oplus] (11) at (.5, 0) {};
\node [style=rn] (12) at (1, 0) {};
\end{pgfonlayer}
\begin{pgfonlayer}{edgelayer}
\draw [style=simple] (0) to (2);
\draw [style=simple] (10) to (12);
\draw [style=simple] (1) to (11);
\end{pgfonlayer}
\end{tikzpicture}
&=
\begin{tikzpicture}[baseline={([yshift=-.5ex]current bounding box.center)}]
\begin{pgfonlayer}{nodelayer}
\node [style=onein] (0) at (0, .5) {};
\node [style=dot] (1) at (.5, .5) {};
\node [style=oneout] (2) at (1, .5) {};
\node [style=onein] (3) at (1.5, .5) {};
\node [style=rn] (4) at (2, .5) {};
\node [style=onein] (10) at (0, 0) {};
\node [style=oplus] (11) at (.5, 0) {};
\node [style=rn] (12) at (2, 0) {};
\end{pgfonlayer}
\begin{pgfonlayer}{edgelayer}
\draw [style=simple] (0) to (2);
\draw [style=simple] (3) to (4);
\draw [style=simple] (10) to (12);
\draw [style=simple] (1) to (11);
\end{pgfonlayer}
\end{tikzpicture} & \ref{CNTiv}\\
&=:
\begin{tikzpicture}[baseline={([yshift=-.5ex]current bounding box.center)}]
\begin{pgfonlayer}{nodelayer}
\node [style=onein] (0) at (0, .5) {};
\node [style=rn] (2) at (2, .5) {};
\node [style=zeroin] (10) at (0, 0) {};
\node [style=rn] (12) at (2, 0) {};
\end{pgfonlayer}
\begin{pgfonlayer}{edgelayer}
\draw [style=simple] (0) to (2);
\draw [style=simple] (10) to (12);
\end{pgfonlayer}
\end{tikzpicture}
\end{align*}

\end{itemize}
\end{proof}
%%%%%%%%%%%%%%%%%%%%%%%%%%%

The following lemma is an intuitive result which we shall use:

\begin{lemma}
\label{lemma:totalordegen}
For every $f \in \CNOT(0,n)$, $f$ is either total or degenerate (i.e., $\Omega$).
\end{lemma}
\begin{proof}
%The claim follows from Lemma \ref{lemma:behaveasexpected} and Lemma \ref{lemma:tensorOmegaDegen}.
Consider any circuit $f:0\to n$ for any $n \in \mathbb{N}$.  We prove that $f$ is either total or degenerate by induction.
\begin{itemize}
\item If $|f|=0$, then $f$ is a permutation, and is therefore total.

\item Inductively suppose that $f$ is either total or degenerate.  Consider any circuit $g:n\to m$ such that $|g|=1$.  If $f$ is degenerate, then $fg$ is degenerate as well by Lemma \ref{lemma:tensorOmegaDegen}.  Otherwise, suppose that $f$ is total.  There are three cases:

\begin{itemize}
\item If $\onein \in g$, then $\bar{fg} = \bar{f\otimes g} = \bar{f} \otimes \bar{g} = 1 \otimes 1 = 1$.
\item If $\oneout \in g$, then the gate to the left of $\oneout$ is either $\onein$ or $\zeroin$.  If it is $\onein$, then as $\onein\oneout=1$, it is total.  Otherwise, if it is $\oneout$, then $\onein\oneout=:\Omega$, so $fg$ is degenerate by Lemma \ref{lemma:tensorOmegaDegen}.
\item If $\cnot \in g$, then $fg$ is total by Lemma \ref{lemma:behaveasexpected}.
\end{itemize}
\end{itemize}
\end{proof}

%%%%%%%%%%%%%%%%%%%%%%%%%%%%

To prove that $h_0$ takes maps in $\CNOT$ to well-defined maps in $\Par(\CTor_2)^*$, we construct a torsor-like operation in $\CNOT$ which gives the internal torsor structure which we mentioned above.  This will act as the para-multiplication when we project the last 1 out of 3 wires.

\begin{definition}
\label{defn:plusmap}
Define a family of maps $+_n:3n\to 3n$ in $\CNOT$ inductively such that such that on no wires, $+_0:=1_0$.  Furthermore, on one wire:

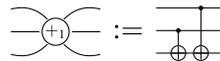
\begin{align*}
\begin{tikzpicture}[baseline={([yshift=-.5ex]current bounding box.center)}]
\begin{pgfonlayer}{nodelayer}
\node [style=rn] (0) at (0, 1) {};
\node [style=rn] (1) at (0, .5) {};
\node [style=rn] (2) at (0, 0) {};
\draw[fill=white] (1,0.5) circle [radius=.3] node (10) {$+_1$};
\node [style=rn] (20) at (2, 1) {};
\node [style=rn] (21) at (2, .5) {};
\node [style=rn] (22) at (2, 0) {};
\end{pgfonlayer}
\begin{pgfonlayer}{edgelayer}
\draw [style=simple, bend left] (0) to (10);
\draw [style=simple] (1) to (10);
\draw [style=simple, bend right] (2) to (10);
\draw [style=simple, bend right] (20) to (10);
\draw [style=simple] (21) to (10);
\draw [style=simple, bend left] (22) to (10);
\end{pgfonlayer}
\end{tikzpicture}
:=
\begin{tikzpicture}[baseline={([yshift=-.5ex]current bounding box.center)}]
\begin{pgfonlayer}{nodelayer}
\node [style=rn] (0) at (0, 1) {};
\node [style=rn] (1) at (.5, 1) {};
\node [style=dot] (2) at (1, 1) {};
\node [style=rn] (3) at (1.5, 1) {};
\node [style=rn] (10) at (0, 0.5) {};
\node [style=dot] (11) at (.5, 0.5) {};
\node [style=rn] (12) at (1, 0.5) {};
\node [style=rn] (13) at (1.5, 0.5) {};
\node [style=rn] (20) at (0, 0) {};
\node [style=oplus] (21) at (.5, 0) {};
\node [style=oplus] (22) at (1, 0) {};
\node [style=rn] (23) at (1.5, 0) {};
\end{pgfonlayer}
\begin{pgfonlayer}{edgelayer}
\draw [style=simple] (0) to (3);
\draw [style=simple] (10) to (13);
\draw [style=simple] (20) to (23);
\draw [style=simple] (2) to (22);
\draw [style=simple] (11) to (21);
\end{pgfonlayer}
\end{tikzpicture}
\end{align*}

Furthermore, for any $n>1$:

$$
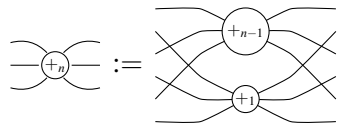
\begin{tikzpicture}[baseline={([yshift=-.5ex]current bounding box.center)}]
\begin{pgfonlayer}{nodelayer}
\node [style=rn] (0) at (0, 1) {};
\node [style=rn] (1) at (0, .5) {};
\node [style=rn] (2) at (0, 0) {};
\draw[fill=white] (1,0.5) circle [radius=.3] node (10) {$+_n$};
\node [style=rn] (20) at (2, 1) {};
\node [style=rn] (21) at (2, .5) {};
\node [style=rn] (22) at (2, 0) {};
\end{pgfonlayer}
\begin{pgfonlayer}{edgelayer}
\draw [style=simple, bend left] (0) to (10);
\draw [style=simple] (1) to (10);
\draw [style=simple, bend right] (2) to (10);
\draw [style=simple, bend right] (20) to (10);
\draw [style=simple] (21) to (10);
\draw [style=simple, bend left] (22) to (10);
\end{pgfonlayer}
\end{tikzpicture}
:=
\begin{tikzpicture}[baseline={([yshift=-.5ex]current bounding box.center)}]
\begin{pgfonlayer}{nodelayer}
\node [style=rn] (-000) at (-1, 2.5) {};
\node [style=rn] (-001) at (-1, 2) {};
\node [style=rn] (-002) at (-1, 1.5) {};
\node [style=rn] (-100) at (-1, 1) {};
\node [style=rn] (-101) at (-1, .5) {};
\node [style=rn] (-102) at (-1, 0) {};
\node [style=rn] (000) at (0, 2.5) {};
\node [style=rn] (001) at (0, 2) {};
\node [style=rn] (002) at (0, 1.5) {};
\draw[fill=white] (1,2) circle [radius=0.52] node (010) {$+_{n-1}$};
\node [style=rn] (020) at (2, 2.5) {};
\node [style=rn] (021) at (2, 2) {};
\node [style=rn] (022) at (2, 1.5) {};
\node [style=rn] (100) at (0, 1) {};
\node [style=rn] (101) at (0, .5) {};
\node [style=rn] (102) at (0, 0) {};
\draw[fill=white] (1,0.5) circle [radius=.3] node (110) {$+_1$};
\node [style=rn] (120) at (2, 1) {};
\node [style=rn] (121) at (2, .5) {};
\node [style=rn] (122) at (2, 0) {};
\node [style=rn] (200) at (3, 2.5) {};
\node [style=rn] (201) at (3, 2) {};
\node [style=rn] (202) at (3, 1.5) {};
\node [style=rn] (300) at (3, 1) {};
\node [style=rn] (301) at (3, .5) {};
\node [style=rn] (302) at (3, 0) {};
\end{pgfonlayer}
\begin{pgfonlayer}{edgelayer}
\draw plot [smooth, tension=0.3] coordinates { (-000) (000) (010)};
\draw plot [smooth, tension=0.3] coordinates { (-001) (100) (110)};
\draw plot [smooth, tension=0.3] coordinates { (-002) (001) (010)};
\draw plot [smooth, tension=0.3] coordinates { (-101) (002) (010)};
\draw plot [smooth, tension=0.3] coordinates { (-102) (102) (110)};
\draw plot [smooth, tension=0.3] coordinates { (-100) (101) (110)};
\draw plot [smooth, tension=0.3] coordinates { (302) (122) (110)};
\draw plot [smooth, tension=0.3] coordinates { (301) (022) (010)};
\draw plot [smooth, tension=0.3] coordinates { (300) (121) (110)};
%\draw plot [smooth, tension=0.3] coordinates { (020) (010) (021)};
\draw plot [smooth, tension=0.3] coordinates { (201) (120) (110)};
\draw plot [smooth, tension=0.3] coordinates { (202) (021) (010)};
\draw plot [smooth, tension=0.3] coordinates { (200) (020) (010)};
\end{pgfonlayer}
\end{tikzpicture}
$$
\end{definition}

Now we show that $\Total(\CNOT)(0,\_)$ really does produces maps which preserve torsor structure.

%%%%%%%%%%%%%%%%%%%%%%%%%%%%%%%%

\begin{lemma}
\label{lemma:pluscommutes}
For any map $f:n\to m$ in $\CNOT$, $(f\otimes f\otimes f)+_m = +_n(f\otimes f\otimes f)$.
\end{lemma}
\begin{proof}
For any map $f:n\to m$ in $\CNOT$, we prove $\otimes^3 f+_m = +_n \otimes^3 f$ by induction on the size of $f$.
\begin{itemize}
\item When $|f|=0$, $f$ is a permutation, so it is immediate that $\otimes^3 f+_m = +_n \otimes^3 f$.
\item Suppose that $\otimes^3 f+ = + \otimes^3 f$ for all $|f| \leq k$.
 Consider some $f:n\to m$ with $|f|=k+1$.  We can decompose $f=gh$ for some $|g|=1$ and $|h|=k$.
Note that $\otimes^3 f+ = \otimes^3 (gh) + = \otimes^3 g + \otimes^3 h$ by $|f|$.

We proceed by cases on the elements of $g$.

\begin{description}
\item[$\mathsf{cnot} \in g$:] 
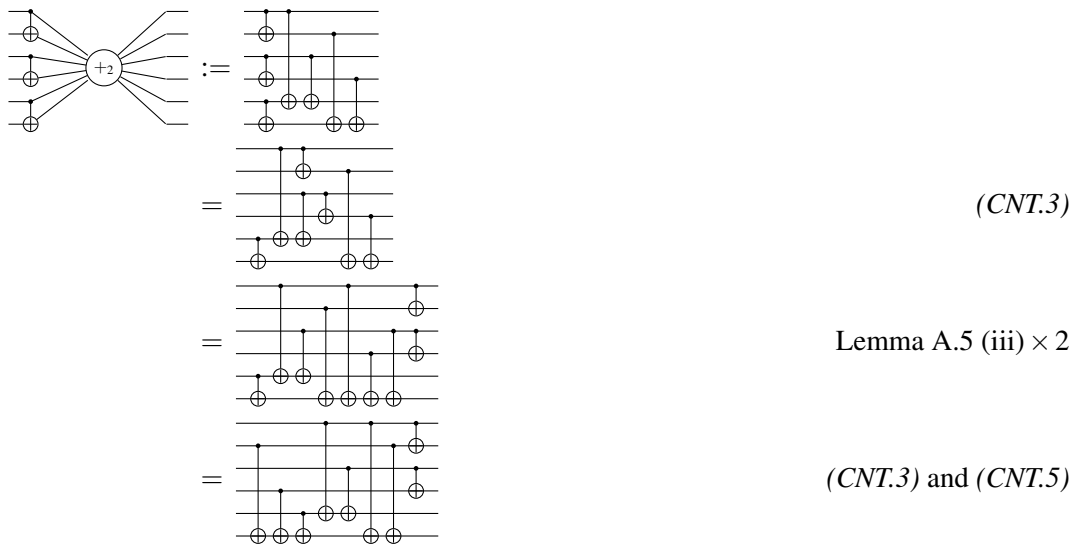
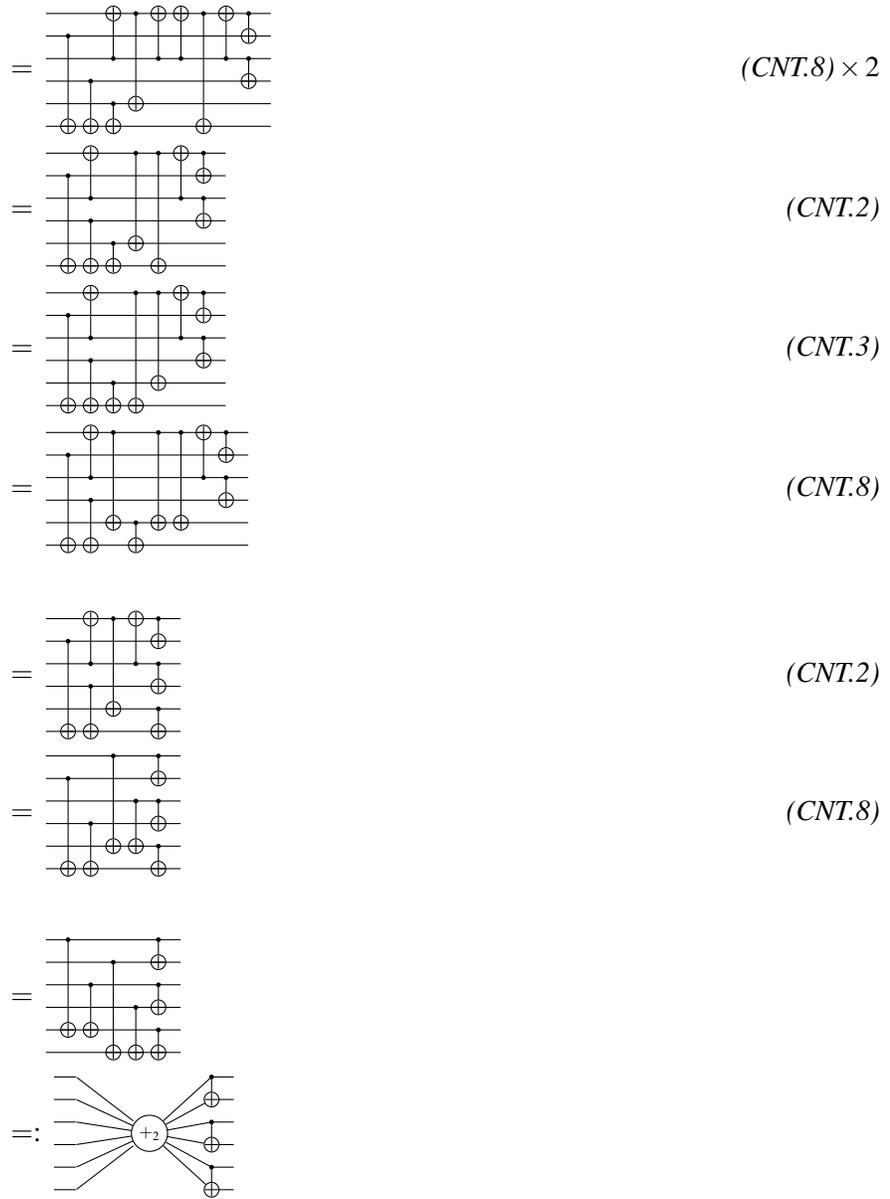
\begin{flalign*}

\end{flalign*}

Therefore, $\otimes^3 f+ = \otimes^3 g + \otimes^3 h = + \otimes^3 (gh) = +\otimes^3 f$

\item[ $|1 \rangle \in g$ ]
Suppose otherwise that $\onein \in g$.  Then by Lemma \ref{lemma:behaveasexpected}, $\otimes^3 f+ = \otimes^3 g + \otimes^3 h = + \otimes^3 (gh) = +\otimes^3 f$.

\item[ $\langle 1 | \in g$ ]
If $\oneout \in g$, dually to the previous case, $\otimes^3 f+ = \otimes^3 g + \otimes^3 h = + \otimes^3 (gh) = +\otimes^3 f$.

\end{description}

\end{itemize}
\end{proof}
 %%%%%%%%%%%%%%%%%%%%%%%%%%%%%%%%%%%%%%%%%%%%%
 
 We also show that $h_0$ produces well defined maps in $\Par(\CTor_2)$.

\begin{lemma}
\label{lemma:plusClosure}
Consider any $n,m \in \mathbb{N}$ and map $f\in\CNOT(n,m)$. For any $x,y,z \in \CNOT(0,n)$ such that $\bar{xf}=\bar{yf}=\bar{zf}=1_0$, it follows that $\bar{(x\otimes y \otimes z)+_n\otimes^3 f } = 1_0$.
\end{lemma}

\begin{proof}
Consider an arbitrary map $f\in\CNOT(n,m)$ and any $x,y,z \in \CNOT(0,n)$ such that $\bar{xf}=\bar{yf}=\bar{zf}=1_0$.
By Lemma \ref{lemma:pluscommutes}:
\begin{align*}
\bar{(x\otimes y \otimes z) +_n \otimes^3f }
&=\bar{(x\otimes y \otimes z)\otimes^3f  +_m}
= \bar{(xf \otimes yf \otimes zf) +_m}\\
&=\bar{(xf \otimes yf \otimes zf) \bar{+_m}}
=\bar{(xf \otimes yf \otimes zf) 1_{3m}}\\
&=\bar{xf \otimes yf \otimes zf}
=\bar{xf} \otimes \bar{yf} \otimes \bar{zf}
= 1_0 \otimes 1_0 \otimes 1_0
= 1_0
\end{align*}
\end{proof}

%%%%%%%%%%%%%%%%%%%%%%%%%%%%%%%%%%%%

We now prove that $H_0: \CNOT \rightarrow \Par(\Tor_2)^*$ is a functor.

\begin{lemma}
\label{lem:FisFunctor}
$h_0$ can be factored as $H_0 \Par(U)$ and $h_0$ preserves torsor structure. Thus, $\CNOT$ has internal torsor structure which is preserved by $h_0$.
\end{lemma}

\begin{proof}
First, we prove for every $f:n\to m$ in $\CNOT$, $h_0(f)$ can be regarded as a map in $\Par(\CTor_2)^*$. Consider an arbitrary map $f:n\to m$ in $\CNOT$.  If $f$ is degenerate, then $h_0(f)$ vacuously preserves torsor structure.  Suppose otherwise that there exists some $x,y,z\in \mathbb{Z}_2^n$ and $x',y',z'  \in \mathbb{Z}_2^m$ such that $\hat{x'}=\hat{x}f$, $\hat{y'}=\hat{y}f$ and $\hat{z'}=\hat{z}f$. Forgetting the first $2$ out of $3$ wires, by Lemma \ref{lemma:plusClosure}, $h_0(f)(x\oplus y \oplus z)$ is defined; moreover, by Lemma \ref{lemma:pluscommutes}, $h_0(f)(x\oplus y \oplus z)=h_0(f)(x)\oplus h_0(f)(y) \oplus h_0(f)(z)$, so $h_0(f)$ preserves the para-multiplication.
\end{proof}

%%%%%%%%%%%%%%%%%%%%%%%%%%%%%%%%%%%%%%%

\begin{corollary}
\label{rem:pullback}
The functor $h_0$ can be lifted to a functor $\tilde H_0: \CNOT \rightarrow \ParIso(\CTor_2)^*$.
\end{corollary}

\begin{proof}
As $h_0$ can be factored through $\ParIso(\Sets)$ by Lemma \ref{remark:restrictionfunctor} and $\Par(\CTor_2)^*$ by  Lemma \ref{lem:FisFunctor}, it is also a functor to $\tilde H_0:\CNOT\to\ParIso(\CTor_2)^*$ by pullback.
%Recall that the square described in Remark \ref{remark:pullback} is a pullback; therefore, by Remark \ref{rem:pullback}, Corollary \ref{remark:restrictionfunctor} and Lemma \ref{lem:FisFunctor}, there is a unique functor $\tilde H_0:\CNOT\to\ParIso(\CTor_2)^*$ induced by the pullback of $h_0$ along itself.  
\end{proof}

%%%%%%%%%%%%%%%%%%%%%%%%%%%

\begin{lemma}
\label{lemma:PreservesInverseProducts}
$\tilde H_0:\CNOT \to \ParIso(\CTor_2)^*$ preserves inverse products.
\end{lemma}
\begin{proof}
We prove that $\tilde H_0$ preserves inverse products by the examination of components within $\ParIso(\CTor_2)^*$.

Consider two arbitrary circuits $f:n\to m$ and $f':n'\to m'$ in $\CNOT$.  Moreover, consider any $(b_1,\cdots, b_{n+n'}) \in \mathbb{Z}_2^{n+n'}$.  Then by construction of $\tilde H_0$:
\begin{align*}
\tilde H_0 (f\otimes g)(b_1, \cdots ,b_{n+n'})
&=\tilde H_0 (\hat{b_1, \cdots ,b_{n+n'}}(f\otimes g))(*)\\
&= \tilde H_0 (\hat{b_1, \cdots ,b_{n}}f)\otimes \tilde H_0 (\hat{b_{n+1}, \cdots ,b_{n+n'}}g)(*)\\
&= \tilde H_0 (\hat{b_1, \cdots ,b_{n}}f\otimes \hat{b_{n+1}, \cdots ,b_{n+n'}}g)(*)\\
&=(\tilde H_0 (f)\otimes \tilde H_0 (g)) (b_1, \cdots ,b_{n+n'})
\end{align*}
Moreover:
\begin{align*}
\tilde H_0 (f\Delta)(b_1, \cdots ,b_{n+n'})
=\tilde H_0 (f \otimes f)(b_1, \cdots ,b_{n+n'})
=(\tilde H_0 (f)\otimes \tilde H_0 (f  )) (b_1, \cdots ,b_{n+n'})
\end{align*}

Therefore, $\tilde H_0$ preserves inverse products.
\end{proof}

\begin{remark} {\em 
As mentioned in the main body of the paper, an alternative way to define $\tilde H_0$ would be to use the fact that $\CNOT$ is defined freely on gates with relations.  Thus, it would have sufficed to provide the interpretation of the gates and then verify the relations.  Our proof produces the same functor but it avoids direct verification of the identities.  The direct proof may, in fact, be more straightforward; however, some of the lemmas which we have established in our proof will be reused.}
\end{remark}

As we will reduce the full and faithfulness of $\tilde H_0$ to its full and faithfulness on restriction idempotents, the next section is dedicated to describing a normal form for restriction idempotents in $\CNOT$ and 
establishing the full and faithfulness of $\tilde H_0$ on restriction idempotents.

%%%%%%%%%%%%%%%%%%%%%%%%%%%%%%%%%%%%%%%%%%%%%%%%%%%%%%%%%%%%%%%%%%%%%%%%%%%%%%%%%%%%

\subsection{Clausal form for restriction idempotents in \texorpdfstring{$\CNOT$}{CNOT}}
\label{C.2}
%%%%%%%%%%%%%%%%%%%%%%%%%%%%%%%%%%%%%%%%%%%%%%%%%%%%%%%%%%%%%%%%%%%%%%%%%%%%%%%%%%%%

The restriction idempotents of $\Par(\CTor)^{*}$ are determined by finite sets of torsor equations.  The restriction idempotents of $\CNOT$ also have a normal form,  as a conjunction of clauses: we call this the clausal form for restriction idempotents in $\CNOT$.  As torsor equations can be translated into clauses it follows that $\tilde H_0$ is full on restriction idempotents.   Furthermore, by showing  that one can perform Gaussian elimination on these 
clauses we prove that $\tilde{H}_0$ is faithful on restriction idempotents.

\begin{definition} For any $n \in \mathbb{N}$, $1 \leq i \leq n$, define the map $\mathsf{swap}_{(i,n)}: n \rightarrow n$ is a map in $\CNOT$ inductively as follows:
$$
\swap_{(i,n)} :=
\begin{cases}
1_n & \text{If } $i=0$\\
(  1_{n-(i+1)} \otimes \swap \otimes 1_{i-1}) (\swap_{(i+1,n)})  & \text{Otherwise}
\end{cases}
$$

\end{definition}

For example, consider the circuit $\swap_{(4,5)}$:

$$

$$
Note that $\swap_{(i,n)} \swap_{(i,n)}^\circ: n \rightarrow n$ is the identity.

\begin{definition}
Given any $n\in \mathbb{N}$ and $1\leq i \leq n$, define the \textbf{literal} $l_{i,n}$ to be the following:
$$
\swap_{(i,n+1)}  \ (\cnot \otimes 1_{n-2})\ \swap_{(i,n+1)} ^\circ
$$
\end{definition}

For example, consider the literal $l_{4,5}$:

$$
%
$$

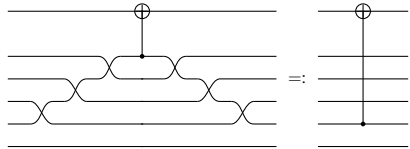
\begin{definition}
A \textbf{clause} $c: n \rightarrow n$ is a map in $\CNOT$ which is the composition of literals in the following form:
$$
c = (\zeroin \otimes 1_n)l_{i_1,n}l_{i_2,n} \cdots l_{i_m,n} (a \otimes 1_n)
$$
where $a: 1 \rightarrow 0 $ is either $\oneout$ or $\zeroout$.
\end{definition}

In a clause, the wire which begins with an input ancillary bit and ends with an output ancillary bit and on which literals act is called the \textbf{clause wire}. 
The following are examples of clauses in which the top wire is the clause wire:

$$
%
$$

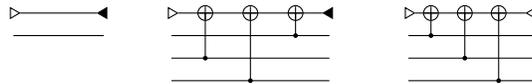
\begin{definition} \label{clausalform}
A map in $\CNOT$ is said to be in \textbf{clausal form} if it can be decomposed into a sequence of clauses.
\end{definition}

We wish now to show that every restriction idempotent in $\CNOT$ can be expressed in clausal form. To achieve this, the following observations are useful:

\begin{lemma}~
\label{3Results}
\begin{enumerate}[label=(\roman*)]
\item \begin{align*}
%
\end{align*}

\end{enumerate}

\end{lemma}
\begin{lemma}
\label{lemma:clause_idempotent}
In $\CNOT$ clauses are idempotent.

\end{lemma}
\begin{proof}
To show idempotence of a clause, we describe how to duplicate it.   First, using the naturality of $\Delta$, split the input ancilla bit on the clause wire  $\hat{0} = (\hat{0} \otimes  \hat{0})\nabla$.  Then by repeatedly using Lemma \ref{3Results} {\em (iii)}, copy all of the literals onto both of the new clause wires.  Then use the naturality of $\nabla$ on the output ancillary bit $\nabla \la b | := \la b | \otimes \la b |$ to split both clauses apart.

\end{proof}

\begin{proposition}
In $\CNOT$
\begin{enumerate}[label=(\roman*)]
\item Every restriction idempotent is equivalent to a circuit in clausal form.
\item Every circuit in clausal form is a restriction idempotent.
\end{enumerate}
\end{proposition}
\begin{proof}\
\begin{enumerate}[label=(\roman*)]
\item
%%%%%%%%%%%%%%%%%%%%%%%
Given a restriction idempotent $e:n\to n$ for some $n\in \mathbb{N}$, we prove $e$ is in clausal form by induction on the size of the circuit.

\begin{itemize}
\item $1_n$ is in clausal form as $1_n = 1_n\otimes 1_0 = 1_n \otimes \onein\oneout$.

\item Suppose inductively for all circuits $f$ such that $|f|<k$, the restriction idempotent $\bar f = ff\cnv$ is in clausal form. Consider some restriction idempotent $f:n\to n$ such that $|f|= k$.  Decompose $f$ into circuits $g$ and $h$ such that $|g|=1$, $|h|=k-1$ and $f=gh$.  Consider the three following cases.

\begin{itemize}
\item
Suppose that $\oneout \in g$.

Note that $\bar f = ff\cnv = ghh\cnv g\cnv = g\bar h g\cnv$.  The two circuits $g$ and $g\cnv$ form a clause, by Lemma \ref{3Results} {\em(iv)}.
Recall that $\bar h$ is in clausal form by supposition, as $|h|=k-1<k$; therefore, $\bar f$ is the composition of clauses, and thus a clause itself.

\item
Suppose that $\onein \in g$.

As $\bar f = ff\cnv = ghh\cnv g\cnv = g\bar h g\cnv$ and $\bar h$ is in clausal form by Lemma \ref{lemma:behaveasexpected}, we push $g$ and $g \cnv$ toward the middle through $\bar h$ until the two ancill\ae\ meet and annihilate each other by \ref{CNTii}.  This process may have turned some literals into $\mathsf{not}$ gates, so slide the {\sf not} gates to the right side of the clause wires with \ref{CNTv}.  This may toggle the output ancill\ae\ on some clause wires; however, as the output ancilla of a clause can be either $\zeroout$ or $\oneout$, $\bar f$ is in clausal form.

\item
Suppose that $\cnot \in g$.

Again, we push both $\cnot$ gates inward through $\bar h$.  If the $\cnot$ gates reach each other by $\ref{CNTii}$, we are done.  If this does not happen immediately, as $\bar h$ is in clausal form, there are two cases: the control or operating bits may be adjacent to a control bit of a clause.  In the first case, the two control bits commute by \ref{CNTiii}.  In the other case, by Lemma \ref{lemmas} \ref{lem:helperLemma}, we may add another literal to the clause and pass the operating bit through.
\end{itemize}
\end{itemize}
%%%%%%%%%%%%%%%%%%%%%%%
\item
Given a circuit $t = d_1 d_2 \cdots d_m$ in clausal form where $d_1, \cdots , d_m$ are clauses,
\[ tt = d_1 \cdots d_m d_1 \cdots d_m = d_1^2 d_2^2\cdots d_m ^2 = d_1 d_2\cdots d_m \]
as clauses commute by $\ref{CNTv}$ and are idempotent by Lemma \ref{lemma:clause_idempotent}.  Therefore, as $\CNOT$ is an inverse category $t$ is a restriction idempotent.
\end{enumerate}
\end{proof}

\begin{theorem}
\label{Theorem: restriction faithful}
$\tilde H_0: \CNOT \rightarrow \ParIso(\CTor_2)^*$ is full and faithful on restriction idempotents.
\end{theorem}
\begin{proof}
A restriction in $\Par(\Tor_2)$ is given by a span in which both legs are equal and, thus, monic.   Thus, restrictions correspond precisely to subobjects in $\Tor_2$.  However, these are determined by sets of torsor equations  of the form:
$$\left\{ \sum_j b_{i,j}  = a_i\right\}_i$$
Each equation, $\sum_j b_{i,j}  = a_i$, corresponds in turn to a clause which picks out the wires $b_{i,j}$ and has output ancillary bit $\la a_i |$.  This immediately means that $\tilde H_0$ is full on restriction idempotents.

Consider now an arbitrary restriction idempotent expressed as a circuit in clausal form,  under $\tilde H_0$ it corresponds to a set of equations.
We must show that two restriction idempotents in $\CNOT$, whose corresponding sets of equations are equivalent in $\CTor_2$, must be equal in $\CNOT$.   This amounts to showing that we can perform Gaussian elimination on 
clauses in $\CNOT$, as two sets of equations are equivalent in $\CTor_2$ if and only if they can be shown so by Gaussian elimination steps.

Given two clauses $c$ and $c'$ we show that we can perform the Gaussian elimination step
$$\{c,c'\} \mapsto \{c,c+c'\}$$
and maintain equality.

We first join the input ancill\ae\ of both clause wires into $(\zeroin\otimes\zeroin) := \zeroin \Delta$ using naturality of $\Delta$.

By \ref{CNTii}, we copy two copies of each literal in $c$ to the right of the input ancilla.  By  Lemma \ref{3Results} {\em (iii)} push one copy of each new literal through $\Delta$.  On one wire all of the literals will annihilate, and on the other only the common literals between $c$ and $c'$ will annihilate.  Use Lemma \ref{3Results} (i) and \ref{3Results} {\em (ii)} to split the literals to the left and right of $\Delta$.  This may have shifted the input ancilla of the second clause to be $\onein$.  In this case, use $\ref{CNTii}$ to push a {\sf not} gate from the left to right of the clause wire and negate the output ancillary bit of the second clause.  The result is two clauses corresponding to the Gaussian elimination step.  Therefore, we can perform Gaussian elimination on clauses in $\CNOT$.

For example, suppose we are given a circuit determined by the equations:
$$
\begin{bmatrix}
1 & 0 & 1\\
1 & 1 & 0
\end{bmatrix}
\begin{bmatrix}
x_1\\
x_2\\
x_3
\end{bmatrix}
=
\begin{bmatrix}
1\\
0
\end{bmatrix}
$$
We can perform Gaussian elimination as follows
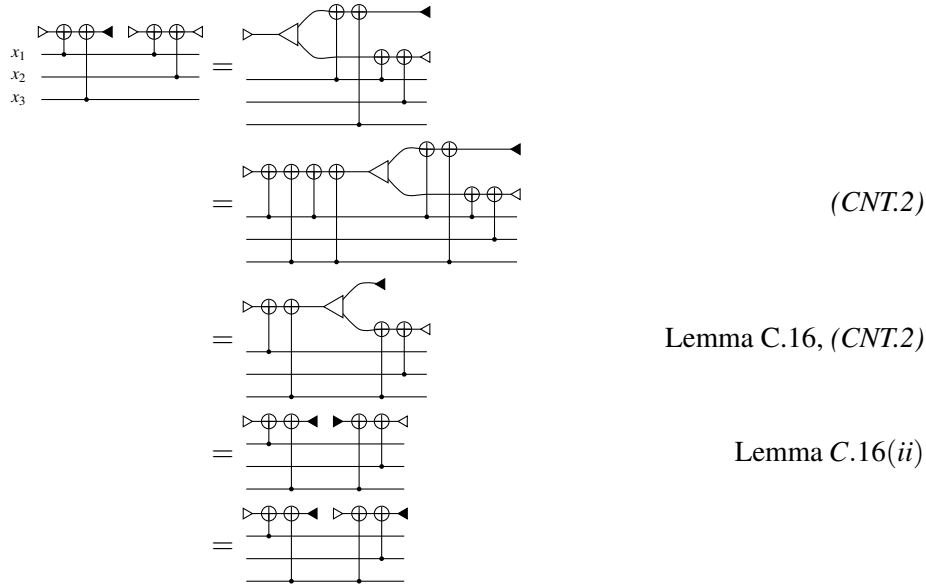
\begin{align*}
\begin{tikzpicture}[baseline={([yshift=-.5ex]current bounding box.center)}]
\begin{pgfonlayer}{nodelayer}
\node [style=rn] (-1) at (-.5, 1) {$x_1$};
\node [style=rn] (-2) at (-.5, .5) {$x_2$};
\node [style=rn] (-3) at (-.5, 0) {$x_3$};
%%%
\node [style=zeroin] (0) at (0, 1.5) {};
\node [style=oplus] (1) at (.5, 1.5) {};
\node [style=oplus] (2) at (1, 1.5) {};
\node [style=oneout] (3) at (1.5, 1.5) {};
\node [style=zeroin] (4) at (2, 1.5) {};
\node [style=oplus] (5) at (2.5, 1.5) {};
\node [style=oplus] (6) at (3, 1.5) {};
\node [style=zeroout] (7) at (3.5, 1.5) {};
\node [style=rn] (10) at (0, 1) {};
\node [style=dot] (11) at (.5, 1) {};
\node [style=rn] (12) at (1, 1) {};
\node [style=rn] (13) at (1.5, 1) {};
\node [style=rn] (14) at (2, 1) {};
\node [style=dot] (15) at (2.5, 1) {};
\node [style=rn] (16) at (3, 1) {};
\node [style=rn] (17) at (3.5, 1) {};
\node [style=rn] (20) at (0, .5) {};
\node [style=rn] (21) at (.5, .5) {};
\node [style=rn] (22) at (1, .5) {};
\node [style=rn] (23) at (1.5, .5) {};
\node [style=rn] (24) at (2, .5) {};
\node [style=rn] (25) at (2.5, .5) {};
\node [style=dot] (26) at (3, .5) {};
\node [style=rn] (27) at (3.5, .5) {};
\node [style=rn] (30) at (0, 0) {};
\node [style=rn] (31) at (.5, 0) {};
\node [style=dot] (32) at (1, 0) {};
\node [style=rn] (33) at (1.5, 0) {};
\node [style=rn] (34) at (2, 0) {};
\node [style=rn] (35) at (2.5, 0) {};
\node [style=rn] (36) at (3, 0) {};
\node [style=rn] (37) at (3.5, 0) {};
\end{pgfonlayer}
\begin{pgfonlayer}{edgelayer}
\draw plot [smooth, tension=1] coordinates {(0) (3)};
\draw plot [smooth, tension=1] coordinates {(4) (7)};
\draw plot [smooth, tension=1] coordinates {(10) (17)};
\draw plot [smooth, tension=1] coordinates {(20) (27)};
\draw plot [smooth, tension=1] coordinates {(30) (37)};
\draw plot [smooth, tension=1] coordinates {(1) (11)};
\draw plot [smooth, tension=1] coordinates {(5) (15)};
\draw plot [smooth, tension=1] coordinates {(2) (32)};
\draw plot [smooth, tension=1] coordinates {(6) (26)};
\end{pgfonlayer}
\end{tikzpicture}
&=
\begin{tikzpicture}[baseline={([yshift=-.5ex]current bounding box.center)}]
\begin{pgfonlayer}{nodelayer}
%%%
\node [style=zeroin] (0) at (0, 2) {};
\node [style=fanout] (1) at (1, 2) {};
\node [style=rn] (1a0) at (1.5, 2.5) {};
\node [style=oplus] (1a1) at (2, 2.5) {};
\node [style=oplus] (1a2) at (2.5, 2.5) {};
\node [style=rn] (1a3) at (3, 2.5) {};
\node [style=rn] (1a4) at (3.5, 2.5) {};
\node [style=oneout] (1a5) at (4, 2.5) {};
\node [style=rn] (1b0) at (1.5, 1.5) {};
\node [style=rn] (1b1) at (2, 1.5) {};
\node [style=rn] (1b2) at (2.5, 1.5) {};
\node [style=oplus] (1b3) at (3, 1.5) {};
\node [style=oplus] (1b4) at (3.5, 1.5) {};
\node [style=zeroout] (1b5) at (4, 1.5) {};
\node [style=rn] (10) at (0, 1) {};
\node [style=dot] (11) at (2, 1) {};
\node [style=dot] (12) at (3, 1) {};
\node [style=rn] (13) at (4, 1) {};
\node [style=rn] (20) at (0, .5) {};
\node [style=dot] (21) at (3.5, .5) {};
\node [style=rn] (22) at (4, .5) {};
\node [style=rn] (30) at (0, 0) {};
\node [style=dot] (31) at (2.5, 0) {};
\node [style=rn] (32) at (4, 0) {};
\end{pgfonlayer}
\begin{pgfonlayer}{edgelayer}
\draw plot [smooth, tension=.5] coordinates {(1) (0)};
\draw plot [smooth, tension=.5] coordinates {(1) (1a0) (1a1) (1a2) (1a3) (1a4) (1a5)};
\draw plot [smooth, tension=.5] coordinates {(1) (1b0) (1b1) (1b2) (1b3) (1b4) (1b5)};
\draw plot [smooth, tension=.5] coordinates {(10) (13)};
\draw plot [smooth, tension=.5] coordinates {(20) (22)};
\draw plot [smooth, tension=.5] coordinates {(30) (32)};
\draw plot [smooth, tension=.5] coordinates {(11) (1a1)};
\draw plot [smooth, tension=.5] coordinates {(12) (1b3)};
\draw plot [smooth, tension=.5] coordinates {(21) (1b4)};
\draw plot [smooth, tension=.5] coordinates {(31) (1a2)};
\end{pgfonlayer}
\end{tikzpicture}&\\
&=
\begin{tikzpicture}[baseline={([yshift=-.5ex]current bounding box.center)}]
\begin{pgfonlayer}{nodelayer}
%%%
\node [style=zeroin] (-4) at (-2, 2) {};
\node [style=oplus] (-3) at (-1.5, 2) {};
\node[style=oplus] (-2) at (-1, 2) {};
\node [style=oplus] (-1) at (-.5, 2) {};
\node [style=oplus] (0) at (0, 2) {};
\node [style=fanout] (1) at (1, 2) {};
\node [style=rn] (1a0) at (1.5, 2.5) {};
\node [style=oplus] (1a1) at (2, 2.5) {};
\node [style=oplus] (1a2) at (2.5, 2.5) {};
\node [style=rn] (1a3) at (3, 2.5) {};
\node [style=rn] (1a4) at (3.5, 2.5) {};
\node [style=oneout] (1a5) at (4, 2.5) {};
\node [style=rn] (1b0) at (1.5, 1.5) {};
\node [style=rn] (1b1) at (2, 1.5) {};
\node [style=rn] (1b2) at (2.5, 1.5) {};
\node [style=oplus] (1b3) at (3, 1.5) {};
\node [style=oplus] (1b4) at (3.5, 1.5) {};
\node [style=zeroout] (1b5) at (4, 1.5) {};
\node [style=rn] (-13) at (-2, 1) {};
\node [style=dot] (-12) at (-1.5, 1) {};
\node [style=dot] (-11) at (-.5, 1) {};
\node [style=rn] (10) at (0, 1) {};
\node [style=dot] (11) at (2, 1) {};
\node [style=dot] (12) at (3, 1) {};
\node [style=rn] (13) at (4, 1) {};
\node [style=rn] (-21) at (-2, .5) {};
\node [style=rn] (20) at (0, .5) {};
\node [style=dot] (21) at (3.5, .5) {};
\node [style=rn] (22) at (4, .5) {};
\node [style=rn] (-33) at (-2, 0) {};
\node [style=dot] (-32) at (-1, 0) {};
\node [style=dot] (-31) at (0, 0) {};
\node [style=rn] (30) at (0, 0) {};
\node [style=dot] (31) at (2.5, 0) {};
\node [style=rn] (32) at (4, 0) {};
\end{pgfonlayer}
\begin{pgfonlayer}{edgelayer}
\draw plot [smooth, tension=.5] coordinates {(-4) (1)};
\draw plot [smooth, tension=.5] coordinates {(1) (1a0) (1a1) (1a2) (1a3) (1a4) (1a5)};
\draw plot [smooth, tension=.5] coordinates {(1) (1b0) (1b1) (1b2) (1b3) (1b4) (1b5)};
\draw plot [smooth, tension=.5] coordinates {(-13) (13)};
\draw plot [smooth, tension=.5] coordinates {(-21) (22)};
\draw plot [smooth, tension=.5] coordinates {(-33) (32)};
\draw plot [smooth, tension=.5] coordinates {(11) (1a1)};
\draw plot [smooth, tension=.5] coordinates {(12) (1b3)};
\draw plot [smooth, tension=.5] coordinates {(21) (1b4)};
\draw plot [smooth, tension=.5] coordinates {(31) (1a2)};
\draw plot [smooth, tension=.5] coordinates {(-3) (-12)};
\draw plot [smooth, tension=.5] coordinates {(-2) (-32)};
\draw plot [smooth, tension=.5] coordinates {(-1) (-11)};
\draw plot [smooth, tension=.5] coordinates {(0) (-31)};
\end{pgfonlayer}
\end{tikzpicture}&\ref{CNTii}\\
&=
\begin{tikzpicture}[baseline={([yshift=-.5ex]current bounding box.center)}]
\begin{pgfonlayer}{nodelayer}
%%%
\node [style=zeroin] (-2) at (0, 2) {};
\node [style=oplus] (-1) at (.5, 2) {};
\node [style=oplus] (0) at (1, 2) {};
\node [style=fanout] (1) at (2, 2) {};
\node [style=rn] (1a0) at (2.5, 2.5) {};
\node [style=oneout] (1a1) at (3, 2.5) {};
\node [style=rn] (1b0) at (2.5, 1.5) {};
\node [style=oplus] (1b1) at (3, 1.5) {};
\node [style=oplus] (1b2) at (3.5, 1.5) {};
\node [style=zeroout] (1b3) at (4, 1.5) {};
\node [style=rn] (10) at (0, 1) {};
\node [style=dot] (11) at (.5, 1) {};
\node [style=rn] (12) at (4, 1) {};
\node [style=rn] (20) at (0, .5) {};
\node [style=dot] (22) at (3.5, .5) {};
\node [style=rn] (23) at (4, .5) {};
\node [style=rn] (30) at (0, 0) {};
\node [style=dot] (31) at (1, 0) {};
\node [style=dot] (32) at (3, 0) {};
\node [style=rn] (33) at (4, 0) {};
\end{pgfonlayer}
\begin{pgfonlayer}{edgelayer}
\draw plot [smooth, tension=.5] coordinates {(-2) (-1) (0) (1)};
\draw plot [smooth, tension=.5] coordinates {(1) (1a0) (1a1)};
\draw plot [smooth, tension=.5] coordinates {(1) (1b0) (1b1) (1b2) (1b3)};
\draw plot [smooth, tension=.5] coordinates {(10) (13)};
\draw plot [smooth, tension=.5] coordinates {(20) (23)};
\draw plot [smooth, tension=.5] coordinates {(30) (33)};
\draw plot [smooth, tension=.5] coordinates {(11) (-1)};
\draw plot [smooth, tension=.5] coordinates {(31) (0)};
\draw plot [smooth, tension=.5] coordinates {(32) (1b1)};
\draw plot [smooth, tension=.5] coordinates {(22) (1b2)};
\end{pgfonlayer}
\end{tikzpicture}&\text{Lemma \ref{3Results}, \ref{CNTii}}\\
&=
\begin{tikzpicture}[baseline={([yshift=-.5ex]current bounding box.center)}]
\begin{pgfonlayer}{nodelayer}
\node [style=zeroin] (0) at (0, 1.5) {};
\node [style=oplus] (1) at (.5, 1.5) {};
\node [style=oplus] (2) at (1, 1.5) {};
\node [style=oneout] (3) at (1.5, 1.5) {};
\node [style=onein] (4) at (2, 1.5) {};
\node [style=oplus] (5) at (2.5, 1.5) {};
\node [style=oplus] (6) at (3, 1.5) {};
\node [style=zeroout] (7) at (3.5, 1.5) {};
\node [style=rn] (10) at (0, 1) {};
\node [style=dot] (11) at (.5, 1) {};
\node [style=rn] (12) at (1, 1) {};
\node [style=rn] (13) at (1.5, 1) {};
\node [style=rn] (14) at (2, 1) {};
\node [style=rn] (15) at (2.5, 1) {};
\node [style=rn] (16) at (3, 1) {};
\node [style=rn] (17) at (3.5, 1) {};
\node [style=rn] (20) at (0, .5) {};
\node [style=rn] (21) at (.5, .5) {};
\node [style=rn] (22) at (1, .5) {};
\node [style=rn] (23) at (1.5, .5) {};
\node [style=rn] (24) at (2, .5) {};
\node [style=rn] (25) at (2.5, .5) {};
\node [style=dot] (26) at (3, .5) {};
\node [style=rn] (27) at (3.5, .5) {};
\node [style=rn] (30) at (0, 0) {};
\node [style=rn] (31) at (.5, 0) {};
\node [style=dot] (32) at (1, 0) {};
\node [style=rn] (33) at (1.5, 0) {};
\node [style=rn] (34) at (2, 0) {};
\node [style=dot] (35) at (2.5, 0) {};
\node [style=rn] (36) at (3, 0) {};
\node [style=rn] (37) at (3.5, 0) {};
\end{pgfonlayer}
\begin{pgfonlayer}{edgelayer}
\draw plot [smooth, tension=1] coordinates {(0) (3)};
\draw plot [smooth, tension=1] coordinates {(4) (7)};
\draw plot [smooth, tension=1] coordinates {(10) (17)};
\draw plot [smooth, tension=1] coordinates {(20) (27)};
\draw plot [smooth, tension=1] coordinates {(30) (37)};
\draw plot [smooth, tension=1] coordinates {(1) (11)};
\draw plot [smooth, tension=1] coordinates {(5) (35)};
\draw plot [smooth, tension=1] coordinates {(2) (32)};
\draw plot [smooth, tension=1] coordinates {(6) (26)};
\end{pgfonlayer}
\end{tikzpicture} & \text{Lemma } \ref{3Results} (ii)\\
&=
\begin{tikzpicture}[baseline={([yshift=-.5ex]current bounding box.center)}]
\begin{pgfonlayer}{nodelayer}
%%%
\node [style=zeroin] (0) at (0, 1.5) {};
\node [style=oplus] (1) at (.5, 1.5) {};
\node [style=oplus] (2) at (1, 1.5) {};
\node [style=oneout] (3) at (1.5, 1.5) {};
\node [style=zeroin] (4) at (2, 1.5) {};
\node [style=oplus] (5) at (2.5, 1.5) {};
\node [style=oplus] (6) at (3, 1.5) {};
\node [style=oneout] (7) at (3.5, 1.5) {};
\node [style=rn] (10) at (0, 1) {};
\node [style=dot] (11) at (.5, 1) {};
\node [style=rn] (12) at (1, 1) {};
\node [style=rn] (13) at (1.5, 1) {};
\node [style=rn] (14) at (2, 1) {};
\node [style=rn] (15) at (2.5, 1) {};
\node [style=rn] (16) at (3, 1) {};
\node [style=rn] (17) at (3.5, 1) {};
\node [style=rn] (20) at (0, .5) {};
\node [style=rn] (21) at (.5, .5) {};
\node [style=rn] (22) at (1, .5) {};
\node [style=rn] (23) at (1.5, .5) {};
\node [style=rn] (24) at (2, .5) {};
\node [style=rn] (25) at (2.5, .5) {};
\node [style=dot] (26) at (3, .5) {};
\node [style=rn] (27) at (3.5, .5) {};
\node [style=rn] (30) at (0, 0) {};
\node [style=rn] (31) at (.5, 0) {};
\node [style=dot] (32) at (1, 0) {};
\node [style=rn] (33) at (1.5, 0) {};
\node [style=rn] (34) at (2, 0) {};
\node [style=dot] (35) at (2.5, 0) {};
\node [style=rn] (36) at (3, 0) {};
\node [style=rn] (37) at (3.5, 0) {};
\end{pgfonlayer}
\begin{pgfonlayer}{edgelayer}
\draw plot [smooth, tension=1] coordinates {(0) (3)};
\draw plot [smooth, tension=1] coordinates {(4) (7)};
\draw plot [smooth, tension=1] coordinates {(10) (17)};
\draw plot [smooth, tension=1] coordinates {(20) (27)};
\draw plot [smooth, tension=1] coordinates {(30) (37)};
\draw plot [smooth, tension=1] coordinates {(1) (11)};
\draw plot [smooth, tension=1] coordinates {(5) (35)};
\draw plot [smooth, tension=1] coordinates {(2) (32)};
\draw plot [smooth, tension=1] coordinates {(6) (26)};
\end{pgfonlayer}
\end{tikzpicture}
\end{align*}

Which represents the reduced system of linear equations:
$$
\begin{bmatrix}
1 & 0 & 1\\
0 & 1 & 1
\end{bmatrix}
\begin{bmatrix}
x_1\\
x_2\\
x_3
\end{bmatrix}
=
\begin{bmatrix}
1\\
1
\end{bmatrix}
$$

Hence, if the image of two circuits are equal under the functor $\tilde H_0$, then they are the same.
\end{proof}

\subsection{\texorpdfstring{$\tilde H_0: \CNOT \to \ParIso(\CTor_2)^*$}{H0} is essentially surjective}

%%%%%%%%%%%%%%%%%%%%%%%%%%%%%%%%%%%%%%%%%%%%%%%%%%%%%%%%%%%%%%%%%%%%%%%%%%%%%%%%%%%%

In order to prove that $\tilde H_0$ is essentially surjective, we invoke the alternative characterization of $\CTor_2$ given by  Proposition \ref{rem:tor2}.

\begin{proposition}
\label{prop:essentiallysurjective}
$\tilde H_0$ is essentially surjective.
\end{proposition}
\begin{proof}
Consider any torsor $(X,\times)$ in $\ParIso(\CTor_2)^*$.  There is some $n \in \mathbb{N}$ such that $(X,\times)\cong(\mathbb{Z}_2^n,\_\oplus\_\oplus\_)$ by Proposition \ref{rem:tor2}.  However, $\Total(\CNOT)(0,n)= \mathbb{Z}_2^n$.
\end{proof}

%%%%%%%%%%%%%%%%%%%%%%%%%%%%%%%%%%%%%%%%%%%%%%%%%%%%%%%%%%%%%%%%%%%%%%%%%%%%%%%%%%%%

\subsection{\texorpdfstring{$\tilde H_0: \CNOT \to \ParIso(\CTor_2)^{*}$}{H0} is full}

%%%%%%%%%%%%%%%%%%%%%%%%%%%%%%%%%%%%%%%%%%%%%%%%%%%%%%%%%%%%%%%%%%%%%%%%%%%%%%%%%%%%

In order to prove that $\tilde H_0$ is full, we will prove two useful results:

\begin{lemma}
\label{lemma:fullCopy}
Let $F:\X \to \Y$ be an inverse product preserving functor between discrete inverse categories. Let $f$ be a partial isomorphism in $\mathbb{Y}$. If $\la \bar{f}, f \ra := \Delta(\bar f \otimes {f})$ and $\la \bar{f\cnv}, f\cnv \ra:=  \Delta(\bar{f\cnv} \otimes {f\cnv})$ are in the image of $F$, then $f$ and $f\cnv$ are also in the image of $F$.
\end{lemma}

\begin{proof}  Observe that in $\mathbb{Y}$ we have:
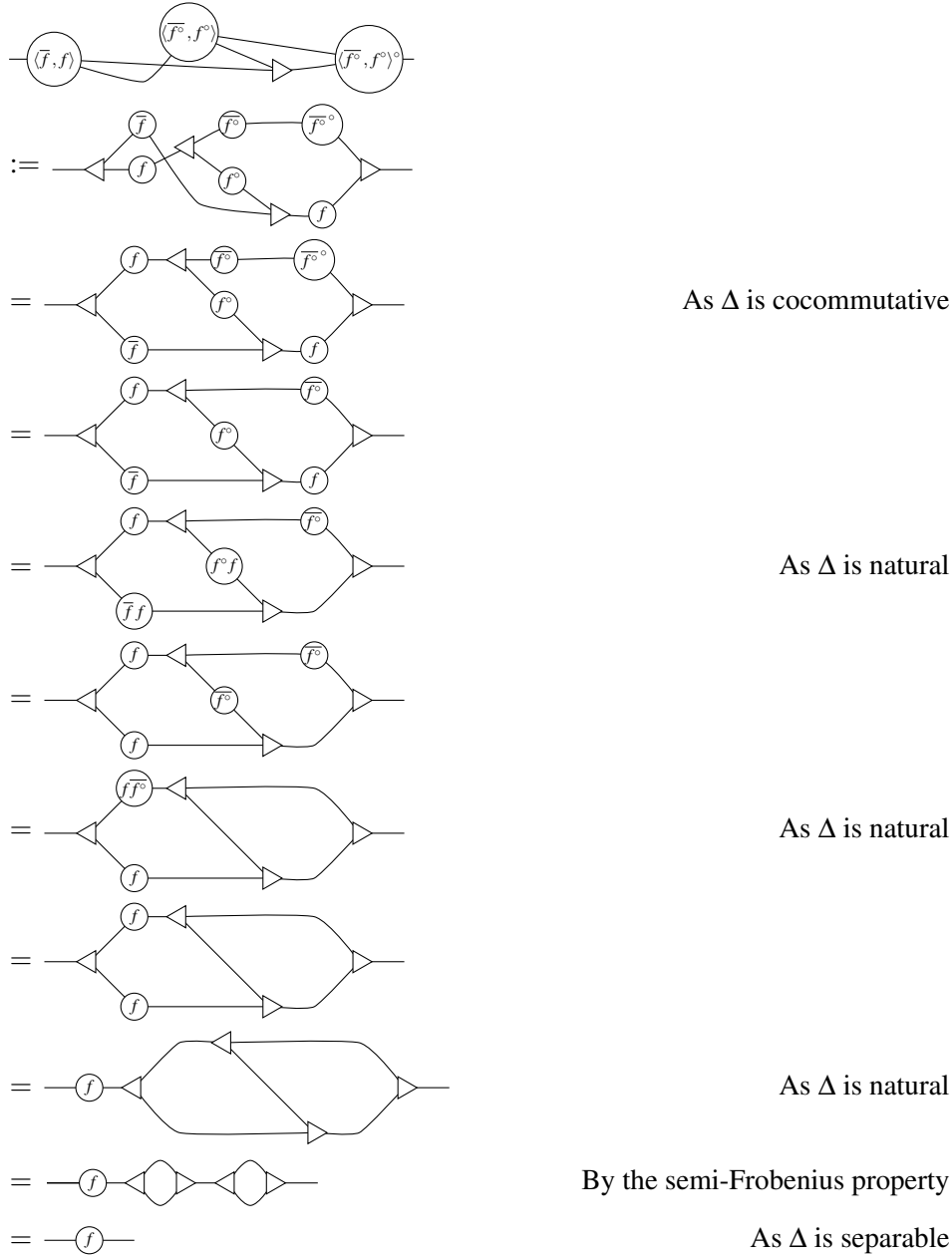
\begin{align*}
&
\begin{tikzpicture}[baseline={([yshift=-.5ex]current bounding box.center)}]
\begin{pgfonlayer}{nodelayer}
\node [style=rn] (0) at (0, 1) {};
\draw[fill=white] (1,1) circle [radius=.6] node (1) {$\la\bar{f}, f \ra$};
\node [style=rn] (2) at (3, 2) {};
\node [style=rn] (3) at (3, .5) {};
\draw[fill=white] (4,1.6) circle [radius=.65] node (4) {$\la\bar{f\cnv}, {f\cnv} \ra$};
\node [style=fanin] (5) at (6, .75) {};
\draw[fill=white] (8,1) circle [radius=.75] node (6) {$\la\bar{f\cnv}, {f\cnv} \ra\cnv$};
\node [style=rn] (7) at (9, 1) {};
\end{pgfonlayer}
\begin{pgfonlayer}{edgelayer}
\draw plot [smooth, tension=0.3] coordinates { (0) (1)};
\draw plot [smooth, tension=0.3] coordinates { (1) (3) (4)};
\draw plot [smooth, tension=0.3] coordinates { (1) (5) (6)};
\draw plot [smooth, tension=0.3] coordinates { (6) (7)};
\draw plot [smooth, tension=0.3] coordinates { (4) (5)};
\draw plot [smooth, tension=0.3] coordinates { (4) (6)};
\end{pgfonlayer}
\end{tikzpicture}\\
&:=
\begin{tikzpicture}[baseline={([yshift=-.5ex]current bounding box.center)}]
\begin{pgfonlayer}{nodelayer}
\node [style=rn] (0) at (0, 1) {};
\node [style=fanout] (1) at (1, 1) {};
\draw[fill=white] (2,1) circle [radius=.3] node (2) {$f$};
\draw[fill=white] (2,2) circle [radius=.3] node (3) {$\bar{f}$};
\node [style=rn] (4) at (3.25, .25) {};
\node [style=fanout] (5) at (3, 1.5) {};
\draw[fill=white] (4,2) circle [radius=.3] node (6) {$\bar {f\cnv}$};
\draw[fill=white] (4,.75) circle [radius=.3] node (7) {${f\cnv}$};
\node [style=fanin] (8) at (5, 0) {};
\draw[fill=white] (6,2) circle [radius=.45] node (9) {${\bar{f\cnv}}\cnv$};
\draw[fill=white] (6,0) circle [radius=.3] node (10) {$f$};
\node [style=fanin] (11) at (7, 1) {};
\node [style=rn] (12) at (8, 1) {};
\end{pgfonlayer}
\begin{pgfonlayer}{edgelayer}
\draw plot [smooth, tension=0.3] coordinates { (0) (1)};
\draw plot [smooth, tension=0.3] coordinates { (1) (2)};
\draw plot [smooth, tension=0.3] coordinates { (1) (3)};
\draw plot [smooth, tension=0.3] coordinates { (2) (5)};
\draw plot [smooth, tension=0.3] coordinates { (3) (4) (8)};
\draw plot [smooth, tension=0.3] coordinates { (5) (7) (8)};
\draw plot [smooth, tension=0.3] coordinates { (5) (6) (9) (11)};
\draw plot [smooth, tension=0.3] coordinates { (8) (10) (11)};
\draw plot [smooth, tension=0.3] coordinates { (11) (12)};
\end{pgfonlayer}
\end{tikzpicture}\\
&=
\begin{tikzpicture}[baseline={([yshift=-.5ex]current bounding box.center)}]
\begin{pgfonlayer}{nodelayer}
\node [style=rn] (0) at (0, 1) {};
\node [style=fanout] (1) at (1, 1) {};
\draw[fill=white] (2,2) circle [radius=.3] node (2) {$ f$};
\draw[fill=white] (2,0) circle [radius=.3] node (3) {$\bar{f}$};
\node [style=rn] (4) at (3.25, .25) {};
\node [style=fanout] (5) at (3, 2) {};
\draw[fill=white] (4,2) circle [radius=.3] node (6) {$\bar {f\cnv}$};
\draw[fill=white] (4,1) circle [radius=.3] node (7) {${f\cnv}$};
\node [style=fanin] (8) at (5, 0) {};
\draw[fill=white] (6,2) circle [radius=.45] node (9) {${\bar {f\cnv}}\cnv$};
\draw[fill=white] (6,0) circle [radius=.3] node (10) {$f$};
\node [style=fanin] (11) at (7, 1) {};
\node [style=rn] (12) at (8, 1) {};
\end{pgfonlayer}
\begin{pgfonlayer}{edgelayer}
\draw plot [smooth, tension=0.3] coordinates { (0) (1)};
\draw plot [smooth, tension=0.3] coordinates { (1) (2)};
\draw plot [smooth, tension=0.3] coordinates { (1) (3)};
\draw plot [smooth, tension=0.3] coordinates { (2) (5)};
\draw plot [smooth, tension=0.3] coordinates { (3) (8)};
\draw plot [smooth, tension=0.3] coordinates { (5) (7) (8)};
\draw plot [smooth, tension=0.3] coordinates { (5) (6) (9) (11)};
\draw plot [smooth, tension=0.3] coordinates { (8) (10) (11)};
\draw plot [smooth, tension=0.3] coordinates { (11) (12)};
\end{pgfonlayer}
\end{tikzpicture}&\text{As $\Delta$ is cocommutative}\\
&=
\begin{tikzpicture}[baseline={([yshift=-.5ex]current bounding box.center)}]
\begin{pgfonlayer}{nodelayer}
\node [style=rn] (0) at (0, 1) {};
\node [style=fanout] (1) at (1, 1) {};
\draw[fill=white] (2,2) circle [radius=.3] node (2) {$ f$};
\draw[fill=white] (2,0) circle [radius=.3] node (3) {$\bar{f}$};
\node [style=rn] (4) at (3.25, .25) {};
\node [style=fanout] (5) at (3, 2) {};
\draw[fill=white] (4,1) circle [radius=.3] node (7) {$f\cnv$};
\node [style=fanin] (8) at (5, 0) {};
\draw[fill=white] (6,2) circle [radius=.3] node (9) {$\bar{f\cnv}$};
\draw[fill=white] (6,0) circle [radius=.3] node (10) {$f$};
\node [style=fanin] (11) at (7, 1) {};
\node [style=rn] (12) at (8, 1) {};
\end{pgfonlayer}
\begin{pgfonlayer}{edgelayer}
\draw plot [smooth, tension=0.3] coordinates { (0) (1)};
\draw plot [smooth, tension=0.3] coordinates { (1) (2)};
\draw plot [smooth, tension=0.3] coordinates { (1) (3)};
\draw plot [smooth, tension=0.3] coordinates { (2) (5)};
\draw plot [smooth, tension=0.3] coordinates { (3) (8)};
\draw plot [smooth, tension=0.3] coordinates { (5) (7) (8)};
\draw plot [smooth, tension=0.3] coordinates { (5) (9) (11)};
\draw plot [smooth, tension=0.3] coordinates { (8) (10) (11)};
\draw plot [smooth, tension=0.3] coordinates { (11) (12)};
\end{pgfonlayer}
\end{tikzpicture}\\
&=
\begin{tikzpicture}[baseline={([yshift=-.5ex]current bounding box.center)}]
\begin{pgfonlayer}{nodelayer}
\node [style=rn] (0) at (0, 1) {};
\node [style=fanout] (1) at (1, 1) {};
\draw[fill=white] (2,2) circle [radius=.3] node (2) {$ f$};
\draw[fill=white] (2,0) circle [radius=.4] node (3) {$\bar{f}f$};
\node [style=rn] (4) at (3.25, .25) {};
\node [style=fanout] (5) at (3, 2) {};
\draw[fill=white] (4,1) circle [radius=.4] node (7) {$f\cnv f$};
\node [style=fanin] (8) at (5, 0) {};
\draw[fill=white] (6,2) circle [radius=.3] node (9) {$\bar{f\cnv}$};
\node [style=rn] (10) at (6,0) {};
\node [style=fanin] (11) at (7, 1) {};
\node [style=rn] (12) at (8, 1) {};
\end{pgfonlayer}
\begin{pgfonlayer}{edgelayer}
\draw plot [smooth, tension=0.3] coordinates { (0) (1)};
\draw plot [smooth, tension=0.3] coordinates { (1) (2)};
\draw plot [smooth, tension=0.3] coordinates { (1) (3)};
\draw plot [smooth, tension=0.3] coordinates { (2) (5)};
\draw plot [smooth, tension=0.3] coordinates { (3) (8)};
\draw plot [smooth, tension=0.3] coordinates { (5) (7) (8)};
\draw plot [smooth, tension=0.3] coordinates { (5) (9) (11)};
\draw plot [smooth, tension=0.3] coordinates { (8) (10) (11)};
\draw plot [smooth, tension=0.3] coordinates { (11) (12)};
\end{pgfonlayer}
\end{tikzpicture} &\text{As $\Delta$ is natural}\\
&=
\begin{tikzpicture}[baseline={([yshift=-.5ex]current bounding box.center)}]
\begin{pgfonlayer}{nodelayer}
\node [style=rn] (0) at (0, 1) {};
\node [style=fanout] (1) at (1, 1) {};
\draw[fill=white] (2,2) circle [radius=.3] node (2) {$f$};
\draw[fill=white] (2,0) circle [radius=.3] node (3) {$f$};
\node [style=rn] (4) at (3.25, .25) {};
\node [style=fanout] (5) at (3, 2) {};
\draw[fill=white] (4,1) circle [radius=.3] node (7) {$\bar{f\cnv}$};
\node [style=fanin] (8) at (5, 0) {};
\draw[fill=white] (6,2) circle [radius=.3] node (9) {$\bar{f\cnv} $};
\node [style=rn] (10) at (6,0) {};
\node [style=fanin] (11) at (7, 1) {};
\node [style=rn] (12) at (8, 1) {};
\end{pgfonlayer}
\begin{pgfonlayer}{edgelayer}
\draw plot [smooth, tension=0.3] coordinates { (0) (1)};
\draw plot [smooth, tension=0.3] coordinates { (1) (2)};
\draw plot [smooth, tension=0.3] coordinates { (1) (3)};
\draw plot [smooth, tension=0.3] coordinates { (2) (5)};
\draw plot [smooth, tension=0.3] coordinates { (3) (8)};
\draw plot [smooth, tension=0.3] coordinates { (5) (7) (8)};
\draw plot [smooth, tension=0.3] coordinates { (5) (9) (11)};
\draw plot [smooth, tension=0.3] coordinates { (8) (10) (11)};
\draw plot [smooth, tension=0.3] coordinates { (11) (12)};
\end{pgfonlayer}
\end{tikzpicture}\\
&=
\begin{tikzpicture}[baseline={([yshift=-.5ex]current bounding box.center)}]
\begin{pgfonlayer}{nodelayer}
\node [style=rn] (0) at (0, 1) {};
\node [style=fanout] (1) at (1, 1) {};
\draw[fill=white] (2,2) circle [radius=.4] node (2) {$ f\bar{f\cnv}$};
\draw[fill=white] (2,0) circle [radius=.3] node (3) {$f$};
\node [style=rn] (4) at (3.25, .25) {};
\node [style=fanout] (5) at (3, 2) {};
\node [style=rn] (7) at (4,1) {};
\node [style=fanin] (8) at (5, 0) {};
\node [style=rn] (9) at (6,2) {};
\node [style=rn] (10) at (6,0) {};
\node [style=fanin] (11) at (7, 1) {};
\node [style=rn] (12) at (8, 1) {};
\end{pgfonlayer}
\begin{pgfonlayer}{edgelayer}
\draw plot [smooth, tension=0.3] coordinates { (0) (1)};
\draw plot [smooth, tension=0.3] coordinates { (1) (2)};
\draw plot [smooth, tension=0.3] coordinates { (1) (3)};
\draw plot [smooth, tension=0.3] coordinates { (2) (5)};
\draw plot [smooth, tension=0.3] coordinates { (3) (8)};
\draw plot [smooth, tension=0.3] coordinates { (5) (7) (8)};
\draw plot [smooth, tension=0.3] coordinates { (5) (9) (11)};
\draw plot [smooth, tension=0.3] coordinates { (8) (10) (11)};
\draw plot [smooth, tension=0.3] coordinates { (11) (12)};
\end{pgfonlayer}
\end{tikzpicture}&\text{As $\Delta$ is natural}\\
&=
\begin{tikzpicture}[baseline={([yshift=-.5ex]current bounding box.center)}]
\begin{pgfonlayer}{nodelayer}
\node [style=rn] (0) at (0, 1) {};
\node [style=fanout] (1) at (1, 1) {};
\draw[fill=white] (2,2) circle [radius=.3] node (2) {$ f$};
\draw[fill=white] (2,0) circle [radius=.3] node (3) {$f$};
\node [style=rn] (4) at (3.25, .25) {};
\node [style=fanout] (5) at (3, 2) {};
\node [style=rn] (7) at (4,1) {};
\node [style=fanin] (8) at (5, 0) {};
\node [style=rn] (9) at (6,2) {};
\node [style=rn] (10) at (6,0) {};
\node [style=fanin] (11) at (7, 1) {};
\node [style=rn] (12) at (8, 1) {};
\end{pgfonlayer}
\begin{pgfonlayer}{edgelayer}
\draw plot [smooth, tension=0.3] coordinates { (0) (1)};
\draw plot [smooth, tension=0.3] coordinates { (1) (2)};
\draw plot [smooth, tension=0.3] coordinates { (1) (3)};
\draw plot [smooth, tension=0.3] coordinates { (2) (5)};
\draw plot [smooth, tension=0.3] coordinates { (3) (8)};
\draw plot [smooth, tension=0.3] coordinates { (5) (7) (8)};
\draw plot [smooth, tension=0.3] coordinates { (5) (9) (11)};
\draw plot [smooth, tension=0.3] coordinates { (8) (10) (11)};
\draw plot [smooth, tension=0.3] coordinates { (11) (12)};
\end{pgfonlayer}
\end{tikzpicture}\\
&=
\begin{tikzpicture}[baseline={([yshift=-.5ex]current bounding box.center)}]
\begin{pgfonlayer}{nodelayer}
\node [style=rn] (-1) at (-1, 1) {};
\draw[fill=white] (0,1) circle [radius=.3] node (0) {$ f$};
\node [style=fanout] (1) at (1, 1) {};
\node [style=rn] (2) at (2,2) {};
\node [style=rn] (3) at (2,0) {};
\node [style=rn] (4) at (3.25, .25) {};
\node [style=fanout] (5) at (3, 2) {};
\node [style=rn] (7) at (4,1) {};
\node [style=fanin] (8) at (5, 0) {};
\node [style=rn] (9) at (6,2) {};
\node [style=rn] (10) at (6,0) {};
\node [style=fanin] (11) at (7, 1) {};
\node [style=rn] (12) at (8, 1) {};
\end{pgfonlayer}
\begin{pgfonlayer}{edgelayer}
\draw plot [smooth, tension=0.3] coordinates { (-1) (0) (1)};
\draw plot [smooth, tension=0.3] coordinates { (1) (2) (5)};
\draw plot [smooth, tension=0.3] coordinates { (1) (3) (8)};
\draw plot [smooth, tension=0.3] coordinates { (5) (7) (8)};
\draw plot [smooth, tension=0.3] coordinates { (5) (9) (11)};
\draw plot [smooth, tension=0.3] coordinates { (8) (10) (11)};
\draw plot [smooth, tension=0.3] coordinates { (11) (12)};
\end{pgfonlayer}
\end{tikzpicture}&\text{As $\Delta$ is natural}\\
&=
\begin{tikzpicture}[baseline={([yshift=-.5ex]current bounding box.center)}]
\begin{pgfonlayer}{nodelayer}
\node [style=rn] (0) at (-1, 1) {};
\draw[fill=white] (0,1) circle [radius=.3] node (1) {$ f$};
\node [style=fanout] (2) at (1, 1) {};
\node [style=fanin] (3) at (2, 1) {};
\node [style=fanout] (4) at (3, 1) {};
\node [style=fanin] (5) at (4, 1) {};
\node [style=rn] (6) at (5, 1) {};
\node [style=rn] (10) at (1.5, 1.5) {};
\node [style=rn] (11) at (3.5, 1.5) {};
\node [style=rn] (20) at (1.5, .5) {};
\node [style=rn] (21) at (3.5, .5) {};
\end{pgfonlayer}
\begin{pgfonlayer}{edgelayer}
\draw plot [smooth, tension=.7] coordinates {(1) (0) (2)};
\draw plot [smooth, tension=.7] coordinates {(3) (4)};
\draw plot [smooth, tension=.7] coordinates {(5) (6)};
\draw plot [smooth, tension=.7] coordinates {(2) (10) (3)};
\draw plot [smooth, tension=.7] coordinates {(2) (20) (3)};
\draw plot [smooth, tension=.7] coordinates {(4) (11) (5)};
\draw plot [smooth, tension=.7] coordinates {(4) (21) (5)};
\end{pgfonlayer}
\end{tikzpicture}&\text{By the semi-Frobenius property}\\
&=
\begin{tikzpicture}[baseline={([yshift=-.5ex]current bounding box.center)}]
\begin{pgfonlayer}{nodelayer}
\node [style=rn] (0) at (-1, 1) {};
\draw[fill=white] (0,1) circle [radius=.3] node (1) {$ f$};
\node [style=rn] (2) at (1, 1) {};
\end{pgfonlayer}
\begin{pgfonlayer}{edgelayer}
\draw plot [smooth, tension=.7] coordinates {(0) (1) (2)};
\end{pgfonlayer}
\end{tikzpicture}&\text{As $\Delta$ is separable}
\end{align*}
Therefore, $f \in F(\mathbb{X})$ and by symmetry $f\cnv \in F(\mathbb{X})$ as well.

\end{proof}

\begin{lemma}
\label{simulating_total_maps}
If $f\in \CTor_2(\mathbb{Z}_2^n, \mathbb{Z}_2^m)$, then there is a map $g \in \CNOT(n,m)$ with $\tilde H_0(g) = \la 1, f\ra$.
\end{lemma}

\begin{proof}
Consider $f\in \CTor_2(\mathbb{Z}_2^n, \mathbb{Z}_2^m)$.
Recall that $f$ may be regarded as a linear map $t:\mathbb{Z}_2^n\to \mathbb{Z}_2^m$ with a shift $(b_1, \cdots, b_m)$.  Consider the standard bases $\{e_i\}$ and $\{m_j\}$ of $\mathbb{Z}_2^n$ and $\mathbb{Z}_2^m$, respectively.  As $t$ is a linear map, for any $1\leq i \leq n$ there are unique coefficients $a_{i,j}\in\mathbb{Z}_2$ for all  $1\leq j \leq n$ such that:
$$f(e_i) = \sum_{i=1}^m a_{i,j}m_i$$
However, as $\mathbb{Z}_2^n$ is a vector space over $\mathbb{Z}_2$, the coefficients $a_{i,j}$ are either $0$ or $1$ so they determine for each $i$ a subset of the $m_i$.

Construct a circuit by starting with $1_n\otimes \hat{ b_1, \cdots, b_m}$ and condition on $a_{ij}$; apply  a $\cnot$ gate from the $i$th wire to the $i+j$th wire for all $1\leq i \leq n$ and $1\leq j \leq m$.  Call this new circuit $g$.
Given any $(c_1,\cdots, c_n) \in \mathbb{Z}_2^n$, by Lemma $\ref{lemma:behaveasexpected}$:

\[
|c_1,\cdots, c_n\ra g
%= \bigotimes_{j=1}^m \left| b_j+\sum_{m_j\in M_{i}}   \right\ra
= \bigotimes_{j=1}^m \left| b_j+\sum_{i=1}^n a_{i,j}c_i \right\ra
\]

Therefore, $\tilde H_0 (g)(c_1,\cdots, c_n) = f(c_1,\cdots, c_n)$ and thus $\tilde H_0 (g) = f$.

For example, consider the map $f\in\CTor_2(\mathbb{Z}_2^3, \mathbb{Z}_2^2)$ given by the affine transformation with a linear component $T$ and shift $S$ such that:

\[
T=
\begin{bmatrix}
1 & 0 & 1\\
1 & 0 & 0
\end{bmatrix}
\hspace*{1cm}
\text{and}
\hspace*{1cm}
S=
\begin{bmatrix}
0 \\
1 
\end{bmatrix}
\]

Then the corresponding circuit $g$ such that $\tilde H_0(g) = \la 1,f \ra$ is:

$$
\begin{tikzpicture}[baseline={([yshift=-.5ex]current bounding box.center)}]
\begin{pgfonlayer}{nodelayer}
\node [style=rn] (40) at (-1, 2) {};
\node [style=dot] (41) at (.5, 2) {};
\node [style=dot] (42) at (1, 2) {};
\node [style=rn] (43) at (1.5, 2) {};
\node [style=rn] (44) at (2, 2) {};
\node [style=rn] (30) at (-1, 1.5) {};
\node [style=rn] (31) at (.5, 1.5) {};
\node [style=rn] (32) at (1, 1.5) {};
\node [style=rn] (33) at (1.5, 1.5) {};
\node [style=rn] (34) at (2, 1.5) {};
\node [style=rn] (20) at (-1, 1) {};
\node [style=rn] (21) at (.5, 1) {};
\node [style=rn] (22) at (1, 1) {};
\node [style=dot] (23) at (1.5, 1) {};
\node [style=rn] (24) at (2, 1) {};
\node [style=zeroin] (10) at (-.5, .5) {};
\node [style=oplus] (11) at (.5, .5) {};
\node [style=rn] (12) at (1, .5) {};
\node [style=rn] (13) at (1.5, .5) {};
\node [style=rn] (14) at (2, .5) {};
\node [style=onein] (0) at (-.5, 0) {};
\node [style=rn] (1) at (.5, 0) {};
\node [style=oplus] (2) at (1, 0) {};
\node [style=oplus] (3) at (1.5, 0) {};
\node [style=rn] (4) at (2, 0) {};
\end{pgfonlayer}
\begin{pgfonlayer}{edgelayer}
\draw [style=simple] (40) to (44);
\draw [style=simple] (30) to (34);
\draw [style=simple] (20) to (24);
\draw [style=simple] (10) to (14);
\draw [style=simple] (0) to (4);
\draw [style=simple] (41) to (11);
\draw [style=simple] (42) to (2);
\draw [style=simple] (23) to (3);
\end{pgfonlayer}
\end{tikzpicture}
$$
\end{proof}

We are now ready to prove:

\begin{proposition}
\label{prop:full}
$\tilde H_0:\CNOT\to\ParIso(\CTor_2)^*$ is full.
\end{proposition}

\begin{proof}
Suppose $\xymatrix@1{A & A' \ar[l]_{f} \ar[r]^{g} & B}$ is a partial isomorphism in $\Par(\CTor_2)^{*}$.  Thus $f$ and $g$ are monics.  If $A'$ is empty we can simulate the map as $\tilde H_0(\Omega_{n,m})$ for some $n,m \in \N$.  On the otherhand if  $A'$ is non-empty
then there is a total map $r$ with $fr = 1_{A'}$ as the object $A'$ is injective (as it is injective as a $\mathbb{Z}_2$-vector space).  This means the 
total map $\xymatrix@1{A & A \ar@{=}[l] \ar[r]^{rg} & B}$ extends $(f,g)$ (so $(f,g) \leq (1_A,rg)$) as 
$$\xymatrix{ & A' \ar[dd]^f \ar[ld]_f \ar[rd]^g \\ A & & B \\ & A \ar@{=}[lu] \ar[ru]_{rg} }$$
But by Lemma \ref{simulating_total_maps} there is a map $k \in \CNOT$ with $\tilde H_0(k) = \la 1, rg \ra$.  By the fullness of $\tilde H_0$ on restriction idempotents there is an $e$ with $\tilde H_0(e) = (f,f)$ but then
$$\tilde H_0(ek) = \tilde H_0(e)\tilde H_0(k) = (f,f)\la 1_a,rg \ra= \la \bar{(f,g)}, (f,g) \ra.$$
Similarly we can implement  $\la \bar{(g,f)}, (g,f) \ra$ and therefore by Lemma \ref{lemma:fullCopy} we can implement $(f,g)$.

Thus $\tilde H_0$  is full.
\end{proof}

%%%%%%%%%%%%%%%%%%%%%%%%%%%%%%%%%%%%%
\subsection{\texorpdfstring{$\tilde H_0: \CNOT \to \ParIso(\CTor_2)^{*}$}{H0} is faithful}
%%%%%%%%%%%%%%%%%%%%%%%%%%%%%%%%%%%%%%

We reduce the faithfulness of $\tilde H_0$ to being faithful on restriction idempotents in two stages.

\begin{lemma}
\label{lemma:faithfulidempotents}
A restriction functor $F:\mathbb{X}\to \mathbb{Y}$ between inverse categories is faithful if and only if it reflects and is faithful on restriction idempotents.
\end{lemma}

To reflect restriction idempotents means that, whenever $h: A \to A$ is an endomorphism with $F(h)$ a restriction idempotent, then $h$ is a restriction idempotent itself.

\begin{proof}~

\begin{description}
\item[$\Rightarrow$]: Suppose $F$ is faithful then it is faithful on restriction idempotents. If $F(g) = \bar{F(g)}$, then $F(g)F(g) = F(gg) = F(g)$. So $g$ is an idempotent and thus a restriction idempotent (as all idempotents are restriction idempotents in an inverse category).

\item[$\Leftarrow$]: Suppose $F$ reflects and is faithful on restriction idempotents and that $F(f) = F(g)$ for $f,g$ which are parallel maps in $\X$. This means that \[ \bar{F(f)} = F(f)F(f)\cnv = F(f) F(g) \cnv = F(fg \cnv) \]
So $fg \cnv$ is a restriction idempotent as is $g \cnv f$. But \[ F(fg \cnv) = F(f) F(g) \cnv = F(f) F(f) \cnv = F(ff \cnv) \] So $fg \cnv = ff \cnv$. Thus $g \cnv$ is the partial inverse of $f$, and hence $g \cnv = f \cnv$.

\end{description}

\end{proof}

\iffalse
Suppose $F$ reflects and is faithful on restriction idempotents.

Consider maps $f,g \in \mathbb{X}(A,B)$ with $U(f)=U(g)$, then:
$$U(fg\cnv)=U(f)U(g\cnv) = U(f)U(g)\cnv = U(f)U(f)\cnv$$

However, $U(f)U(f)\cnv$ is a restriction idempotent and, as $F$ by assumption reflects restriction idempotents, $fg\cnv$ is a restriction idempotent.  Furthermore, 
as $U(f)U(f)\cnv = U(f)U(f\cnv) = U(ff\cnv) = U(\bar{f})$ and $F$ is faithful on idempotents $f g\cnv = \bar{f\cnv}$.  Similarly $f\cnv g= \bar{f\cnv}$ and as $F$ is faithful on idempotents $\bar{f\cnv} = \bar{g\cnv}$ but this 
means $f$ is the partial inverse of $g\cnv$ which means $f=g$.

Conversely, if $F$ is faithful it is certainly faithful on restriction idempotents.  Suppose, $g: A \to A$ and $F(g)$ is a restriction idempotent then  $F(g) = \bar{F(g)} = F(\bar{g})$ so, as $F$ is faithful, $g=\bar{g}$ as desired.
\end{proof}
\fi

Therefore, as $\tilde H_0$ is a restriction functor, it suffices to prove that $\tilde H_0$ reflects and is faithful on restriction idempotents.  However, we can do better for discrete inverse categories:

\begin{lemma}
\label{lemma:faithfullemma}
A restriction functor $F:\mathbb{X}\to \mathbb{Y}$ between discrete inverse categories which preserves the inverse product is faithful if and only if it is faithful on restriction idempotents.
\end{lemma}
\begin{proof}~

\begin{description}

\item[$\Rightarrow$]: If $F$ is faithful it is certainly faithful on restriction idempotents.

\item[$\Leftarrow$]: By Lemma \ref{lemma:faithfulidempotents}, it suffices to prove that $F$ reflects restriction idempotents. 

Suppose $F(f) = F( \bar{f} )$, then 

\[ F(\bar{f}) = F(f) \cap \bar{F(f)} = F(f) \cap F(\bar{f}) = F( f \cap \bar{f}) \]

Since $\bar{f}$ and $f \cap \bar{f}$ are restriction idempotents and $F$ is faithful on restriction idempotents, then $\bar{f} = f \cap \bar{f} \leq f$. But then $\bar{f} \leq f$ iff $\bar{f}f = \bar{f}$. So $f=\bar{f}$ as $\bar{f}f = f$.

\end{description}

\iffalse
Consider a functor $F:\mathbb{X}\to\mathbb{Y}$ between discrete inverse categories that preserves the tensor products, diagonal and codiagonal maps.

As for arbitrary maps $f$ and $g$, the meet of $f$ and $g$, $f \cap g := \Delta (f \otimes g) \nabla$ \cite[Prop. 4.3.6]{Giles}, it follows that $F$ also preserves meets.

Suppose that $F$ is faithful on restriction idempotents and there is some endomorphism $g$ in $\mathbb{X}$ such that  $F(g) = \bar{F(g)}$.  Then, by \cite[Lem. 3.4.2]{Giles},

$$F(\bar{g}) = \bar{F(g)} = F(g) \cap \bar{F(g)} = F(g) \cap F(\bar{g}) = F(g \cap \bar{g})$$

Therefore, as $\bar g$ and $g \cap \bar{g}$ are restriction idempotents, $\bar{g} = g \cap \bar{g}$.

Therefore, $\bar{g} = g \cap \bar{g} \leq g$

So $\bar{g} \leq g$, and thus, $\bar{g}g = \bar{g}=g$.  Therefore, $F$ reflects restriction idempotents.
\fi

\end{proof}

Therefore, by Lemma \ref{lemma:faithfullemma} and Lemma \ref{lemma:PreservesInverseProducts}, it suffices to show that $\tilde H_0$ is faithful on restriction idempotents to prove that it is faithful.  However, we already 
have proven that $\tilde H_0$ is faithful on idempotents in Theorem \ref{Theorem: restriction faithful} so we have:

\begin{proposition}
\label{lemma:Gfaithful}
$\tilde H_0:\CNOT\to \ParIso(\CTor_2)^*$ is faithful.
\end{proposition}

Finally this gives the main theorem:

\begin{theorem}  \label{thm:CNOTEquiv}
There is an equivalence of categories between $\CNOT$ and $\ParIso(\CTor)^{*}$.
\end{theorem}
\begin{proof}
The equivalence functor $\tilde H_0:\CNOT\to\ParIso(\CTor_2)^*$ is full, faithful, and essentially surjective by Propositions \ref{prop:full}, \ref{lemma:Gfaithful} and \ref{prop:essentiallysurjective}, respectively.
\end{proof}
\end{document}